\documentclass[11pt]{article}
\usepackage{amsmath,amsthm,amssymb,amsfonts}
\usepackage[normalem]{ulem}
\usepackage{comment}
\usepackage[comma]{natbib}

\usepackage[letterpaper,top=2cm,bottom=2cm,left=3cm,right=3cm,marginparwidth=1.75cm]{geometry}
\usepackage{tikz}
 \usetikzlibrary{positioning,arrows.meta}
\usepackage{setspace}
\usepackage{graphicx}
\usepackage{lipsum}
\usepackage[font=normalsize]{subcaption}
\usepackage{placeins}
\usepackage[hidelinks,
            bookmarks,
            pagebackref=true,
            linktocpage=true]{hyperref}

\usepackage{etoolbox}
\usepackage{epigraph}

\renewenvironment{abstract}{%
  \begin{center}%
    \bfseries\large
    \abstractname
  \end{center}%
  \quotation
}{%
  \endquotation
}

\newtheorem{definition}{Definition}

\newtheorem{proposition}{Proposition}
\newtheorem{corollary}{Corollary}
\newtheorem{remark}{Remark}
\newtheorem{lemma}{Lemma}
\newtheorem{lem}{Lemma}[section]

\newtheorem{prop}{Proposition}[section]

\title{\huge Information Greenhouse:\\ \LARGE Optimal Persuasion for Medical Test-Avoiders}
\author{Zhuo Chen\thanks{Center of Economic Research, Shandong University. Email: zhuochen@sdu.edu.cn. I would like to thank Russell Golman, Jaimie Lien and Jie Zheng as well as participants in GAMES 2024 (Beijing), the 14th Conference on Economic Design (Essex), and the 2025 World Congress of the Econometric Society (Seoul) for the useful comments. I gratefully acknowledge financial support from the Natural Science Foundation of Shandong Province (Grant No. ZR2025MS1138).}}
\date{\today}

\begin{document}

\maketitle
\normalsize

\begin{abstract}
\setstretch{1.2} 
Patients often avoid medical tests because the information they provide, although medically useful, is psychologically painful.
This paper studies optimal communication between a doctor and an information-avoidant patient who can refuse testing and treatment.
I characterize when optimal communication creates an \textit{information greenhouse}, a commitment to reward participation with comforting information about the untreated prognosis.
When testing is voluntary and the patient is unwilling to be tested under extreme pessimism, an information greenhouse is optimal and takes the form of committed comfort, which provides reassuring information after the test.
When the patient can reject the consultation at the outset and the patient’s prior belief about the untreated prognosis is intermediate, an information greenhouse is optimal and takes the form of precautionary comfort, which provides reassuring information before the test.
In all other cases in which the patient can be persuaded, warning-based policies that trigger pessimism prevail.


\end{abstract}

\newpage
\begin{flushright}
    \textit{``To cure sometimes, to relieve often, to comfort always.''}\\
    --- On the statue of Dr. Edward Livingston Trudeau
\end{flushright}

\section{Introduction}

Sometimes the hardest medical decision is not whether to undergo treatment, but whether to find out if treatment is needed.
Although a patient who notices suspicious symptoms may fully understand the medical value of testing, he may still delay, or even refuse to be tested.
The problem is not that the test lacks instrumental value, nor primarily that undergoing the test is costly in monetary or physical terms.
Rather, the very information that helps the patient choose the right action may also destroy the temporary psychological shelter provided by uncertainty.
This tendency to deliberately avoid useful information for psychological or emotional reasons is commonly referred to as \textit{information avoidance} \citep{golman2017information,sweeny2010information}.
For instance, although early detection is widely recognized as important for the prognosis and treatment of breast cancer, evidence suggests that information avoidance is linked to delays in, or resistance to, routine mammography \citep[e.g.][]{caplan1995patient,meechan2002delay,melnyk2012avoiding,zanella2016experiencing}.
Such delays can result in significant health and economic losses.

This paper studies how a doctor should communicate with a patient who needs information to choose the right action but may avoid testing because the information itself is psychologically painful. 
The doctor cannot force the patient to be tested, nor does she rely on subsidies, penalties, or behavioral nudges to alter his decision.
Her only instrument is the design of the information environment the patient faces during the medical consultation. 
With dynamic commitment, the doctor shapes the way information about the untreated prognosis will be disclosed both before and after the patient's testing decision.
Such communication matters for two reasons. 
First, it affects the patient's subsequent treatment decision, because whether he accepts treatment depends on his belief about the untreated prognosis. 
Second, it affects his willingness to be tested in the first place, because the patient anticipates not only what he may learn, but also how learning it will feel. 
Thus, medical communication is not merely about transmitting accurate information, but also about making the patient willing to receive it.

I formalize this problem as a two-stage medical communication game in which an information-avoidant patient first decides whether to undergo the test and, after an unfavorable result, decides whether to accept treatment.
 The sender is a doctor who maximizes the patient's physical health outcome, while the receiver is a patient with anticipatory utility who cares not only about final health but also about the emotional consequences of beliefs formed before uncertainty is resolved \citep{loewenstein1987anticipation,caplin2001psychological}. 
 Treatment is useful only if he is sick, and it raises his probability of a healthy final outcome at a physical cost. 
 {The medical test fully reveals whether the patient is sick, and the test result is publicly observable.
 The doctor's disclosure instead shapes the patient's belief about his prognosis without treatment.
 We consider this setting because, in medical practice, doctors are often required to report standardized and verifiable diagnostic evidence truthfully but may still influence patients’ decisions by communicating related medical information, such as the prognosis without treatment.}

Figure \ref{fig:intro_tree} illustrates the structure of the consultation process.
The doctor has two opportunities to disclose information.
She may disclose information before the testing decision, which I call \textit{ex ante disclosure}.
If the patient is tested and found sick, she may provide additional information before the treatment decision, which I call \textit{interim disclosure}.
Interim disclosure plays a dual role.
Looking forward, it shapes the patient's belief about the untreated prognosis and therefore affects the treatment decision.
Looking backward, the patient anticipates this disclosure when deciding whether to be tested, and thus a promise about future communication can influence an earlier decision to enter the information process.

\begin{figure}
    \centering
\begin{tikzpicture}[scale=2, font=\small]
    \tikzset{
  solid node/.style={circle,draw,inner sep=1.5,fill=black},
  hollow node/.style={circle,draw,inner sep=1.5},
  level 1/.style={level distance=5mm,sibling distance=1.5cm},
  level 2/.style={level distance=5mm,sibling distance=1.2cm},
  level 3/.style={level distance=5mm,sibling distance=1.2cm},
  level 4/.style={level distance=5mm, sibling distance=1cm},
  level 5/.style={level distance=6mm, sibling distance=1cm}
}
\node(0)[solid node,label=right:{{Doctor}}]{}
    child {
        node[solid node,label=right:{{Patient}}]{}
        child{
            node[solid node]{}
            child{
                node[hollow node]{}
                edge from parent [draw=black] node[left,yshift=5]{Healthy}
            }
            child{
                node[solid node,label=left:{{Doctor}}]{}
                child{
                    node[solid node,label=left:{{Patient}}]{}
                    child{
                        node[hollow node]{}
                        edge from parent [draw=black] node[left]{Treat}
                    }
                    child{
                        node[hollow node,black]{}
                        edge from parent [draw=black] node[right]{No treatment}
                    }
                    edge from parent [draw=black] node[right]{\textbf{Interim disclosure}}
                }
                edge from parent [draw=black] node[right,yshift=5]{Sick}
            }
            edge from parent [draw=black] node[left,yshift=6]{Test}
        }
        child{
            node[hollow node,black]{}
            edge from parent [draw=black] node[right,yshift=6]{No test}
        }
        edge from parent [draw=black] node[left]{\textbf{Ex ante disclosure}}
    };
\end{tikzpicture}
\caption{Timing of Communication and Decisions}
\vspace{-.2cm}
\label{fig:intro_tree}
\end{figure}

This dual role creates the central tradeoff of the paper.
A warning-oriented disclosure, which sometimes makes a favorable prognosis appear unfavorable but never the reverse, can increase the probability that the patient receives sufficiently bad news to induce treatment.
However, the same warning may deter testing because the patient anticipates the psychological pain of confronting bad news.
By contrast, a comfort-oriented disclosure, which sometimes makes an unfavorable prognosis appear favorable but never the reverse, can make testing more acceptable by preserving favorable uncertainty, but it may leave the patient too optimistic about the prognosis and thereby weaken his willingness to accept treatment.
The doctor must balance the patient's willingness to be tested against his willingness to be treated.

The central question of this paper is when optimal persuasion creates an \textit{information greenhouse}, the policy under which the doctor commits to reward a desirable action with comfort-oriented disclosure that makes the prospect of a bad test result less distressing.
In the baseline model, the answer depends on whether extreme pessimism about the untreated prognosis leads the patient to take the test or avoid it.
When such pessimism frightens the patient into taking the test, an information greenhouse is unnecessary.
The optimal policy is \textit{warning-in-advance}, under which the doctor provides sufficiently bad news before the testing decision to induce the patient to take the test, and the same pessimism ensures that, if found sick, he accepts treatment without further information.
When such pessimism would instead make the patient bury his head in the sand and avoid testing, an information greenhouse can induce him to take the test, and the optimal policy is \textit{committed comfort}, under which the doctor commits before the testing decision to provide comfort-oriented disclosure after the test reveals that the patient is sick.
Of these two policies, committed comfort can be optimal only when the patient's cost of testing is psychological rather than physical.

I also consider voluntary consultation, in which the patient may refuse the consultation before receiving any disclosure.
An information greenhouse is optimal when his belief about the untreated prognosis is intermediate, and it takes the form of \textit{precautionary comfort}.
Under this policy, the doctor commits before the consultation decision to provide comfort-oriented disclosure after the patient enters the consultation but before the testing decision, and provides no further information after the test.
When the patient is more pessimistic, warning-in-advance is optimal, whereas when he is more optimistic, no information policy can persuade him to enter the consultation and take the test.

These results also clarify why an information greenhouse requires both dynamic disclosure and the doctor’s ability to commit in advance to what she will disclose later.
By placing disclosure after a desired action, the doctor can turn the prospect of future information into a carrot; with anticipatory utility, the effective reward is not transparency, but comforting information that makes bad news easier to face.
Meanwhile, the ability to commit is essential because the promised comfort is not sequentially optimal: once the patient has taken the action being rewarded, the doctor no longer has an incentive to provide it and would instead choose disclosure that more strongly motivates treatment.

 \subsection{Literature}
 This paper first relates to the literature on anticipatory utility and information avoidance. 
 Anticipatory emotions are introduced into economics by \cite{akerlof1982economic} and \cite{loewenstein1987anticipation}, and are extended to environments with risk and uncertainty by \cite{wu1999anxiety}, \cite{caplin2001psychological} and \cite{epstein2008living}. 
 A large body of work studies how the tradeoff between the instrumental value of information and the emotional cost of learning affects information-acquisition decisions, especially in health-related contexts \citep{caplin2003aids,koszegi2003health,koszegi2006ego}. 
 My paper builds on this insight but shifts the focus from the individual's information-acquisition decision to the sender's problem of designing the information environment. 
 The doctor must choose not only whether information should be disclosed, but also when warning or comfort should be used to make both testing and treatment acceptable.

 The paper is also related to work on information transmission to psychological receivers. 
 \cite{caplin2004supply} and \cite{koszegi2006emotional} show that anticipatory feelings can distort communication even when the sender is well intentioned, because attempts to protect the receiver's emotions may affect credibility. 
 \cite{schweizer2018optimal} study optimal test design for an anxious patient, and \cite{lipnowski2018disclosure} analyze persuasion with non-standard information preferences. 
 My model differs in two respects. 
 First, the doctor and the patient are not aligned in the relevant objective: the doctor maximizes physical health, while the patient also cares about anticipatory feelings. 
 Second, the patient controls whether to acquire the test result, while the doctor controls the background information about the untreated prognosis.
 This separation between voluntary testing and designed disclosure generates the central tradeoff between making testing acceptable and making treatment attractive.

 The dynamic structure connects the paper to the literature on dynamic persuasion and information design with forward-looking receivers \citep{ely2020moving,habibi2020motivation,chen2026accelerator}. 
 This literature shows that a sender's commitment to future information can affect earlier actions: future disclosure can guide later decisions and provide incentives for earlier behavior. 
 In those settings, however, the receiver typically values future information positively. 
 In my model, the promised future information may be psychologically costly. 
 As a result, interim disclosure is not simply a reward for taking the test; it may comfort the patient into testing, warn him into treatment, or deter him from entering the information environment. 
 This changes the role of dynamic commitment and creates the warning-versus-comfort tradeoff characterized in the paper.
 
 A related set of papers studies persuasion problems in which information design interacts with participation or incentive constraints. 
 The most closely related paper is \citet{rosar2017test}, who studies test design under voluntary participation. 
 In Rosar's model, the sender designs the test itself and the agent's reluctance is driven by the payoff consequences of evaluation.
 In my model, however, the medical test is fixed and fully revealing; the doctor instead chooses how to disclose background information about the untreated prognosis, and the patient's reluctance comes from the anticipated emotional cost of learning bad news.
 The paper is also related to \citet{boleslavsky2020bayesian}, where information design interacts with moral-hazard incentives. 
 In my paper, there is no moral hazard: the relevant earlier action is the patient's voluntary decision to enter the test or the consultation, and the same patient later decides whether to accept treatment.

 \bigskip
 The remainder of the paper is organized as follows. 
 Section~\ref{section.model} presents the model. 
 Section~\ref{section.voluntary_testing} characterizes optimal communication under voluntary testing. 
 Section~\ref{section.exante} studies voluntary consultation, where the patient can refuse the doctor's communication at the outset. 
 Section~\ref{section.dynamic_role} studies the role of dynamic disclosure and characterizes when the full-commitment outcome remains implementable without dynamic commitment.
 Section~\ref{section.diagnostic_disclosure} examines optimal disclosure when the diagnostic result itself is designable.
 Section~\ref{section.conclusion} concludes.

\section{The Model}
\label{section.model}
This section outlines the formal model.

\subsection{Economic Environment}
 A patient (``he'') suspects he might have a life-threatening disease and seeks medical advice from a doctor (``she''). 
 The patient's \textit{initial health status} is either \textit{sick} ($\theta=0$) or \textit{not sick} ($\theta=1$), while his \textit{final health status} is either \textit{healthy} or \textit{unhealthy}. 
 I use \textit{health probability} for the probability that the patient is ultimately healthy. 
  
 \textbf{The health function.}
 When $\theta=1$, the health probability equals $1$.
 When $\theta=0$, the health probability is calculated as
  \begin{align*}
    \bar{p}e+\underline{p}(1-e).
  \end{align*}
 Here, $e\in\{0,1\}$ denotes the patient's treatment action. 
 The patient receives a medical intervention when $e=1$ and forgoes treatment when $e=0$. 
 A patient who is not sick needs no treatment. 
 If he is sick and forgoes treatment, his health probability is $\underline p$. 
 Treatment, such as taking medication or undergoing surgery, raises this probability to $\bar p>\underline p$.
 The treatment action incurs cost $ce$, where $c$ may represent monetary expense, physical pain, or both. 

 \textbf{The medical test.}
 Neither party observes $\theta$ at first, and the only way to learn it is through the medical test. 
 The test is costless and fully revealing, and hence after testing both parties observe whether the patient is sick.
 Let $\alpha\in(0,1)$ be the common prior that $\theta=1$, and let $a=1$ and $a=0$ denote the patient's decisions to accept and reject the medical test, respectively.
 As in many medical settings, the patient cannot undergo treatment without first being tested.

 \textbf{Information disclosure.}
 The doctor aims to maximize the patient's health probability, and her only instrument is information disclosure.
 At the beginning of the interaction, the doctor can commit to an information policy, and this commitment is credible to the patient.

 The doctor's disclosure concerns the untreated prospect $\underline p$, the probability that a sick patient remains healthy without treatment.\footnote{Disclosure about $\bar p$ produces an analogous problem. See Appendix~\ref{appendix.barp_disclosure}.}
 I assume that $\underline p$ takes values in the binary set $\{\underline p_H,\underline p_L\}$, where $\bar p>\underline p_H>\underline p_L$.
 This binary specification is the minimal environment that separates uncertainty about the patient's health status from uncertainty about the severity of remaining untreated.
 Besides the initially unknown health status $\theta$, the only uncertainty relevant for disclosure is the realized value of $\underline p$; $\bar p$ and $c$ are common knowledge. Neither party initially observes the realization of $\underline p$, which is independent of $\theta$. A committed disclosure device conditions its message distribution on the realized $\underline p$, and thus the doctor need not observe it herself.

 \begin{remark}[Two dimensions of medical uncertainty]
 The diagnostic result and information about the untreated prospect are distinct information objects.
 The patient learns from the test whether he is sick and from the doctor what happens if he is sick and remains untreated.
 The model therefore has a two-dimensional uncertainty structure.
 To induce the patient to acquire information about one dimension, the doctor uses information about the other dimension as a persuasion instrument.

 This formulation captures settings in which diagnostic results are standardized and verifiable, while doctors retain discretion over communicating prognosis and the consequences of non-treatment.
 Keeping the diagnosis fully revealing is also consistent with norms of truthful medical communication and informed consent.
 Section~\ref{section.diagnostic_disclosure} studies optimal disclosure of the diagnostic result, which concerns $\theta$ rather than $\underline p$.
 \end{remark}

 The doctor has two time windows in which to disclose information about $\underline p$.
 The \textit{ex ante disclosure} occurs before the patient chooses whether to take the test.
 The \textit{interim disclosure} occurs after either an unfavorable test result or a refusal to test.
 Following a favorable result, $\underline p$ is irrelevant and no further disclosure is needed.
 The patient begins with a prior belief $\mu_0\in(0,1)$ that $\underline p=\underline p_H$.
 The ex ante disclosure updates the prior $\mu_0$ to an interim belief $\mu_1$, and the interim disclosure subsequently updates $\mu_1$ to an ex post belief $\mu_2$.

 Denote the ex ante disclosure by
 \begin{align*}
     \tau_0\in\mathcal{R}(\mu_0)
     \equiv
     \left\{
     \tau\in\Delta([0,1])
     \,\middle|\,
     \int\mu\,\tau(d\mu)=\mu_0
     \right\}.
 \end{align*}
 The doctor observes the patient's testing decision and, if the patient is tested, the diagnostic result.
 She can therefore condition the interim disclosure on this observation.
 For every $\mu_1\in\operatorname{Supp}(\tau_0)$, let $\tau_1^1(\cdot\mid\mu_1)\in\mathcal{R}(\mu_1)$ denote the disclosure used after the patient is tested and found sick, and let $\tau_1^0(\cdot\mid\mu_1)\in\mathcal{R}(\mu_1)$ denote the disclosure used after a refusal to test.
 An information policy is therefore a triple $\left\langle\tau_0,\tau_1^1(\cdot\mid\cdot),\tau_1^0(\cdot\mid\cdot)\right\rangle$, which can be interpreted as the standard operating procedure of the medical consultation.
  
 \textbf{The timeline.} There are two dates: $1$ and $2$. 
 The game proceeds as follows. 
 At date $1$, the doctor commits to an information policy and implements its ex ante component.
 The patient then decides whether to take the medical test.
 After observing that decision and, if testing occurs, its result, the doctor implements the prescribed interim disclosure.
 At date $2$, the patient first decides whether to accept the treatment and pays the physical cost $ce$.  
 Then the final health status is realized and experienced.

 \textbf{Equilibrium concept.} The equilibrium concept is perfect Bayesian equilibrium, which requires sequential rationality at every decision node and belief updating according to Bayes' rule whenever possible.

\subsection{Preferences}
 I write $e \in \{0,1\}$ for the patient's date-2 treatment choice and $e^{*}$ for the date-1 expectation of that choice. 
 Formally, at the end of date 1, the health probability is given by
  \begin{align*}
     \eta(a,e^*,\mu_1)=\left\{
        \begin{array}{cc}
            \alpha+(1-\alpha)\underline{p}(\mu_1) & a=0\\
            \bar{p}e^*+\underline{p}(\mu_1)(1-e^*) & a=1\text{ and }\theta=0\\
            1 & a=1\text{ and }\theta=1
        \end{array}
     \right.,
 \end{align*}
 where $e^*$ is the patient's (rational) expectation of his action at date 2 and $\underline{p}(\mu)\equiv\mu \underline{p}_H+(1-\mu)\underline{p}_L$. At date 2, I normalize the physical prize of being eventually healthy to $1$ and assume that there are no anticipatory emotions within a stage. Thus, the patient's date-2 payoff is $u_2(a,e,\mu_1)=\eta(a,e,\mu_1)-ce$.
 Let
 \begin{align*}
    v_2(a,\mu_1)=\max_{e\in\{0,1\}}u_2(a,e,\mu_1)
 \end{align*}
 denote his indirect utility at date 2. 

 I model the patient's preferences using the psychological expected utility framework of \citet{caplin2001psychological}. Under this framework, the patient's choices can be understood as a subgame-perfect Nash equilibrium in an \textit{intrapersonal game} between two successive selves.

 The key assumption of the psychological expected utility model is the dynamic inconsistency arising from the temporal distance between decision making and uncertainty resolution.
 To capture this dynamic inconsistency, let the patient transform his date-2 indirect utility through \(\phi\) when evaluating it at date 1, such that his date-1 payoff is
 \begin{align*}
     u_1(a,\mu_1)=\phi\circ v_2(a,\mu_1).
 \end{align*}
 Function $\phi$ measures the distortion of the patient's preference by the anticipatory emotions at date 1; the model reduces to the standard one when \(\phi\) is linear.
 Assume that $\phi$ is continuous, smooth, and strictly increasing.
 Throughout the paper, I normalize $\phi(1)=1$ without loss of generality, such that the anticipatory utility from certainty of a healthy outcome equals $1$.
 At every date-1 decision node, the patient maximizes the expected value of $u_1(a,\mu_1)$ given his current belief and a correct anticipation of his future behavior.

 The curvature of function $\phi$ determines the patient's information attitude: The patient is information-avoidant (resp. loving, neutral) if and only if $\phi$ is concave (resp. convex, linear) in probabilities.\footnote{The same behavioral tendency is sometimes referred to as \textit{information aversion}. I use the term \textit{information avoidance} because \textit{information aversion} is sometimes used more narrowly to describe a preference for one-shot resolution of uncertainty; see, for example, \cite{andries2020information}.}\footnote{The classical interpretation of this result is based on the relationship between risk attitude and time. Suppose the probability of being eventually healthy is $\eta$. Then if the uncertainty is resolved completely at date 1, the utility is $\eta\phi(1)$, while if the uncertainty is resolved completely at date 2, the utility is $\phi(\eta)$. The former is larger if and only if ``immediate risk'' is more desirable than its corresponding ``remote risk''.}
 To model information avoidance, I assume $\phi$ is strictly concave throughout this paper unless explicitly stated.{\footnote{Appendix~\ref{section.general} relaxes global concavity and studies more general information attitudes, including an inverse S-shaped $\phi$. The warning-versus-comfort logic survives, although the optimal signal may require richer posterior support.}}
 The concavity of $\phi$ can be interpreted as a result of \textit{diminishing marginal utility of hope}, since hope preservation is a common psychological interpretation of information avoidance \citep{chew1994hope,ahlbrecht1996resolution,kocher2014let,chen2022preference,masatlioglu2023intrinsic}.
 If $\phi(v)-v$ captures the emotional value of hope associated with date-2 payoff \(v\), the concavity of \(\phi\) means that the marginal emotional benefit of a more hopeful belief diminishes.

 The doctor's objective function is to maximize $\mathbb{E}[\eta]$, the expected final health probability, evaluated under the patient's equilibrium behavior.\footnote{Following the standard convention in Bayesian persuasion, I assume that an indifferent patient chooses an action that maximizes the doctor's payoff \citep{kamenica2011bayesian}.} 
 Two divergences from the patient's objective stand out. 
 First, the doctor does not internalize the treatment cost $c$ that the patient bears whenever $e=1$, and hence she would prefer treatment to be accepted more readily than the patient does. 
 Second, the doctor does not internalize the psychological prize embedded in $\phi$: she evaluates outcomes in raw probabilities, while the patient's date-1 valuation is curved by anticipatory emotions, which is precisely why information avoidance arises in the first place. 


\section{Optimal Persuasion under Voluntary Testing}
\label{section.voluntary_testing}

This section studies voluntary testing.
In this case, the patient has to enter the consultation and hear the doctor's initial disclosure, but he may still refuse the medical test.
The doctor's problem is to design persuasion that balances two objectives, namely inducing the patient to take the test and, conditional on a bad test result, inducing him to accept treatment.

\subsection{Incentive Constraints}
\label{subsection.baseline}

The doctor's persuasion problem has two incentive constraints.
The treatment constraint is relevant after the patient is tested and found sick, and the patient accepts treatment only if he believes that remaining untreated is sufficiently harmful.
The testing constraint is relevant before the test, and the patient accepts the test only if the information environment following a possible bad result is psychologically acceptable.

The treatment constraint is governed by the patient's posterior belief about the untreated prospect.
Because the untreated prospect is binary, I use $\mu$ to denote the patient's belief that the untreated prospect is the favorable one $\underline{p}_H$, and when formed before the testing decision, this belief is written as $\mu_1$.

If the patient takes the medical test and the result is $\theta=0$, he accepts treatment if and only if
\begin{equation*}
    \bar{p}-c\geq\underline{p}(\mu)\equiv\mu\underline{p}_H+(1-\mu)\underline{p}_L,
\end{equation*}
which can be rearranged to $\mu\leq\mu_e\equiv \underline{p}^{-1}\left(\bar{p}-c\right)$.
Thus a higher $\mu$ makes the patient more optimistic about the untreated prospect, and treatment is accepted if and only if $\mu\leq\mu_e$.
Assume $\mu_e\in(0,1)$ to avoid trivial results.

For any diagnostic belief $\alpha$ and untreated-prospect belief $\mu$, let $U(\alpha,\mu)$ denote the patient's payoff after a diagnostic message and before treatment; that is,
\begin{align*}
U(\alpha,\mu)&\equiv
\phi\!\left(\max\left\{\alpha+(1-\alpha)\bar p-c,\,
\alpha+(1-\alpha)\underline p(\mu)\right\}\right).
\end{align*}
Also, before the testing decision, $U_1(\alpha,\mu)$ denotes his payoff from taking the test. Thus,
\begin{align*}
U_1(\alpha,\mu)&\equiv
\alpha+(1-\alpha)\phi\!\left(\max\left\{\bar p-c,\underline p(\mu)\right\}\right)\\&=\alpha U(1,\mu)+(1-\alpha)U(0,\mu).
\end{align*}
The payoff from refusing the test is $U_0(\alpha,\mu)\equiv
\phi\!\left(\alpha+(1-\alpha)\underline p(\mu)\right)$.

The testing constraint compares the expected value of taking the test with the value of refusing it.
If the patient is tested, interim disclosure affects his payoff only when the test reveals that he is sick.
If the patient refuses, the outside option depends on the disclosure $\tau_1^0$ promised upon refusal.
For any interim belief $\mu_1$ and any pair of signals $(\tau_1^1,\tau_1^0)\in \mathcal{R}(\mu_1)\times \mathcal{R}(\mu_1)$, the patient accepts the test if and only if
\begin{equation}
\label{eq:primitive_testing_constraint}
\begin{aligned}
     \int_0^{1}U_1(\alpha,\mu)\tau_1^1(d\mu)=&\alpha \int_0^1U(1,\mu)\tau_1^1(d\mu)\\&+(1-\alpha)\int_0^1U(0,\mu)\tau_1^1(d\mu)\geq
     \int_0^1U_0(\alpha,\mu)\tau_1^0(d\mu).
\end{aligned}
\end{equation}

\begin{remark}[Diagnostic fear in the testing constraint]
Equation~\eqref{eq:primitive_testing_constraint} makes the two sources of uncertainty explicit.
On the one hand, fixing the interim disclosures $\tau_1^1$ and $\tau_1^0$, since $U_0$ is concave in its first argument,
\begin{equation*}
U_0(\alpha,\mu)
\geq
\alpha U_0(1,\mu)+(1-\alpha)U_0(0,\mu).
\end{equation*}
The gap captures the patient's disutility from learning the test result.
Thus, condition~\eqref{eq:primitive_testing_constraint} incorporates medical test avoidance.
On the other hand, the doctor can affect the testing incentives through a carefully designed disclosure about $\underline p$.
Unlike the receiver's convex value function in classical persuasion, $U_1$ and $U_0$ are not convex in their second argument because of the anticipatory emotions, and this is the defining feature of the present disclosure problem.
\end{remark}

Because the doctor's disclosure changes only $\mu$ and leaves the prior $\alpha$ unchanged, I use the one-dimensional reductions $V(\mu)\equiv U_1(\alpha,\mu)$ and $V_0(\mu)\equiv U_0(\alpha,\mu)$ to simplify the notation.
For $\mu\leq\mu_e$, $V(\mu)$ is constant because the patient accepts treatment after an unfavorable diagnosis.\footnote{At low values of \(\mu\), the patient believes that he is unlikely to be healthy eventually without treatment if he is sick. He therefore chooses treatment, making his payoff independent of \(\mu\). Note that this region does not imply that the patient believes he is likely to be sick. Holding $\mu$ fixed, that case corresponds to a low $\alpha$ in $U_0(\cdot,\mu)$, which is concave in its first argument.}
The testing constraint in \eqref{eq:primitive_testing_constraint} can therefore be written as
\begin{equation}
\label{eq:testing_constraint}
\begin{aligned}
     \int_0^{1}V(\mu)\tau_1^1(d\mu)
     \geq
     \int_0^1V_0(\mu)\tau_1^0(d\mu).
\end{aligned}
\end{equation}

When the patient refuses the test, the continuation disclosure affects only the value of refusing.
Since $V_0(\mu)$ is concave, full disclosure minimizes the expected refusal value.
Thus, for the purpose of the testing constraint, the relevant outside option is the full-disclosure value
\begin{equation*}
     \bar{V}_0(\mu)\equiv
     \mu\cdot\phi(\alpha+(1-\alpha)\underline{p}_H)
     +(1-\mu)\phi(\alpha+(1-\alpha)\underline{p}_L),
\end{equation*}
which is linear in $\mu$.
Substituting full disclosure for $\tau_1^0$, condition \eqref{eq:testing_constraint} becomes
\begin{equation}
\label{eq:simplified_testing_constraint}
\begin{aligned}
     \int_0^{1}V(\mu)\tau_1^1(d\mu)\geq\bar{V}_0(\mu_1).
\end{aligned}
\end{equation}

\subsection{Main Result: Warning-in-Advance and Committed Comfort}

Fix a prior belief $\mu$. If a signal $\tau\in\mathcal{R}(\mu)$ induces at most two posterior beliefs, Bayesian plausibility allows the signal to be identified with its support. In particular, when the support of $\tau$ is binary, I write $\tau=\{x,y\}$, where $x>\mu>y$ are the upper and lower posterior beliefs.

\begin{definition}[Warning-oriented and comfort-oriented signals]
\label{def:signal_orientation}
An informative binary signal $\{x,y\}$ is 
\textit{warning-oriented} if $x=1$, and is \textit{comfort-oriented} if $y=0$.
\end{definition}

Under a warning-oriented signal, the favorable state is sometimes pooled with the unfavorable state, and thus the favorable state can generate a noisy unfavorable message. Relative to full disclosure, this makes the patient more likely to receive news that makes him pessimistic about the untreated prospect, and hence the signal is termed \textit{warning-oriented}. Conversely, under a comfort-oriented signal, the unfavorable state is sometimes pooled with the favorable state, and thus the unfavorable state can generate a noisy favorable message. Relative to full disclosure, this makes the patient more likely to receive news that leaves him optimistic about the untreated prospect, and hence the signal is termed \textit{comfort-oriented}.

\begin{definition}[Warning-in-advance]
\label{def:warning_in_advance}
An information policy is \textit{warning-in-advance} if there exists $\mu\in[0,1]$ such that the ex ante disclosure is $\tau_0=\{\mu_0\}$ when $\mu_0\leq\mu$ and $\tau_0=\{1,\mu\}$ when $\mu_0>\mu$, while the interim disclosure is no disclosure.
\end{definition}

Appendix~\ref{appendix.physical_cost_benchmark} shows that, when the psychological information cost induced by anticipatory emotions is replaced by a constant physical test cost, warning-in-advance defined above is optimal whenever testing can be induced. 
This benchmark therefore makes warning-in-advance a natural candidate in the present model.
In words, if the patient is sufficiently pessimistic at the outset, the doctor discloses nothing. 
If the patient is relatively optimistic, she uses a warning-oriented signal to make him pessimistic. 
No further information is disclosed after testing.
The warning-oriented signal therefore provides enough information before testing to support both decisions.

\begin{definition}[Committed comfort]
\label{def:committed_comfort}
An information policy is \textit{committed comfort} if there exist thresholds $\mu^L,\mu^H\in[0,1]$, with $\mu^L\leq\mu^H$, such that the ex ante disclosure is
\[
\tau_0=
\begin{cases}
\{0,\mu^L\}, & \mu_0<\mu^L,\\
\{\mu_0\}, & \mu_0\in[\mu^L,\mu^H],\\
\{\mu^H,1\}, & \mu_0>\mu^H,
\end{cases}
\]
while the interim disclosure is comfort-oriented for every $\mu_1\in[\mu^L,\mu^H]$.
\end{definition}

Under committed comfort, when the patient is highly pessimistic, the doctor reveals the bad state with some probability, such that no news makes him more optimistic. 
Conversely, when the patient is highly optimistic, she reveals the good state with some probability, such that no news makes him more pessimistic. 
In this way, the doctor moves the interim belief toward an intermediate interval as far as possible. 
For every belief in this interval, she commits to a comfort-oriented interim signal after a sick result, which provides comfort by pooling the bad state into the good-state message with positive probability. 
Therefore, committed comfort creates an information greenhouse.

\begin{proposition}[Optimal persuasion under voluntary testing]
\label{prop:voluntary_testing}
The optimal information policy has the following form.
\begin{itemize}
      \item If $V(0)\geq\bar V_0(0)$, an optimal policy is warning-in-advance.
      \item If $V(0)<\bar V_0(0)$ and some information policy induces testing, an optimal policy is committed comfort.
\end{itemize}
\end{proposition}
\begin{proof}
     See Appendix \ref{proof-exante}.
\end{proof}
The comparison between $V(0)$ and $\bar V_0(0)$ describes how the patient responds when the untreated prospect is at its worst. If the patient accepts the test when $\mu_1=0$ ($V(0)\geq\bar V_0(0)$), the fear generated by the worst untreated prospect induces him to take the test. I call such a patient \textit{fear-reactive}. If the patient rejects the test when $\mu_1=0$ ($V(0)<\bar V_0(0)$), he responds to the same fear with inaction rather than testing. I call such a patient \textit{fear-avoidant}.
Here, the patient is fear-reactive if and only if
\begin{equation}
\label{eq:fear_reactive}
 c\leq\bar{p}-\phi^{-1}\left(\frac{\phi(\alpha+(1-\alpha)\underline{p}_L)-\alpha}{1-\alpha}\right).
\end{equation}

For a fear-reactive patient, a sufficiently pessimistic posterior induces both testing and treatment, and warning-in-advance is therefore optimal since it maximizes the probability of inducing such a posterior conditional on sickness.
When the patient is fear-avoidant, however, inducing such a low posterior instead frightens him into inaction, requiring the doctor to adopt a different policy. 
Roughly speaking, the committed comfort policy uses a comfort-oriented signal to protect the patient’s anticipatory emotional experience. 
The role of ex ante disclosure is then to move his belief toward an intermediate region $[\mu^L,\mu^H]$, where a comfort-oriented signal is relatively valuable. 
Consequently, the doctor prevents fear of a bad diagnosis from deterring testing when the patient is excessively pessimistic, and prevents him from viewing the test as having too little instrumental value when he is excessively optimistic.

\subsection{How Interim Disclosure Balances Warning and Comfort}
\label{subsection.interim_disclosure}

To understand why the two policies in Proposition \ref{prop:voluntary_testing} arise, it is useful to start from the persuasion problem after the test.
This is guaranteed by the observation that the model can be characterized by backward induction.
That is, the disclosure after the test is characterized first, taking the interim belief as given, and the ex ante disclosure is then chosen to determine the distribution of interim beliefs, taking the resulting continuation-value function as given.

Fix an interim belief $\mu_1$.
If the patient is tested and found sick, the doctor chooses the interim disclosure to maximize physical health subject to the incentive constraint for testing.
Formally, for a posterior belief $\mu$ about the untreated prospect, $P(\mu)$ denotes the doctor's conditional payoff when the patient is tested and found sick:
\begin{align*}
     P(\mu)=\left\{
      \begin{array}{cc}
          \bar{p} & \mu\leq\mu_e \\
          \underline{p}(\mu) & \mu>\mu_e
      \end{array}
     \right.
\end{align*}
Applying \eqref{eq:simplified_testing_constraint}, signal $\tau_1^1(\cdot|\mu_1)$ solves
\begin{equation}
\label{eq:interim_problem}
\begin{aligned}
  \max_{\tau_1^1\in\mathcal{R}(\mu_1)}:\,&\int_0^1P(\mu)\tau_1^1(d\mu)\\
    &\text{subject to: }\int_0^{1}V(\mu)\tau_1^1(d\mu)\geq\bar{V}_0(\mu_1).
\end{aligned}
\end{equation}

\begin{lemma}[Binary interim signals]
\label{lemma:binary_signals}
For each interim belief $\mu_1$, it is sufficient to consider interim signals with at most two posteriors.
If two posteriors are used, they lie on opposite sides of $\mu_e$.
\end{lemma}
\begin{proof}
Both regions $[0,\mu_e]$ and $(\mu_e,1]$ are convex, and on each region both $P$ and $V$ are weakly concave.
Proposition~3.2 of \citet{doval2024constrained} therefore allows the doctor to restrict attention to signals with at most one posterior in each region, which proves the result.
\end{proof}
As is illustrated in Figure~\ref{fig-concavification}(a), the doctor-optimal signal is non-disclosure when $\mu_1<\mu_e$ and is warning-oriented ($\{1,\mu_e\}$) otherwise.
Figure~\ref{fig-concavification}(b) illustrates the other extreme, in which the doctor is instead concerned only with making testing attractive.
The patient-optimal signal is non-disclosure when
\begin{equation}
\label{eq:mu_v}
     \mu_1>\mu_v \equiv \arg\max_{\mu\in(\mu_e,1]}\frac{V(\mu)-V(0)}{\mu}
\end{equation}
and is comfort oriented ($\{\mu_v,0\}$) otherwise.

\begin{figure}[ht!]
   \centering
    \begin{subfigure}[t]{0.47\textwidth}
   \begin{tikzpicture}
   \draw[thick,->](0,0)--(4.5,0) node [below] {$\mu$};
   \draw[thick,->](0,0)--(0,4.5) node [left] {$P(\mu)$};
   \draw[thick, dashed, -](3.0189705450977025,4)--(3.0189705450977025,0) node [below] {$\mu_v$};
   \draw[thick, -, dashed](0,4)--(4,4);
   \draw[thick, -, dashed](0,2)--(4,2);
   \draw[thick, -](1,4)--(0,4)
    node [left] at (0,4) {$\bar{p}$}
    node [left] at (0,0) {$0$}
    node [left] at (0,1) {$\underline{p}(\mu_e)$}
    node [left] at (0,2) {$\underline{p}_H$}
    node [below] at (1,0) {$\mu_e$}
    node [below] at (4,0) {$1$};
   \draw[thick, -](1,1) -- (3.0189705450977025,1.67299);
   \draw[thick, -] (4,2) -- (3.0189705450977025,1.67299);
   \draw[thick, dashed, -](0,1)--(4,1);
   \draw[thick, dashed, -](1,4)--(4,2);
   \draw[thick, dashed, -](1,0)--(1,4);
   \draw[thick, dashed, -](4,0)--(4,4);
   \begin{scope}
      \clip(0,0)rectangle(5,5);
      \fill[gray, opacity=.5] (0,4) -- (1,4) -- (4,2) -- (1,1) -- cycle;
   \end{scope}
  \end{tikzpicture}
  \caption{Doctor-optimal signal.}
  \end{subfigure}
  \hfill
  \begin{subfigure}[t]{0.47\textwidth}
   \begin{tikzpicture}
   \draw[thick,->](0,0)--(4.5,0) node [below] {$\mu$};
   \draw[thick,->](0,0)--(0,4.5) node [left] {$V(\mu)$};
   \draw[domain=1:3.0189705450977025,thick] plot(\x,{(1 - exp(- 1.5 * (\x / 4)))/(1 - exp(-1.5))*4});
   \draw[domain=3.0189705450977025:4,thick] plot(\x,{(1 - exp(- 1.5 * (\x / 4)))/(1 - exp(-1.5))*4});
   \draw[domain=0:1,dashed] plot(\x,{(1 - exp(- 1.5 * (\x / 4)))/(1 - exp(-1.5))*4});
   \fill[gray, opacity=.5, domain=3.0189705450977025:4, variable=\x]
    (0,1.61011) -- (3.0189705450977025,3.48912) -- plot(\x,{(1 - exp(- 1.5 * (\x / 4)))/(1 - exp(-1.5))*4}) -- (1,1.61011) -- cycle;
   \draw[thick, -] (0,1.61011)--(1,1.61011);
   \draw[thick, dashed, -](0,4)--(4,4)
    node [left] at (0,4) {$\phi(\underline{p}(1))$}
    node [left] at (0,0) {$0$}
    node [left] at (0,1.61011) {$\phi(\underline{p}(\mu_e))$}
    node [below] at (3.0189705450977025,0) {$\mu_v$}
    node [below] at (1,0) {$\mu_e$}
    node [below] at (4,0) {$1$};
   \draw[thick, dashed, -](3.0189705450977025,0)--(3.0189705450977025,4);
   \draw[thick, dashed, -](0,1.61011)--(3.0189705450977025,3.48912);
   \draw[thick, dashed, -](0,4)--(4,4);
   \draw[thick, dashed, -](1,0)--(1,4);
   \draw[thick, dashed, -](4,0)--(4,4);
   \node[circle,fill=black,inner sep=0pt,minimum size=3pt] (a) at (3.0189705450977025,3.48912) {};
  \end{tikzpicture}
  \caption{Patient-optimal signal.}
  \end{subfigure}
  \caption{Doctor-optimal and patient-optimal signals.}
  \label{fig-concavification}
\end{figure}

The two signals suggest classifying interim beliefs according to whether the testing constraint is satisfied by the doctor-optimal signal, by full disclosure, or by some interim disclosure:
\begin{align*}
     \mathcal{F} & \equiv \{\mu_1\in[0,1]:\text{the doctor-optimal signal satisfies \eqref{eq:simplified_testing_constraint}}\}, \\
     \mathcal{D} & \equiv \{\mu_1\in[0,1]:\text{full disclosure of $\underline{p}$ satisfies \eqref{eq:simplified_testing_constraint}}\}, \\
     \mathcal{M} & \equiv \{\mu_1\in[0,1]:\text{some interim disclosure in $\mathcal{R}(\mu_1)$ satisfies \eqref{eq:simplified_testing_constraint}}\}.
\end{align*}
On $\mathcal{F}$ the doctor-optimal signal is feasible, and therefore the doctor can maximize her payoff without losing the patient at the testing stage.
On $\mathcal{D}\setminus\mathcal{F}$ the testing constraint binds, and hence the doctor-optimal signal becomes infeasible, but full disclosure of the untreated prospect remains implementable.
On $\mathcal{M}\setminus\mathcal{D}$ even full disclosure is infeasible, and thus the doctor must further soften the post-test message relative to full disclosure to keep the patient willing to be tested.

\begin{lemma}[Motivation-available beliefs]
\label{lemma:three_sets}
The three sets have the following structure.
\begin{itemize}
     \item If the patient is fear-reactive, i.e., if $V(0)\geq\bar V_0(0)$, then $\mathcal{M}=[0,\mu_\mathcal{M}]$, $\mathcal{D}=[0,\mu_\mathcal{D}]$, and $\mathcal{F}=[0,\mu_\mathcal{F}]$ for some $\mu_\mathcal{F}\leq\mu_\mathcal{D}\leq\mu_\mathcal{M}<1$.
     In other words, $\mathcal{F}\subseteq\mathcal{D}\subseteq\mathcal{M}$.
     \item If the patient is fear-avoidant, i.e., if $V(0)<\bar V_0(0)$, then $\mathcal{F}=\mathcal{D}=\emptyset$, and $\mathcal{M}$ is either empty or some interval $[\mu^L_\mathcal{M},\mu^H_\mathcal{M}]$.
\end{itemize}
\end{lemma}
\begin{proof}
See Step 1 of the proof of Proposition~\ref{prop:voluntary_testing} in Appendix~\ref{proof-exante}.
\end{proof}

\begin{proposition}[Optimal interim disclosure]
\label{prop:optimal_interim_disclosure}
Given interim belief $\mu_1$, the optimal signal after $a=1$ and $\theta=0$ is characterized as follows.
\begin{itemize}
     \item If $\mu_1\in\mathcal{F}$, the optimal signal is the doctor-optimal signal.
     \item If $\mu_1\in\mathcal{D}\setminus\mathcal{F}$, the optimal signal is warning-oriented: a binary signal whose upper posterior fully reveals the good state, with the lower belief $l(\mu_1)$ given by
     \begin{equation}
     \label{eq:l_mu}
         l(\mu_1)=1-(1-\mu_1)\frac{V(1)-V(0)}{V(1)-\bar{V}_0(\mu_1)}\leq \mu_e.
     \end{equation}
     \item If $\mu_1\in\mathcal{M}\setminus\mathcal{D}$, the optimal signal is comfort-oriented: a binary signal whose lower posterior fully reveals the bad state, with the upper belief $h(\mu_1)$ defined as the largest solution $x\in[\max\{\mu_1,\mu_e\},1]$ to
     \begin{equation}
     \label{eq:h_mu}
         \left(1-\frac{\mu_1}{h(\mu_1)}\right)V(0)+\frac{\mu_1}{h(\mu_1)}V(h(\mu_1))=\bar{V}_0(\mu_1).
     \end{equation}
\end{itemize}
\end{proposition}
\begin{proof}
See Steps 3--4 of the proof of Proposition~\ref{prop:voluntary_testing} in Appendix~\ref{proof-exante}.
\end{proof}

Proposition 2 shows that the optimal interim disclosure follows a best-good-news rule, under which it maximizes the upper posterior among all signals that satisfy the testing constraint.
In essence, as the testing constraint becomes harder to satisfy, the doctor must sacrifice some treatment incentives by providing less warning and eventually more comfort, and the best-good-news rule captures precisely this adjustment.
For beliefs in $\mathcal{F}$, the doctor-optimal signal is feasible, and hence the doctor can maximize treatment incentives without deterring the test.
For beliefs in $\mathcal{D}\setminus\mathcal{F}$, the testing constraint binds, but the doctor can still preserve a strong warning component by keeping the favorable posterior at $1$ and adjusting the unfavorable posterior.
For beliefs in $\mathcal{M}\setminus\mathcal{D}$, testing can be induced only by shifting the promised disclosure toward comfort. The unfavorable posterior is fixed at $0$, while the favorable posterior becomes less conclusive.

\begin{proposition}[Comparative statics of interim disclosure]
\label{prop:comparative_interim_disclosure}
When $\mu_1\in\mathcal{D}\setminus\mathcal{F}$, $l(\mu_1)$ is decreasing in $\mu_1$.
When $\mu_1\in\mathcal{M}\setminus\mathcal{D}$, $h(\mu_1)$ is weakly decreasing in $\mu_1$ if the patient is fear-reactive and strictly increasing in $\mu_1$ if the patient is fear-avoidant.
\end{proposition}
\begin{proof}
See Appendix \ref{proof.comparative_interim_disclosure}.
\end{proof}
Proposition 3 provides comparative statics that trace how the doctor trades off the patient’s willingness to be tested against treatment incentives as the interim belief changes.
For a fear-reactive patient, the optimal interim signal is initially the doctor-optimal signal $\{1,\mu_e\}$. 
As $\mu_1$ rises within $\mathcal D$, the testing constraint becomes increasingly difficult to satisfy, and thus the doctor must sacrifice some of her payoff to preserve the patient’s willingness to be tested. 
The best-good-news rule implies that she initially keeps the upper posterior at $1$ and reduces the probability that the good state is pooled into the bad-news message.
The lower posterior therefore falls, which raises both the instrumental value of information and its anticipatory value. 
The signal reaches the boundary $\mu_{\mathcal D}$ when the lower posterior reaches $0$ and full disclosure just satisfies the testing constraint. 
Beyond $\mathcal D$, the lower posterior cannot fall any further, and the doctor must instead provide comfort by pooling some of the bad state into the upper message, which lowers the upper posterior $h(\mu_1)$ and provides additional anticipatory value. 
This adjustment continues until the signal reaches the patient-optimal signal $\{0,\mu_v\}$.

Conversely, for a fear-avoidant patient, $\mathcal D$ is empty, and a higher $\mu_1$ makes the testing constraint easier to satisfy. 
To understand why $h(\mu_1)$ moves in the opposite direction, consider how a higher interim belief affects the patient’s participation surplus.
For any binary signal with posterior support $\{x,y\}$, define its participation surplus at interim belief $\mu_1$ as
\[
\Delta(\mu_1;x,y)
\equiv
\frac{x-\mu_1}{x-y}V(y)
+
\frac{\mu_1-y}{x-y}V(x)
-
\bar V_0(\mu_1).
\]
The signal satisfies the testing constraint if and only if $\Delta(\mu_1;x,y)\geq 0$. 
For a binding comfort-oriented signal, the binding condition and the linearity of $\bar V_0$ imply that, holding its posterior support fixed,
\[
\frac{\partial\Delta(\mu_1;h,0)}{\partial\mu_1}
=
\frac{\bar V_0(0)-V(0)}{\mu_1}.
\]
Therefore, the derivative is weakly negative for a fear-reactive patient and strictly positive for a fear-avoidant patient.
Thus, greater optimism weakly tightens the testing constraint for a fear-reactive patient but relaxes it for a fear-avoidant patient.
This reversal is the source of the opposite comparative statics in Proposition~\ref{prop:comparative_interim_disclosure}.
By the best-good-news rule, the doctor uses the additional participation slack for a fear-avoidant patient to raise $h(\mu_1)$, making the good news more favorable and strengthening treatment incentives.

\subsection{Ex Ante Disclosure}
\label{subsection.ex_ante_disclosure}

It remains to solve for the optimal ex ante disclosure by backward induction.
Let $\bar{P}(\mu_1)$ denote the doctor's payoff when the interim disclosure is chosen optimally at belief $\mu_1$:
\[
     \bar{P}(\mu_1)\equiv
     \begin{cases}
     \displaystyle\int_0^1 P(\mu)\tau_1^{1*}(d\mu\mid\mu_1),
     & \mu_1\in\mathcal{M},\\[0.8em]
     \underline{p}(\mu_1),
     & \mu_1\notin\mathcal{M},
     \end{cases}
\]
where $\tau_1^{1*}(\cdot\mid\mu_1)$ is the optimal post-test disclosure from Proposition~\ref{prop:optimal_interim_disclosure}.
The ex ante problem concavifies $\bar{P}$, as depicted in Figure~\ref{fig-exante}. 
For a fear-reactive patient, let $\mu^*\equiv\min\{\mu_e,\mu_{\mathcal F}\}$. 
No disclosure is optimal when $\mu_0\leq\mu^*$, while the optimal ex ante signal is $\{1,\mu^*\}$ when $\mu_0>\mu^*$, followed by no disclosure after testing. 
For a fear-avoidant patient, no policy induces testing if $\mathcal M$ is empty. 
If $\mathcal M=[\mu^L_{\mathcal M},\mu^H_{\mathcal M}]$ is non-empty, there exist $\hat\mu^L,\hat\mu^H\in\mathcal M$, with $\hat\mu^L\leq\hat\mu^H$, such that the optimal ex ante signal is $\{\hat\mu^L,0\}$ when $\mu_0<\hat\mu^L$, no disclosure when $\mu_0\in[\hat\mu^L,\hat\mu^H]$, and $\{1,\hat\mu^H\}$ when $\mu_0>\hat\mu^H$. 
After any interim belief in this interval, by Proposition~\ref{prop:optimal_interim_disclosure}, the doctor commits to a comfort-oriented disclosure if the patient is found sick.

The intuition begins with the optimal interim disclosure identified in Proposition~\ref{prop:optimal_interim_disclosure}, whose lower posterior is no greater than $\mu^*$. 
If this posterior is instead induced before testing, it becomes an interim belief at which the patient is willing to take the test and, if found sick, accept treatment without further disclosure. 
This property makes it possible to reconsider whether the disclosure needs to remain after the test.

For a fear-reactive patient, since the lower posterior is no greater than $\mu^*$, moving the interim disclosure to the ex ante stage does not lower the doctor’s health payoff. 
After the lower message, the patient remains willing both to take the test and to accept treatment if found sick.
After the upper message, the original interim policy would not have induced treatment, and therefore a refusal to be tested causes no additional health loss. 
Under any ex ante signal obtained by moving an optimal interim disclosure forward, whenever the patient takes the test, he accepts treatment if found sick without further disclosure. 
The doctor therefore maximizes her health payoff by maximizing the probability of testing, which is achieved by warning-in-advance, $\{1,\mu^*\}$, followed by no disclosure after the test.

For a fear-avoidant patient, the same transformation is unavailable. 
The lower posterior of the comfort-oriented interim signal is $0$, but disclosing this message before testing would cause the patient to refuse the test. 
The doctor therefore uses ex ante disclosure to move the patient’s belief into an intermediate region and commits to a comfort-oriented signal after a bad diagnosis. 

\begin{figure}[ht!]
   \centering
    \begin{subfigure}[t]{0.47\textwidth}
   \begin{tikzpicture}
   \draw[thick,->](0,0)--(4.5,0) node [right] {$\mu$};
   \draw[thick,->](0,0)--(0,4.5) node [left] {$\bar{P}(\mu)$};
   \draw[thick, -, dashed] (0.5,4)--(4,1.5);
   \draw[thick,-] (0,4) -- (0.5,4);
   \draw[thick] (0.5,4) -- (1,3.5);
   \draw[domain=1:3,thick] plot(\x,{-(1/8)*\x^2-3/4*\x+35/8});
   \draw[thick,-] (3,1) -- (4,1.5);
   \fill[gray, opacity=.5]
      (0,4) -- (0.5,4) -- (4,1.5) -- (3,1)  -- cycle;
   \draw[thick, -, dashed](0,4)--(4,4)
    node [left] at (0,4) {$\bar{p}$}
    node [left] at (0,0) {$0$}
    node [left] at (0,1.5) {$\underline{p}_H$}
    node [below] at (4,0) {$1$};
   \draw[thick, dashed, -](0.5,4)--(0.5,0) node [below] {$\mu^*$};
   \draw[thick, dashed, -](1,4)--(1,0) node [above right] {$\mu_\mathcal{D}$};
   \draw[thick, dashed, -](3,4)--(3,0) node [below] {$\mu_\mathcal{M}$};
   \draw[thick, dashed, -](0,1.5)--(4,1.5);
   \draw[thick, dashed, -](4,0)--(4,4);
  \end{tikzpicture}
  \caption{Fear-reactive patient.}
  \end{subfigure}
  \hfill
    \begin{subfigure}[t]{0.47\textwidth}
   \begin{tikzpicture}
   \draw[thick,->](0,0)--(4.5,0) node [right] {$\mu$};
   \draw[thick,->](0,0)--(0,4.5) node [left] {$\bar{P}(\mu)$};
   \draw[thick, -] (0,1) -- (1,1.5);
   \draw[thick, -] (3,2.5) -- (4,3);
   \draw[thick, -, dashed] (1,1.5) --(3,2.5);
   \draw[thick, -, dashed] (1.48339,3.91238) -- (4,3);
   \draw[thick, -, dashed] (0,1) -- (1,4);
   \draw[domain=1:3,thick] plot(\x,{-3/8*\x^2+3/4*\x+29/8});
   \fill[gray, opacity=.5, domain=1:1.48339, variable=\x]
    (0,1) -- (1,4) -- plot(\x,{-3/8*\x^2+3/4*\x+29/8}) -- (4,3) -- cycle;
   \draw[thick, -, dashed](0,4)--(4,4)
    node [left] at (0,4) {$\bar{p}$}
    node [left] at (0,0) {$0$}
    node [left] at (0,1) {$\underline{p}_L$}
    node [left] at (0,3) {$\underline{p}_H$}
    node [below] at (4,0) {$1$};
   \draw[thick, dashed, -](1,4)--(1,0) node [below] {$\mu^L_\mathcal{M}$};
   \draw[thick, dashed, -](3,4)--(3,0) node [below] {$\mu^H_\mathcal{M}$};
   \draw[thick, dashed, -] (1.48339,4) -- (1.48339,0) node [above right] {$\hat{\mu}^H$};
   \draw[thick, dashed, -](0,1)--(4,1);
   \draw[thick, dashed, -](0,3)--(4,3);
   \draw[thick, dashed, -](4,0)--(4,4);
  \end{tikzpicture}
  \caption{Fear-avoidant patient.}
  \end{subfigure}
  \caption{Optimal ex ante disclosure under voluntary testing.}
  \label{fig-exante}
\end{figure}

\section{Voluntary Consultation and Ex Ante Participation}
\label{section.exante}

 The previous section assumes that the patient is at least willing to enter the consultation and listen to the doctor's initial message. 
 This section relaxes that assumption and considers a stronger form of avoidance: the patient may refuse the consultation altogether before any disclosure has taken place. 
 The problem now involves not only voluntary testing, but also \textit{voluntary consultation}, because the doctor must ensure that the patient is willing to enter the consultation in the first place.

 In this case, the ex ante disclosure faces an additional \textit{ex ante participation constraint}.
 That is, before the ex ante disclosure, the patient now has an additional choice to refuse the medical consultation and receive the payoff $V_0(\mu_0)=\phi(\alpha+(1-\alpha)\underline{p}(\mu_0))$.
 Thus, given the triple $\langle\tau_0,\tau_1^1,\tau_1^0\rangle$, let
 \begin{align*}
     \overline{\mathcal{M}}\equiv\left\{\mu_1\in\text{Supp}(\tau_0)\bigg|\,\int_0^{1}V(\mu)\tau_1^1(d\mu|\mu_1)\geq\int_0^1V_0(\mu)\tau_1^0(d\mu|\mu_1)\right\}
 \end{align*}
 be the set of interim beliefs at which the patient is willing to be tested. 
 Then if policy $\langle\tau_0,\tau_1^1,\tau_1^0\rangle$ is sufficient to guarantee the ex ante participation constraint, it must satisfy
 \begin{equation}
     \label{eq.exantePC}
     \int_{\overline{\mathcal{M}}}\left(\int_0^1V(\nu)\tau_1^1(d\nu|\mu)\right)\tau_0(d\mu)+\int_{\overline{\mathcal{M}}^C}\left(\int_0^1V_0(\nu)\tau_1^0(d\nu|\mu)\right)\tau_0(d\mu)\geq V_0(\mu_0).
 \end{equation}

 The main technical difficulty here is that the ex ante and interim disclosures can no longer be solved separately by backward induction.
 This is because, once listening to the doctor is voluntary, the ex ante disclosure cannot be designed solely with the post-test treatment incentive in mind: it must simultaneously balance the patient's willingness to be tested against his willingness to be treated.
 Hence, the ex ante and interim disclosures jointly determine whether the ex ante participation constraint is satisfied.
 
 To solve the problem, I introduce the following notation. 
 First, define function
 \begin{align*}
     \mathcal{V}(\mu)\equiv\max\left\{V_0(\mu),\alpha+(1-\alpha)\phi(\bar{p}-c)\right\}
 \end{align*}
 as the patient's payoff function when the doctor commits to disclose no information in the interim disclosure. 
 That is, the patient either takes the test and accepts the treatment, or refuses to see the doctor at the outset.
 Second, belief $\mu_\mathcal{N}$ is the solution to equation 
 \begin{align*}
    \phi(\alpha+(1-\alpha)\underline{p}(\mu_\mathcal{N}))=\alpha+(1-\alpha)\phi(\bar{p}-c),
 \end{align*}
 that is, given there is no information in the interim disclosure, $\mu_\mathcal{N}$ is the threshold at which the patient is indifferent between taking and not taking the test.
 
 \begin{lemma}
 \label{lemma.SRP}
   Suppose the patient is fear-reactive. 
   The doctor's optimal payoff can be implemented by an information policy that has at most two ex ante posteriors and discloses no information at the interim stage.
   If such a policy induces treatment with positive probability and $\mu_0>\mu_\mathcal N$, it is binary with the lower belief no greater than $\mu_\mathcal N$.
   Treatment occurs only when the lower belief is realized.
 \end{lemma}
 \begin{proof}
     See Appendix \ref{proof.SRP}.
 \end{proof}
Lemma \ref{lemma.SRP} shows that, for a fear-reactive patient, interim disclosure is not needed in the optimal policy. Any feasible dynamic policy can be replaced by a binary ex ante policy that preserves the doctor's payoff and makes entering the consultation weakly more attractive. The replacement combines all histories that lead to treatment into one lower posterior and all remaining histories into one upper posterior. At the lower posterior, the patient is willing to take the test and, if found sick, accepts treatment without further disclosure. Interim disclosure can therefore be dispensed with.

 Lemma~\ref{lemma.SRP} therefore reduces the problem to the same binary-disclosure problem studied in Subsection~\ref{subsection.interim_disclosure}, with the ex ante participation constraint replacing the testing constraint. Belief $\mu_\mathcal{V}$ is identified by
 \begin{align*}
     \mu_\mathcal{V}=\sup\left\{\mu\bigg|\,\operatorname{cav}\circ\mathcal{V}(\mu)>\mathcal{V}(\mu)\right\},
 \end{align*}
 and hence by concavification, if there is no interim disclosure, the patient can benefit from persuasion if and only if $\mu_0<\mu_\mathcal{V}$.
 
 \begin{proposition}
 \label{proposition-PC}
  Suppose the ex ante participation constraint (\ref{eq.exantePC}) is present.
  If the patient is fear-reactive, it is optimal for the doctor to disclose no information in the interim disclosure.
  In the ex ante disclosure, the patient can be induced to take the test if and only if $\mu_0<\mu_\mathcal{V}$.
  The optimal ex ante disclosure is given as follows.
  \begin{itemize}
      \item If $\mu_0\leq\mu_\mathcal{N}$, non-disclosure is optimal.
      \item If $\mu_0\in(\mu_\mathcal{N},\mu_\mathcal{V})$ and full disclosure satisfies the ex ante participation constraint, the optimal ex ante signal is warning-oriented, with the lower belief equal to
      \begin{equation}
        \label{eq.thelowerbelief}
          1-(1-\mu_0)\cdot\frac{\mathcal{V}(1)-\mathcal{V}(0)}{\mathcal{V}(1)-\mathcal{V}(\mu_0)}.
      \end{equation}
      \item If $\mu_0\in(\mu_\mathcal{N},\mu_\mathcal{V})$ and full disclosure violates the ex ante participation constraint, the optimal ex ante signal is comfort-oriented, with the upper belief $\bar{h}(\mu_0)$ given by
      \begin{align*}
          \left(1-\frac{\mu_0}{\bar{h}(\mu_0)}\right)\cdot\mathcal{V}(0)+\frac{\mu_0}{\bar{h}(\mu_0)}\cdot\mathcal{V}\left(\bar{h}(\mu_0)\right)=\mathcal{V}(\mu_0).
      \end{align*}
  \end{itemize}
  When the patient is fear-avoidant, no policy can satisfy ex ante participation while inducing treatment with positive probability.
 \end{proposition}
 \begin{proof}
     See Appendix \ref{proof.PC}.
 \end{proof}

 For a fear-reactive patient, Lemma~\ref{lemma.SRP} shows that the optimal policy uses only ex ante disclosure. Proposition~\ref{proposition-PC} then shows that the best-good-news rule identified in Proposition~\ref{prop:voluntary_testing} continues to apply. Starting from warning-in-advance, the doctor keeps the upper belief at $1$ and makes the ex ante signal more informative by lowering its lower belief. The low message still induces testing and treatment. The doctor continues this adjustment until the ex ante participation constraint is satisfied or full disclosure is reached. If full disclosure satisfies the constraint, the optimal signal remains warning-oriented. If full disclosure violates it, the doctor uses a comfort-oriented ex ante signal, which I call \textit{precautionary comfort}. 

 This comfort-oriented signal in the precautionary comfort policy plays a dual role. By allowing some unfavorable prognoses to generate reassuring information, it makes the consultation acceptable to the patient. At the same time, it retains a message that reveals the worst prognosis, and thus the patient's belief falls to $0$. This maximally pessimistic belief plays the same decisive role as under warning-in-advance, inducing a fear-reactive patient to take the test and accept treatment if found sick. Precautionary comfort therefore satisfies the ex ante participation constraint without abandoning the voluntary-testing logic that calls for warning whenever pessimism induces both testing and treatment. For a fear-avoidant patient, however, since the same warning makes him refuse the test, even the reassuring message in the ex ante disclosure cannot induce treatment. Thus, no information policy can satisfy ex ante participation while inducing treatment with positive probability.

\section{The Role of Dynamic Disclosure and Dynamic Commitment}
\label{section.dynamic_role}

This section studies the roles of dynamic disclosure and dynamic commitment.
To establish a clear benchmark, consider first the case that the doctor can only disclose information before the testing decision.

\begin{proposition}[Value of dynamic disclosure]
\label{prop:value_dynamic}
If the patient is fear-reactive, the optimal static policy is the static counterpart of the warning-in-advance policy identified in Proposition~\ref{prop:voluntary_testing}.
If the patient is fear-avoidant, no static policy induces testing with positive probability.
\end{proposition}
\begin{proof}
See Appendix~\ref{proof.value_dynamic}.
\end{proof}

The failure of static persuasion for a fear-avoidant patient clarifies the role of dynamic disclosure. 
In the present model, interim disclosure, occurring between the testing and treatment decisions, allows the doctor to use a promised disclosure as a reward for a desirable earlier action.
Being forward-looking, the patient anticipates this disclosure when deciding whether to be tested, and hence the same disclosure also affects the willingness to be tested.
Moreover, if entering the consultation also requires the patient’s consent, ex ante disclosure also shapes the patient’s decision to enter the consultation in the first place.

Because the patient is information-avoidant, promising more information is not automatically a reward. 
More precise disclosure about the untreated prospect may sharpen the treatment decision, but it also inflicts anticipated psychological pain that deters testing. 
Instead, the reward must be tailored to the patient's anticipatory preferences, and comfort-oriented signals emerge. 
Committed comfort rewards the testing decision by promising a post-test environment in which the unfavorable state is partially pooled into a more optimistic message, softening the anticipated emotional blow of a bad diagnosis. 
Precautionary comfort instead rewards consultation entry directly through a comfort-oriented ex ante signal.
In both cases, the information greenhouse is a carefully constructed shelter, a reward designed not to maximize the precision of future beliefs, but to make the path toward those beliefs emotionally acceptable.

The value of this backward reward depends entirely on dynamic commitment.
\begin{proposition}[Dynamic commitment]
\label{prop:dynamic_commitment}
The full-commitment value is attainable without dynamic commitment if and only if there exists a full-commitment-optimal policy under which every prescribed on-path disclosure remains optimal after preceding patient actions are sunk.
\begin{itemize}
    \item Under voluntary testing, the non-disclosure and warning-in-advance optima characterized in Proposition~\ref{prop:voluntary_testing} do not require dynamic commitment, whereas the committed-comfort optimum requires it.
    \item Under voluntary consultation, every informative optimum characterized in Proposition~\ref{proposition-PC} requires dynamic commitment.
\end{itemize}
\end{proposition}
\begin{proof}
See Appendix~\ref{proof.dynamic_commitment}.
\end{proof}

In dynamic persuasion, a disclosure situated between two patient actions plays a dual role: it shapes the subsequent action by changing beliefs and provides incentives for the preceding action by delivering continuation value.
A dynamic information policy does not require dynamic commitment if and only if each prescribed on-path signal remains optimal after the preceding action is sunk. 
Otherwise, the doctor has a strict incentive to revise the signal once that action is sunk.
Hence, the forward-looking patient, anticipating this revision, cannot be incentivized to take the test.\footnote{\citet{koh2024attention} show in an optimal-stopping environment that the full-commitment value can be attained by a sequentially optimal modification of the information tree, even when the original tree is not sequentially optimal. Thus, a profitable continuation deviation from a particular full-commitment policy does not by itself establish a value of commitment. Proposition~\ref{prop:dynamic_commitment} rules out this possibility here by asking whether any full-commitment-optimal policy is sequentially implementable.}

\begin{corollary}[No commitment under fear avoidance]
\label{cor:no_commitment}
If the patient is fear-avoidant and the doctor lacks dynamic commitment, the patient refuses testing after every ex ante message, and thus no equilibrium induces testing.
\end{corollary}
\begin{proof}
See Appendix~\ref{proof.no_commitment}.
\end{proof}

\section{Optimal Disclosure of Diagnostic Results}
\label{section.diagnostic_disclosure}

The baseline model assumes that the diagnostic result is fully revealing and confines the doctor's discretion to disclosure about the untreated prospect.
This reflects the practical constraint that doctors often cannot manipulate the disclosure of diagnostic test results.
In other settings, however, doctors have some discretion over how such results are communicated.
In this section, I extend the model by allowing the doctor to design disclosure of the diagnostic state $\theta$.

Before testing, the doctor commits to a stochastic disclosure rule that maps each diagnostic state $\theta$ into a probability distribution over messages.
The patient observes the rule before deciding whether to be tested and observes the realized message if he accepts.
Fix $\underline p$, $\bar p$ and $c$. 
Let $\alpha_0$ denote the prior that $\theta=1$, and let $\alpha$ denote a posterior. 
The rule induces some $\tau\in\mathcal R(\alpha_0)$, using the Bayes-plausibility notation from Section~\ref{section.model}.

Assume $\bar p-\underline p\geq c$ and define the treatment threshold
\begin{equation*}
\alpha_e=1-\frac{c}{\bar p-\underline p}.
\end{equation*}
Then after a message induces posterior belief $\alpha$, the patient's date-2 payoff from treatment is $\alpha+(1-\alpha)\bar p-c$, whereas his payoff from remaining untreated is $\alpha+(1-\alpha)\underline p$. 
He therefore accepts treatment if and only if $\alpha\leq\alpha_e$. 
With a slight abuse of notation, his continuation value can be written as
\begin{equation*}
V(\alpha)\equiv
\begin{cases}
\phi\bigl(\alpha+(1-\alpha)\bar p-c\bigr),&\alpha\leq\alpha_e,\\
\phi\bigl(\alpha+(1-\alpha)\underline p\bigr),&\alpha>\alpha_e.
\end{cases}
\end{equation*}
Unlike in the baseline model, each piece of the patient's continuation value is strictly concave in $\alpha$.
The doctor's corresponding expected health payoff is
\begin{equation*}
P(\alpha)=
\begin{cases}
\alpha+(1-\alpha)\bar p,&\alpha\leq\alpha_e,\\
\alpha+(1-\alpha)\underline p,&\alpha>\alpha_e.
\end{cases}
\end{equation*}

The same two unconstrained benchmarks from Section~\ref{subsection.interim_disclosure} can therefore be obtained by globally concavifying $P$ and $V$, as shown in Figure~\ref{fig:two_benchmarks}.
Panel~(a) gives the doctor-optimal benchmark, which ignores the testing constraint and maximizes the doctor's continuation payoff.
For $\alpha_0>\alpha_e$, every signal $\{1,y\}$ with $y\leq\alpha_e$ attains this value.
Among these doctor-optimal signals, the patient prefers $\{1,l(1)\}$, where $l(x)$ denotes the patient's preferred lower posterior among binary signals with upper posterior $x$.\footnote{Appendix~\ref{proof.optimal_diagnostic_disclosure} gives its chord-slope and tangency characterizations.} Remarkably, full disclosure is also among these doctor-optimal signals.
Panel~(b) gives the patient-optimal benchmark. 
This benchmark is attained by the disclosure that makes testing most attractive and has value $\operatorname{cav}\circ V(\alpha_0)$. 
Strict concavity on $[0,\alpha_e]$ and $(\alpha_e,1]$, together with the upward slope jump at $\alpha_e$, implies that the concave envelope of $V$ is linear over a unique interval with one endpoint below $\alpha_e$ and the other above it. 
Denote its lower and upper contacts by $l(\alpha_v)$ and $\alpha_v$, where
$0\leq l(\alpha_v)<\alpha_e<\alpha_v\leq1$.
The patient-optimal signal is $\{\alpha_v,l(\alpha_v)\}$ for $\alpha_0\in(\alpha_e,\alpha_v)$ and non-disclosure for $\alpha_0\geq\alpha_v$.

\begin{figure}[htbp!]
   \centering
    \begin{subfigure}[t]{0.47\textwidth}
   \begin{tikzpicture}
   \draw[thick,->](0,0)--(4.5,0) node [below] {$\alpha$};
   \draw[thick,->](0,0)--(0,4.5) node [left] {$P(\alpha)$};
   \fill[gray, opacity=.5]
      (0,2.04) -- (4,4) -- (1.6,1.48) -- cycle;
   \draw[thick, domain=0:1.6] plot(\x,{2.04+0.49*\x});
   \draw[thick, domain=1.6:4] plot(\x,{-0.2+1.05*\x});
   \draw[thick, dashed, domain=1.6:4] plot(\x,{2.04+0.49*\x});
   \draw[thick, dashed] (1.6,0)--(1.6,4)
      node [below] at (1.6,0) {$\alpha_e$};
   \draw[thick, dashed] (4,0)--(4,4)
      node [below] at (4,0) {$1$};
   \draw[thick, dashed] (0,4)--(4,4)
      node [left] at (0,4) {$1$};
   \node[circle,fill=black,inner sep=0pt,minimum size=3pt] at (1.6,2.824) {};
   \node[circle,draw=black,fill=white,inner sep=0pt,minimum size=3pt] at (1.6,1.48) {};
  \end{tikzpicture}
  \caption{{Doctor-optimal benchmark.}}
  \end{subfigure}
  \hfill
    \begin{subfigure}[t]{0.47\textwidth}
   \begin{tikzpicture}
   \draw[thick,->](0,0)--(4.5,0) node [below] {$\alpha$};
   \draw[thick,->](0,0)--(0,4.5) node [left] {$V(\alpha)$};
   \draw[thick, -, dashed, domain=0:2] plot(\x,{(1-(\x/4-1)^2)*4});
   \draw[thick, -, domain=0:2] plot(\x,{-0.5*(\x-2)^2+3});
   \draw[thick, -, domain=2:4] plot(\x,{(1-(\x/4-1)^2)*4});
   \draw[thick, -, dashed] (1.41421,2.82843) -- (2.82843,3.65685);
   \draw[thick, -, dashed] (1.41421,4) -- (1.41421,0);
   \draw[thick, -, dashed] (2.82843,4) -- (2.82843,0);
   \draw[thick, -, dashed] (2,4) -- (2,0);
   \fill[gray, opacity=.5, domain=0:1.41421, variable=\x]
    plot(\x,{-0.5*(\x-2)^2+3}) -- (2,3.17157) -- (2,2.5) -- cycle;
   \fill[gray, opacity=.5, domain=2.82843:4, variable=\x]
    (2,3.17157) -- plot(\x,{(1-(\x/4-1)^2)*4}) -- (2,2.5) -- cycle;
   \node[circle,fill=black,inner sep=0pt,minimum size=3pt] at (1.41421,2.82843) {};
   \node[circle,fill=black,inner sep=0pt,minimum size=3pt] at (2.82843,3.65685) {};
   \draw[thick, -, dashed](0,4)--(4,4)
    node [left] at (0,4) {$1$}
    node [below] at (4,0) {$1$}
    node [below] at (2,-.08) {$\alpha_e$}
    node [below] at (2.82843,-.08) {$\alpha_v$}
    node [below] at (1.3,0) {$l(\alpha_v)$}
    node [left] at (0,1) {$\phi(\bar p-c)$};
   \draw[thick, dashed, -](0,1)--(4,1);
   \draw[thick, dashed, -](4,0)--(4,4);
  \end{tikzpicture}
  \caption{{Patient-value benchmark.}}
  \end{subfigure}
\caption{The two global concavifications.}
\label{fig:two_benchmarks}
\end{figure}

Refusing the test leaves the patient's diagnostic belief at $\alpha_0$ and yields the outside option $\phi\bigl(\alpha_0+(1-\alpha_0)\underline p\bigr)$. 
If $\alpha_0\leq\alpha_e$, non-disclosure induces treatment, satisfies the testing constraint, and attains the highest feasible health payoff. 
The remainder of the section considers $\alpha_0>\alpha_e$, where the outside option equals $V(\alpha_0)$.
That is, the doctor chooses a distribution of post-test beliefs to solve
\begin{equation}
\begin{aligned}
\max_{\tau\in\mathcal R(\alpha_0)}\quad
&\int_0^1P(\alpha)\,\tau(d\alpha)\\
&\text{subject to}\quad
\int_0^1V(\alpha)\,\tau(d\alpha)\geq V(\alpha_0).
\end{aligned}
\label{eq:general_disclosure_problem}
\end{equation}
The constraint makes the patient willing to take the test; refusal leaves his belief at $\alpha_0$ and yields $V(\alpha_0)$.

As in Lemma \ref{lemma:binary_signals}, \citet[Proposition~3.2]{doval2024constrained} implies that I can restrict attention to at-most-binary signals.
Such signals can be identified as $\{x,y\}$ with $0\leq y\leq\alpha_e<\alpha_0<x\leq1$. 
Since only the lower posterior induces treatment, the doctor's problem can be simplified as
\begin{equation}
\begin{aligned}
\max_{x,y}\quad
&\frac{x-\alpha_0}{x-y}(1-y)\\
&\text{subject to}\quad
\frac{x-\alpha_0}{x-y}V(y)+\frac{\alpha_0-y}{x-y}V(x)\geq V(\alpha_0).
\end{aligned}
\label{eq:coverage}
\end{equation}

\begin{definition}[Warning-tilted and comfort-tilted diagnostic signals]
\label{def:diagnostic_signal_orientation}
 An informative binary diagnostic signal $\{x,y\}$, where $y<\alpha_0<x$, is \textit{warning-tilted} if $(\alpha_0-y)/(x-y)<\alpha_0$. 
 It is \textit{comfort-tilted} if $(\alpha_0-y)/(x-y)>\alpha_0$.
\end{definition}

Under full disclosure, the favorable message occurs with probability $\alpha_0$; under $\{x,y\}$, it occurs with probability $(\alpha_0-y)/(x-y)$.
Every warning-oriented signal is warning-tilted, because $x=1$ implies $(\alpha_0-y)/(x-y)<\alpha_0$. Likewise, every comfort-oriented signal is comfort-tilted, because $y=0$ implies $(\alpha_0-y)/(x-y)>\alpha_0$.
Unlike their oriented counterparts, tilted signals need not fully reveal either state, and hence both good and bad news may be imperfect.

\begin{proposition}[Optimal diagnostic disclosure]
\label{proposition.optimal_diagnostic_disclosure}
There exist unique cutoffs $\alpha_g$ and $\alpha_w$, satisfying
$\alpha_e\leq\alpha_g\leq\alpha_w\leq\alpha_v\leq1$, such that, for $\alpha_0\in(\alpha_e,1)$, optimal disclosure is characterized by the following four cases.
\begin{enumerate}
\item If $\alpha_0\in(\alpha_e,\alpha_g]$, an optimal disclosure is the doctor-optimal signal $\{1,l(1)\}$.
\item If $\alpha_0\in(\alpha_g,\alpha_w]$, the testing constraint binds and the unique optimal posterior distribution is binary. It is warning-tilted if $\alpha_0<\alpha_w$.
\item If $\alpha_0\in(\alpha_w,\alpha_v)$, the testing constraint binds and the unique optimal posterior distribution is a comfort-tilted binary signal.
\item If $\alpha_0\in[\alpha_v,1)$, no feasible disclosure induces treatment with positive probability.
\end{enumerate}
\end{proposition}
\begin{proof}
See Appendix~\ref{proof.optimal_diagnostic_disclosure}.
\end{proof}

Proposition~\ref{proposition.optimal_diagnostic_disclosure} reproduces the four-region structure of Propositions~\ref{prop:optimal_interim_disclosure} and~\ref{proposition-PC}, which indicates that the basic trade-off between warning and comfort is robust to allowing the doctor to disclose the diagnostic result itself.
As in the participation-surplus comparison following Proposition~\ref{prop:comparative_interim_disclosure}, holding a binding signal in Proposition~\ref{proposition.optimal_diagnostic_disclosure} fixed, a higher prior belief in health reduces the gain from testing. 
The patient therefore becomes harder to persuade, and the doctor must use disclosure to add progressively more value to testing. 
The optimal policy starts from the doctor-optimal disclosure $\{1,l(1)\}$.
Once the testing constraint binds, the signal moves toward comfort, becomes comfort-tilted, and eventually can no longer induce treatment with positive probability. 
The following corollary describes this adjustment.

\begin{corollary}[Diagnostic optimism]
\label{cor:interior_diagnostic_path}
For $\alpha_0\in(\alpha_g,\alpha_v)$, write the unique optimal signal as $\{x(\alpha_0),y(\alpha_0)\}$. As $\alpha_0$ rises, $x(\alpha_0)$ falls strictly and $y(\alpha_0)$ falls weakly.
\end{corollary}

The result indicates that when the object of information disclosure changes from auxiliary information about the disease to the diagnostic result itself, the optimal signal need not fully reveal either the favorable or the unfavorable state.
Instead, the optimal signal may be warning/comfort-tilted.
However, as is indicated by Corollary~\ref{cor:interior_diagnostic_path}, the move from warning to comfort is still monotone in the prior belief, which is consistent with Proposition \ref{prop:comparative_interim_disclosure}.

\section{Conclusion}
\label{section.conclusion}

This paper makes two main contributions.
First, it identifies a dynamic persuasion problem that arises when a patient has information-avoidance preferences and the doctor uses disclosure about the untreated prognosis to influence his testing and treatment decisions.
Second, it characterizes when an information greenhouse is optimal, as well as when it can be implemented and how broadly its logic applies.
Unlike warning-based disclosure, an information greenhouse emerges as an optimal strategy when information itself is psychologically costly and warning would scare the patient away.
More specifically, under voluntary testing, it takes the form of committed comfort, which rewards testing with comfort-oriented disclosure after the patient is found sick.
Under voluntary consultation, it takes the form of precautionary comfort, which rewards entering the consultation with comfort-oriented disclosure before the testing decision.

Together, these results provide a practical framework for understanding when warning and comfort are appropriate tools, and they caution against two simple approaches to communicating with information-avoidant patients.
First, warning can motivate treatment, but it may deter the patient from testing or even entering the consultation. A doctor who ignores the patient's information avoidance and relies on warning-based disclosure may therefore sacrifice attainable health gains.
Second, comfort can sustain participation, but it may leave the patient insufficiently motivated to accept treatment.
Therefore, the doctor should warn when fear supports action without undermining participation and use delayed comfort only when warning would otherwise drive the patient away.

Crucially, an information greenhouse can be implemented only if the doctor can credibly commit to future disclosure.
The patient must believe that the promised comfort will still be delivered after he enters the consultation or takes the test.

\newpage
\appendix
\numberwithin{figure}{section}
\part*{Appendix}
\addcontentsline{toc}{part}{Appendix}

\section{Proofs}
\label{appendix.proofs}

\subsection{Proof of Proposition \ref{prop:voluntary_testing}}
\label{proof-exante}

\noindent\textbf{Step 1:} Structure of the motivation-available beliefs.

\medskip
 As a function of the interim belief $\mu_1$, the expected payoff of the doctor-optimal signal is piecewise linear and convex, and at point $\mu=1$ it equals $V(1)=\alpha+(1-\alpha)\phi(\underline{p}_H)$.
 Meanwhile, since the optimal continuation disclosure after refusal is full disclosure of the untreated prospect (by concavity of $V_0$, established in Section~\ref{subsection.baseline}), the payoff of rejecting the test is linear in $\mu$, and at point $\mu=1$ it equals $\bar{V}_0(1)=\phi(\alpha+(1-\alpha)\underline{p}_H)\geq V(1)$ by the concavity of $\phi$.
 Therefore, $\mathcal{F}$ is either empty or an interval $[0,\mu_\mathcal{F}]$, which is fully determined by the comparison at point $\mu=0$.
 In other words, it is non-empty if and only if
 \begin{align*}
     V(0)=\alpha+(1-\alpha)\phi(\bar{p}-c)\geq\bar{V}_0(0)=\phi(\alpha+(1-\alpha)\underline{p}_L),
 \end{align*}
 which implies (\ref{eq:fear_reactive}).
 The same reasoning applies for $\mathcal{D}$ since the expected payoff of full disclosure is linear.

 The expected payoff of the patient-optimal signal is concave, and is linear in $[0,\mu_v]$.
 Again, $\mathcal{M}$ is either empty or an interval since $\bar{V}_0$ is linear.
 If (\ref{eq:fear_reactive}) holds, it is $[0,\mu_\mathcal{M}]$ for some $\mu_\mathcal{M}<1$ since its value at $\mu=1$ is still $V(1)<\bar{V}_0(1)$.
 If (\ref{eq:fear_reactive}) does not hold, by the linearity of function $\operatorname{cav}\circ V(\mu)$ in interval $[0,\mu_v]$, $\mathcal{M}\subset[0,\mu_v]$ can never hold.
 Thus, it is non-empty if and only if the maximum value of $V(\mu)-\bar{V}_0(\mu)$ in $[\mu_v,1]$ is nonnegative, which establishes Lemma~\ref{lemma:three_sets}.

\bigskip
\noindent\textbf{Step 2:} A general result to characterize the optimal binary signal in a constrained Bayesian persuasion problem.

\medskip
In this step only, $\mu_0$ denotes the prior in a {generic} two-action, two-state persuasion problem.
It should not be confused with the ex ante prior in the medical model, which reappears below.

Consider a game played between a \textit{sender} and a \textit{receiver}. 
The sender designs an information structure, after which the receiver decides whether to accept the test. 
The state space is binary, and the test updates the prior belief $\mu_0\in(0,1)$ to a posterior belief $\mu\in[0,1]$, where beliefs are probabilities assigned to one of the two states.
Based on the posterior belief, the receiver chooses an action $e\in\{0,1\}$, and the game ends. 
The Bayesian persuasion problem can be written as
\begin{align*}
    \max_{\tau\in\mathcal{R}(\mu_0)}:&\int_0^1P(\mu)\tau(d\mu)\\
    &\text{subject to: }\int_0^1V(\mu)\tau(d\mu)\geq\bar{V}_0(\mu_0),\quad\text{(PC)}
\end{align*}
where $P$ and $V$ are the indirect utility functions of posterior belief $\mu$ for the sender and the receiver, respectively, and $\bar{V}_0(\mu_0)$ is the payoff when the receiver rejects the test.

Suppose that there exists $\mu_e\in(0,1)$ such that {the receiver chooses $e=1$ for beliefs in $[0,\mu_e]$ and $e=0$ for beliefs in $(\mu_e,1]$}.
Assume $P$ and $V$ are both {continuously differentiable and} concave on $[0,\mu_e]$ and on {$(\mu_e,1]$}.
Then it is without loss to focus on signals in
\begin{align*}
    \mathcal{S}(\mu_0)=\bigg\{\{x,y\}\bigg|1\geq x\geq\max\{\mu_0,\mu_e\}\geq\min\{\mu_0,\mu_e\}\geq y\geq 0\bigg\}.
\end{align*}
By concavification, the unconstrained optimum is also at most binary.
That is, for any $\mu_0$, there exists $\{\hat{x},\hat{y}\}\in\mathcal{S}(\mu_0)$ such that
\begin{align*}
    \hat{\tau}=\{\hat{x},\hat{y}\}\in\arg\max_{\tau\in\mathcal{R}(\mu_0)}\int_0^1P(\mu)\tau(d\mu).
\end{align*}
I focus on the case in which the (PC) constraint binds in the constrained optimum, which means that $\{\hat{x},\hat{y}\}$ violates (PC).
Thus, the optimization problem can be rewritten as
\begin{align*}
   \max_{x,y:\,\{x,y\}\in\mathcal{S}(\mu_0)}:&\,\frac{x-\mu_0}{x-y} P(y)+\frac{\mu_0-y}{x-y} P(x)\\
   &\text{subject to: }
   \frac{x-\mu_0}{x-y} V(y)+\frac{\mu_0-y}{x-y} V(x)=\bar{V}_0(\mu_0). \qquad \text{(PCB)}
\end{align*}
Here, $(x-\mu_0)/(x-y)$ and $(\mu_0-y)/(x-y)$ are respectively the probabilities that posteriors $y$ and $x$ realize, as implied by Bayesian plausibility.

Define the following objects for convenience.
For any $x\in[\max\{\mu_0,\mu_e\},1]$, let
\begin{align*}
   D(x)=\left\{y\in[0,\min\{\mu_0,\mu_e\}]\bigg|\,
   \frac{x-\mu_0}{x-y}V(y)+\frac{\mu_0-y}{x-y}V(x)=\bar{V}_0(\mu_0)\right\}.
\end{align*}
Thus $D(x)$ is the set of lower beliefs that make (PCB) bind.
If $D(x)$ is non-empty, define $d(x)$ by
\begin{align*}
   d(x)=\sup\bigg\{\arg\max_{y\in D(x)}:\,
   \frac{x-\mu_0}{x-y}P(y)+\frac{\mu_0-y}{x-y}P(x)\bigg\}.
\end{align*}
Thus, given $x$, $d(x)$ is the best lower belief that makes (PCB) bind. 
For any function $G$, define
\begin{align*}
   S_G(x)=\frac{G(x)-G(d(x))}{x-d(x)},
\end{align*}
the slope of the segment connecting $(x,G(x))$ and $(d(x),G(d(x)))$.
Finally, after substituting $y=d(x)$, define
\begin{align*}
    \mathcal{P}(x)=\frac{x-\mu_0}{x-d(x)}P(d(x))+\frac{\mu_0-d(x)}{x-d(x)}P(x).
\end{align*}

\begin{lem}
\label{bestgoodnewsrule}
   $\mathcal{P}'(x)\geq 0$ if and only if
   \begin{equation}
   \label{eq.bestgoodnewscriterion}
    \left(S_P(x)-P'(x)\right)\leq
    \frac{S_V(x)-V'(x)}{S_V(x)-V'(d(x))}
    \big(S_P(x)-P'(d(x))\big).
   \end{equation}
   In particular, if \eqref{eq.bestgoodnewscriterion} holds for all feasible $x$, the optimal signal is $\{x^*,d(x^*)\}$, where
   \[
       x^*=\sup\left\{x\in[\max\{\mu_0,\mu_e\},1]\bigg|\,D(x)\neq\emptyset\right\}.
   \]
\end{lem}
\begin{proof}
Using implicit differentiation,
\begin{equation}
\label{eq.d'x}
    d'(x)=-\frac{\mu_0-d(x)}{x-\mu_0}
    \frac{V(x)-V(d(x))-(x-d(x))V'(x)}
    {V(x)-V(d(x))-(x-d(x))V'(d(x))}.
\end{equation}
Also,
\begin{align*}
    \mathcal{P}'(x)=&-\left[\frac{\mu_0-d(x)}{x-d(x)}
    \left(\frac{P(x)-P(d(x))}{x-d(x)}-P'(x)\right)\right]\\
    &+d'(x)\left[-\frac{x-\mu_0}{x-d(x)}
    \left(\frac{P(x)-P(d(x))}{x-d(x)}-P'(d(x))\right)\right].
\end{align*}
Substituting \eqref{eq.d'x} into this expression, $\mathcal{P}'(x)\geq 0$ if and only if
\begin{align*}
    S_P(x)-P'(x)
    &\leq
    \frac{V(x)-V(d(x))-(x-d(x))V'(x)}
    {V(x)-V(d(x))-(x-d(x))V'(d(x))}
    \left(\frac{P(x)-P(d(x))}{x-d(x)}-P'(d(x))\right)\\
    &=\frac{S_V(x)-V'(x)}{S_V(x)-V'(d(x))}
    \left(S_P(x)-P'(d(x))\right),
\end{align*}
which completes the proof.
\end{proof}
By Lemma \ref{bestgoodnewsrule}, if condition \eqref{eq.bestgoodnewscriterion} is satisfied, the optimal signal can be derived in three steps: first check whether the unconstrained optimum satisfies (PC); if not, identify the largest upper belief $x$ such that $D(x)$ is non-empty; then choose $y$ from $D(x)$ to maximize the objective at that upper belief.

\bigskip
\noindent\textbf{Step 3:} Apply Lemma \ref{bestgoodnewsrule} to the interim problem.

\medskip
Fix an interim belief $\mu_1$.
By Lemma \ref{lemma:binary_signals}, it is enough to consider interim signals with at most two posteriors.
{If $\mu_1\in\mathcal{F}$, the doctor-optimal signal already satisfies the testing constraint, and thus I only need to consider $\mu_1\in\mathcal{M}\setminus\mathcal{F}$.}

I first show that the testing constraint (\ref{eq:simplified_testing_constraint}) binds.
Suppose not, such that an optimal signal $\{x,y\}$ satisfies the constraint with strict inequality.
{Suppose first that $x<1$.}
For any sufficiently small $\varepsilon>0$, continuity of $V$ implies that $\{x+\varepsilon,y\}$ remains feasible.
The doctor's payoff after $\theta=0$ is
\[
     \frac{x-\mu_1}{x-y}\bar{p}+\frac{\mu_1-y}{x-y}\underline{p}(x),
\]
which is strictly increasing in the upper posterior $x$.
Thus $\{x+\varepsilon,y\}$ dominates $\{x,y\}$, a contradiction.
{If $x=1$, hold $x$ fixed and increase $y$ toward $\min\{\mu_1,\mu_e\}$. Since $V(y)=V(0)<V(1)$ for $y\leq\mu_e$, this change strictly raises the doctor's payoff and continuously lowers the patient's expected payoff. It therefore remains feasible until the testing constraint binds or $y$ reaches $\min\{\mu_1,\mu_e\}$. In the latter case, the limiting signal is doctor-optimal. It is non-disclosure if $\mu_1\leq\mu_e$ and $\{1,\mu_e\}$ if $\mu_1>\mu_e$. Its feasibility would imply $\mu_1\in\mathcal F$, contradicting $\mu_1\in\mathcal M\setminus\mathcal F$. Hence the testing constraint binds.}

It remains to verify condition \eqref{eq.bestgoodnewscriterion} after substituting the generic prior in Lemma \ref{bestgoodnewsrule} with the interim belief $\mu_1$.
In the present model, $d(x)\leq\mu_e<x$ implies $V(d(x))=V(0)$, $V'(d(x))=0$, $P(d(x))=\bar{p}$, and $P'(d(x))=0$.
Moreover $S_P(x)<0$ and $P'(x)=\underline{p}_H-\underline{p}_L\geq0$.
Thus $S_P(x)-P'(d(x))<0$, and \eqref{eq.bestgoodnewscriterion} is equivalent to
\[
     \frac{S_P(x)-P'(x)}{S_P(x)-P'(d(x))}
     \geq
     \frac{S_V(x)-V'(x)}{S_V(x)-V'(d(x))}.
\]
The left-hand side is
\[
     1+\frac{(\underline{p}_H-\underline{p}_L)(x-d(x))}
     {\bar{p}-\underline{p}(x)}>1,
\]
whereas monotonicity of $V$ gives
\[
     \frac{S_V(x)-V'(x)}{S_V(x)-V'(d(x))}
     =
     1-V'(x)\frac{x-d(x)}{V(x)-V(0)}<1.
\]
Therefore condition \eqref{eq.bestgoodnewscriterion} holds, and Lemma \ref{bestgoodnewsrule} implies that every binding optimum maximizes the upper posterior among feasible binary signals.

\bigskip
\noindent\textbf{Step 4:} Characterize the optimal interim disclosure.

\medskip
 For any $\mu_1\in\mathcal{D}\setminus\mathcal{F}$, when the upper belief is $1$, since $\{0,1\}$ is patient-optimal among this subcategory and $\{0,1\}$ is sufficient to motivate $a=1$, there exists a lower belief $y\leq\mu_e$ such that $\{1,y\}$ makes the participation constraint (\ref{eq:simplified_testing_constraint}) hold with equality. 
 By Lemma \ref{bestgoodnewsrule}, given that the upper belief attains its maximum and the participation constraint binds, $\{1,y\}$ is optimal.
 Since $V(y)=V(0)=\alpha+(1-\alpha)\phi(\bar{p}-c)$, $y$ can be derived by equation
 \begin{align*}
     \frac{1-\mu_1}{1-y}\cdot V(0)+\frac{\mu_1-y}{1-y}\cdot V(1)=\bar{V}_0(\mu_1),
 \end{align*}
 which gives the required expression (\ref{eq:l_mu}).
 
 For $\mu_1\in\mathcal{M}\setminus\mathcal{D}$, no information structure with the upper belief being $1$ is sufficient to meet the participation constraint.
 Note that for any $x>\mu_e>y$, the patient's expected payoff of $\{x,y\}$, which is
 \begin{align*}
     \frac{x-\mu_1}{x-y}\cdot \underbrace{V(y)}_{=V(0)}+\frac{\mu_1-y}{x-y}\cdot V(x),
 \end{align*}
 is obviously a decreasing function of $y$ since $V(x)>V(0)$.
 Therefore, Lemma \ref{bestgoodnewsrule} implies that $y$ must be minimized in order to maximize $x$, and consequently the optimal signal is given by equation
 \begin{align*}
     {\left(1-\frac{\mu_1}{x}\right)\cdot V(0)+\frac{\mu_1}{x}\cdot V(x)=\bar{V}_0(\mu_1)},
 \end{align*}
 {For any solution $x$ to this equation, the doctor's health payoff is
 \[
 \left(1-\frac{\mu_1}{x}\right)\bar p+\frac{\mu_1}{x}\underline p(x),
 \]
 whose derivative with respect to $x$ is $\mu_1(\bar p-\underline p_L)/x^2>0$. The doctor therefore selects the largest solution; denote it by $h(\mu_1)$.}
 
 \bigskip
 \noindent\textbf{Step 5:} Optimal policy when the patient is fear-reactive.
 
 \medskip
 Suppose inequality (\ref{eq:fear_reactive}) holds.
 If $\mu_0\leq\min\{\mu_e,\mu_\mathcal{F}\}$, disclosing no information throughout ($\mu_0=\mu_1=\mu_2$) is sufficient to motivate $a=1$ and guarantees $e=1$ with probability $1$. 
 Consequently, it is optimal.

 Suppose $\mu_0>\min\{\mu_e,\mu_\mathcal{F}\}$.
 By the proposed warning-in-advance policy, the health probability is
 \begin{equation*}
     P^*(\mu_0)=\frac{1-\mu_0}{1-\min\{\mu_e,\mu_{\mathcal{F}}\}}\cdot\bar{p}+\frac{\mu_0-\min\{\mu_e,\mu_{\mathcal{F}}\}}{1-\min\{\mu_e,\mu_{\mathcal{F}}\}}\cdot \underline{p}_H.
 \end{equation*}
 By concavification, it suffices to show that for any $\mu_0>\min\{\mu_e,\mu_\mathcal{F}\}$, the consequent health probability by the optimal interim disclosure derived in Proposition \ref{prop:optimal_interim_disclosure}, denoted by $\bar{P}(\mu_0)$, is no larger than $P^*(\mu_0)$.
 
 First, if $\mu\in[\mu_e,\mu_{\mathcal{F}}]$, the optimal interim disclosure is $\{1,\mu_e\}=\{1,\min\{\mu_e,\mu_{\mathcal{F}}\}\}$, which is also doctor-optimal, and the health probability is exactly $P^*(\mu)$.

 {Second, if $\mu\in\mathcal{M}\setminus\mathcal{D}$, the optimal interim signal is comfort-oriented.}
 {For any comfort-oriented signal $\{0,x\}\in\mathcal{R}(\mu)$, the corresponding health probability is given by}
 \begin{align*}
     \bar{P}(\mu)=\left(1-\frac{\mu}{x}\right)\cdot\bar{p}+\frac{\mu}{x}\cdot \underline{p}(x).
 \end{align*}
 The first-order derivative with respect to $x$ is {$\mu(\bar{p}-\underline{p}_L)/x^2>0$}, and therefore
 \begin{align*}
     \bar{P}(\mu)&<\left(1-\mu\right)\cdot\bar{p}+\mu \underline{p}_H<\frac{1-\mu}{1-\min\{\mu_e,\mu_{\mathcal{F}}\}}\cdot\bar{p}+\frac{\mu-\min\{\mu_e,\mu_{\mathcal{F}}\}}{1-\min\{\mu_e,\mu_{\mathcal{F}}\}}\cdot \underline{p}_H={P^*}(\mu).
 \end{align*}

 {Third, if $\mu\in\mathcal{D}\setminus\mathcal{F}$, the interim signal is warning-oriented, and the health probability is}
 \begin{align*}
     \bar{P}(\mu)=\frac{1-\mu}{1-l(\mu)}\cdot \bar{p}+\frac{\mu-l(\mu)}{1-l(\mu)}\cdot \underline{p}_H.
 \end{align*}
 It suffices to show that $l(\mu)<\min\{\mu_e,\mu_\mathcal{F}\}$. 
 If $\mu_e<\mu_\mathcal{F}$, this claim holds directly.
 If $\mu_e>\mu_\mathcal{F}$, observe that
 \begin{align*}
  V(0)=\bar{V}_0(\mu_\mathcal{F})=\frac{\bar{V}_0(1)-\bar{V}_0(\mu)}{1-\mu}(\mu_\mathcal{F}-\mu)+\bar{V}_0(\mu)\Rightarrow\mu_\mathcal{F}=1-(1-\mu)\frac{\bar{V}_0(1)-V(0)}{\bar{V}_0(1)-\bar{V}_0(\mu)}.
 \end{align*}
 Comparing it with (\ref{eq:l_mu}), we have $\mu_\mathcal{F}>l(\mu)$ since $\bar{V}_0(1)>V(1)$.

 {Finally, if $\mu\notin\mathcal{M}$, testing cannot be induced and hence $\bar P(\mu)=\underline p(\mu)\leq P^*(\mu)$. Indeed,
 \[
 P^*(\mu)-\bar P(\mu)
 =
 \frac{1-\mu}{1-\min\{\mu_e,\mu_{\mathcal F}\}}
 \left[
 \bar p-\underline p\bigl(\min\{\mu_e,\mu_{\mathcal F}\}\bigr)
 \right]
 \geq0,
 \]
 because $\min\{\mu_e,\mu_{\mathcal F}\}\leq\mu_e$ and $\underline p(\mu_e)=\bar p-c<\bar p$.}

 \bigskip
 \noindent\textbf{Step 6:} Optimal policy when the patient is fear-avoidant. 
 
 \medskip
 Suppose now that (\ref{eq:fear_reactive}) does not hold.
 Then for any $\mu\in[\mu_\mathcal{M}^L,\mu_\mathcal{M}^H]$, 
 \begin{align*}
     \bar{P}(\mu)=\left(1-\frac{\mu}{h(\mu)}\right)\cdot\bar{p}+\frac{\mu}{h(\mu)}\cdot \underline{p}(h(\mu)),
 \end{align*}
 where $h(\mu)$ {denotes the largest solution to \eqref{eq:h_mu}, as established in Step 4}.
 For every $\mu\in(\mu_\mathcal{M}^L,\mu_\mathcal{M}^H)$, the largest solution satisfies $h(\mu)>\mu_v$, because fear avoidance makes $(\bar V_0(\mu)-V(0))/\mu$ strictly decreasing, while the lower endpoint satisfies
 \[
 \frac{\bar V_0(\mu_\mathcal{M}^L)-V(0)}{\mu_\mathcal{M}^L}
 =
 \frac{V(\mu_v)-V(0)}{\mu_v}.
 \]
 The largest-solution characterization therefore places $h(\mu)$ strictly on the decreasing portion of $[V(x)-V(0)]/x$. Since
 \[
 \frac{d}{dx}\frac{V(x)-V(0)}{x}
 =
 -\frac{V(x)-V(0)-xV'(x)}{x^2},
 \]
 this implies $D_1(\mu)>0$, and hence the implicit differentiation below is valid. The endpoint results follow by continuity.
 After multiplying both sides of equation \eqref{eq:h_mu} by $h(\mu)$ and differentiating implicitly, we obtain
 \begin{align*}
     {(h'(\mu)-1)\cdot V(0)+V(h(\mu))+\mu V'(h(\mu))h'(\mu)=\bar{V}_0'(\mu)h(\mu)+h'(\mu)\bar{V}_0(\mu)},
 \end{align*}
 which yields
 \begin{equation}
     \label{eq.h'}
     h'(\mu)=\frac{h(\mu)}{\mu}\cdot\frac{V(h(\mu))-V(0)-\bar{V}_0'(\mu)\cdot h(\mu)}{V(h(\mu))-V(0)-V'(h(\mu))\cdot h(\mu)}\equiv\frac{h(\mu)}{\mu}\cdot\frac{D_2(\mu)}{D_1(\mu)},
 \end{equation}
 where
 \begin{align*}
     \begin{array}{ll}
        D_1(\mu)=V(h(\mu))-V(0)-V'(h(\mu))\cdot h(\mu),\text{ and}\\
        D_2(\mu)=V(h(\mu))-V(0)-\bar{V}_0'(\mu)\cdot h(\mu).
    \end{array}
 \end{align*}
 The second-order derivative is given by
 \begin{align*}
     &\frac{\mu h'(\mu)^2 \left(h(\mu)^2 \left(-V''(h(\mu))\right)+2 h(\mu) V'(h(\mu))-2 V(h(\mu))+2 V(0)\right)}{h(\mu)}{+h(\mu)^2\bar V_0''(\mu)}\\&=\left(h(\mu) V'(h(\mu))-V(h(\mu))+V(0)\right) \left(\mu h''(\mu)+2 h'(\mu)\right).
 \end{align*}
 Substituting $h'(\mu)$ with expression (\ref{eq.h'}), we have
 \begin{equation}
 \label{eq.h''}
 \begin{aligned}
     h''(\mu)=&\frac{D_2(\mu)h(\mu)\left(2D_1(\mu)(D_2(\mu)-D_1(\mu))+D_2(\mu) h(\mu)^2 V''(h(\mu))\right)}{\mu^2 D_1(\mu)^3}\\
     &+\frac{h(\mu)^2\bar{V}_0''({\mu})}{\mu\left(h(\mu) V'(h(\mu))-V(h(\mu))+{V(0)}\right)}.
 \end{aligned}
 \end{equation}
 We have
 {\begin{align*}
     \bar{P}''(\mu)=&\bigg[\frac{2h(\mu)h'(\mu)-2\mu h'(\mu)^2+\mu h(\mu)h''(\mu)}{h(\mu)^3}\cdot\bar{p}\\
        &+\frac{\left(h(\mu) \underline{p}'(h(\mu))-\underline{p}(h(\mu))\right) \left(\mu h(\mu) h''(\mu)+2 h'(\mu) \left(h(\mu)-\mu h'(\mu)\right)\right)}{h(\mu)^3}\bigg].
 \end{align*}}
 Substituting $h'(\mu)$ and $h''(\mu)$ by (\ref{eq.h'}) and (\ref{eq.h''}), we obtain
 \begin{equation}
 \label{eq.pbn''}
     {\bar{P}''(\mu)=\underbrace{\frac{h(\mu) \underline{p}'(h(\mu))+\bar{p}-\underline{p}(h(\mu))}{D_1(\mu)}}_{>0}\bigg[\frac{D_2(\mu)^2 h(\mu) V''(h(\mu))}{\mu D_1(\mu)^2}-\bar{V}''_0(\mu)\bigg].}
 \end{equation}
 {Since $D_1(\mu)>0$, $V''(h(\mu))\leq0$, and $\bar{V}_0''(\mu)=0$, we have $\bar{P}''(\mu)\leq0$.}
 {For $\mu\notin\mathcal M$, $\bar P(\mu)=\underline p(\mu)$. For $\mu\in\mathcal M$,
 \[
 \bar P(\mu)-\underline p(\mu)
 =\left(1-\frac{\mu}{h(\mu)}\right)(\bar p-\underline p_L)
 \]
 is nonnegative and concave. Its concave envelope therefore consists of a possibly degenerate contact interval contained in $\mathcal M$ and affine segments connecting this interval to $0$ and $1$.}
 By concavification, the optimal signal must have the proposed structure.

\subsection{Proof of Proposition \ref{prop:comparative_interim_disclosure}}
\label{proof.comparative_interim_disclosure}
Fix a binding warning-oriented signal $\{1,l\}$ at belief $\mu\in\mathcal D\setminus\mathcal F$. Holding $l$ fixed, the linearity of $\bar V_0$ and the binding testing constraint give
\begin{align*}
\frac{\partial\Delta(\mu;1,l)}{\partial\mu}
&=\frac{V(1)-V(0)}{1-l}-\bar V_0'(\mu)
=\frac{V(1)-\bar V_0(1)}{1-\mu}\leq0,\\
\frac{\partial\Delta(\mu;1,l)}{\partial l}
&=-\frac{(1-\mu)\bigl(V(1)-V(0)\bigr)}{(1-l)^2}<0.
\end{align*}
The binding condition therefore implies
\[
l'(\mu)=-\frac{\partial\Delta/\partial\mu}{\partial\Delta/\partial l}\leq0.
\]

For the comfort-oriented region, the binding condition is
\[
\frac{V(h(\mu))-V(0)}{h(\mu)}
=
\frac{\bar V_0(\mu)-V(0)}{\mu}.
\]
Moreover, for $x>\mu_e$,
\[
\frac{d}{dx}\frac{V(x)-V(0)}{x}
=-\frac{D_1(x)}{x^2},
\]
where $D_1(x)\equiv V(x)-V(0)-xV'(x)$. Note that $D_1'(x)=-xV''(x)\geq0$. 
Moreover, strict concavity of $V$ on this interval guarantees that $D_1$ is strictly increasing. 
Hence $[V(x)-V(0)]/x$ is strictly decreasing to the right of its unique maximizer $\mu_v$. 
If $\mathcal M\setminus\mathcal D$ is nonempty, then $\mu_v<1$; otherwise full disclosure maximizes this ratio, and thus any belief in $\mathcal M$ also belongs to $\mathcal D$. 
For $\mu\in\mathcal M\setminus\mathcal D$, some comfort-oriented signal $\{0,x\}$ satisfies the testing constraint, whereas full disclosure $\{0,1\}$ violates it. Since $[V(x)-V(0)]/x$ is maximized at $\mu_v$, these two facts imply
\[
\frac{V(1)-V(0)}{1}
<
\frac{\bar V_0(\mu)-V(0)}{\mu}
\leq
\frac{V(\mu_v)-V(0)}{\mu_v}.
\]
The largest solution to the binding condition therefore lies in $[\mu_v,1)$ and is unique on this interval. Finally, the linearity of $\bar V_0$ gives
\[
\frac{d}{d\mu}
\frac{\bar V_0(\mu)-V(0)}{\mu}
=
\frac{V(0)-\bar V_0(0)}{\mu^2}.
\]
The largest-solution rule selects $h(\mu)$ from $[\mu_v,1)$, where $[V(x)-V(0)]/x$ is strictly decreasing. Meanwhile, $[\bar V_0(\mu)-V(0)]/\mu$ is weakly increasing under fear reactivity and strictly decreasing under fear avoidance. The binding condition therefore implies that $h(\mu)$ is weakly decreasing under fear reactivity and strictly increasing under fear avoidance. The weak decrease under fear reactivity and the strict increase under fear avoidance also hold at a boundary belief with $h(\mu)=\mu_v$, although implicit differentiation is unavailable at the boundary because $D_1(\mu)=0$. For interior beliefs satisfying $h(\mu)>\mu_v$, implicit differentiation gives
\begin{align}
 \label{eq:h_comparative_statics}
 h'(\mu)
 =\frac{h(\mu)^2}{\mu^2D_1(\mu)}
 \bigl(\bar V_0(0)-V(0)\bigr).
\end{align}
Since $D_1(\mu)>0$ at every interior belief, equation
\eqref{eq:h_comparative_statics} yields $h'(\mu)\leq0$ when
$V(0)\geq\bar V_0(0)$ (fear reactivity) and $h'(\mu)>0$ when
$V(0)<\bar V_0(0)$ (fear avoidance). The derivative signs at interior
beliefs and the monotonicity at a boundary belief satisfying
$h(\mu)=\mu_v$ prove Proposition~\ref{prop:comparative_interim_disclosure}.

\subsection{Proof of Lemma \ref{lemma.SRP}}
\label{proof.SRP}

We first show that we need only consider the case in which $\mu_0>\mu_\mathcal N$ and both treatment and no treatment occur with positive probability. If $\mu_0\leq\mu_\mathcal N$, non-disclosure induces testing and treatment, and it is therefore optimal because it attains the doctor's upper bound $\bar p$. If $\mu_0>\mu_\mathcal N$, by Bayes plausibility and the affinity of $\underline p$, any policy under which the patient is never treated gives the doctor the same expected health payoff as non-disclosure, while a policy that treats with probability one gives the patient $\alpha+(1-\alpha)\phi(\bar p-c)<V_0(\mu_0)$ and violates participation.

We therefore suppose that $\mu_0>\mu_\mathcal N$ and consider a feasible policy under which both treatment and no treatment occur with positive probability. Because treatment is relevant only when $\theta=0$, let $y$ and $x$ denote the posterior means conditional on $e=1$ and $e=0$, respectively, given $\theta=0$. Since $\underline p$ is independent of $\theta$, Bayes plausibility gives
\begin{equation}
\Pr(e=1\mid\theta=0)y
+
\Pr(e=0\mid\theta=0)x
=\mu_0,
\qquad
0\leq y\leq\mu_e,
\quad
0\leq x\leq1.
\label{eq:pc_aggregation_bayes}
\end{equation}
The bound $y\leq\mu_e$ follows because the treatment region $[0,\mu_e]$ is convex.
Conditional on $\theta=0$, the doctor's payoff is $\bar p$ when $e=1$ and $\underline p(\mu)$ when $e=0$. Since $\underline p$ is affine, replacing the posteriors in each group with their conditional mean leaves her expected payoff unchanged. The resulting payoff of the doctor is
\begin{equation}
\Pr(e=1\mid\theta=0)\bar p
+
\Pr(e=0\mid\theta=0)\underline p(x).
\label{eq:pc_aggregation_health}
\end{equation}

When $e=1$, the patient's payoff is $\alpha+(1-\alpha)\phi(\bar p-c)$. When $e=0$, he either declines the test and obtains $V_0(\mu)$ or takes the test but remains untreated and obtains
\[
V(\mu)
=\alpha+(1-\alpha)\phi(\underline p(\mu))
\leq
\phi\bigl(\alpha+(1-\alpha)\underline p(\mu)\bigr)
=V_0(\mu),
\]
where the inequality follows from the concavity of $\phi$ and $\phi(1)=1$. Thus, conditional on $e=0$, the patient's payoff is at most $V_0(\mu)$ at every realized posterior $\mu$. Since $x$ is the conditional mean of these posteriors and $V_0$ is concave, the patient's expected payoff conditional on $e=0$ is at most $V_0(x)$, and hence replacing these posteriors with $x$ can only weakly increase his payoff conditional on $e=0$. The patient's expected payoff under the original policy is therefore bounded above by
\begin{equation}
\Pr(e=1\mid\theta=0)
\bigl[\alpha+(1-\alpha)\phi(\bar p-c)\bigr]
+
\Pr(e=0\mid\theta=0)V_0(x).
\label{eq:pc_aggregation_utility}
\end{equation}
Ex ante participation and the concavity of $V_0$ give
\[
\begin{aligned}
&
\Pr(e=1\mid\theta=0)
\bigl[\alpha+(1-\alpha)\phi(\bar p-c)\bigr]
+
\Pr(e=0\mid\theta=0)V_0(x)
\\
&\qquad\geq V_0(\mu_0)\geq
\Pr(e=1\mid\theta=0)V_0(y)
+
\Pr(e=0\mid\theta=0)V_0(x).
\end{aligned}
\]
This implies $\alpha+(1-\alpha)\phi(\bar p-c)\geq V_0(y)$ and hence $y\leq\mu_\mathcal N$ by the definition of $\mu_\mathcal N$ and the strict monotonicity of $V_0$. The same monotonicity gives $\mu_\mathcal N\leq\mu_e$ because $\underline p(\mu_e)=\bar p-c$, the concavity of $\phi$, and $\phi(1)=1$ imply $V_0(\mu_e)\geq\alpha+(1-\alpha)\phi(\bar p-c)=V_0(\mu_\mathcal N)$.
Equation \eqref{eq:pc_aggregation_bayes} then gives $y\leq\mu_\mathcal N<\mu_0<x$.
Solving \eqref{eq:pc_aggregation_bayes} for the two conditional probabilities gives $\Pr(e=1\mid\theta=0)=(x-\mu_0)/(x-y)$ and $\Pr(e=0\mid\theta=0)=(\mu_0-y)/(x-y)$, and therefore \eqref{eq:pc_aggregation_health} equals
\[
\frac{x-\mu_0}{x-y}\bar p
+
\frac{\mu_0-y}{x-y}\underline p(x).
\]

To construct the policy stated in the lemma, let the ex ante signal induce posteriors $y$ and $x$ with these respective Bayes-plausible probabilities and disclose no further information after either message. At $y$, testing yields $\alpha+(1-\alpha)\phi(\bar p-c)\geq V_0(y)$, and the patient accepts treatment when sick. At $x$, refusing the test is optimal. If $x\leq\mu_e$, testing would yield $\alpha+(1-\alpha)\phi(\bar p-c)<V_0(x)$; if $x>\mu_e$, testing without treatment yields $V(x)\leq V_0(x)$. When the latter inequality binds, either action gives the same payoff and the doctor's health payoff is unchanged. Thus, the new policy gives the patient at least his payoff under the original policy and gives the doctor exactly \eqref{eq:pc_aggregation_health}. This proves the claimed reduction.
\subsection{Proof of Proposition \ref{proposition-PC}}
\label{proof.PC}

I first verify that $\mu_\mathcal N$ is well-defined and satisfies $\mu_\mathcal N\leq\mu_e$ under fear reactivity. Condition \eqref{eq:fear_reactive} gives $\alpha+(1-\alpha)\phi(\bar p-c)\geq V_0(0)$, while the concavity of $\phi$ and $\phi(1)=1$ imply $V_0(\mu_e)=\phi\bigl(\alpha+(1-\alpha)(\bar p-c)\bigr)\geq\alpha+(1-\alpha)\phi(\bar p-c)$. Since $V_0$ is continuous and strictly increasing, there is a unique $\mu_\mathcal N\in[0,\mu_e]$ satisfying $V_0(\mu_\mathcal N)=\alpha+(1-\alpha)\phi(\bar p-c)$.

Lemma \ref{lemma.SRP} allows us to restrict attention to binary ex ante signals with no interim disclosure. If $\mu_0\leq\mu_\mathcal N$, non-disclosure induces testing and treatment, and it is therefore optimal because it attains the doctor's upper bound $\bar p$.

Now let $\mu_0>\mu_\mathcal N$ and consider a feasible binary policy $\{x,y\}$ that induces treatment with positive probability, where $0\leq y\leq\mu_\mathcal N<\mu_0<x\leq1$. For a fixed $x$, Bayes plausibility assigns posterior $y$ probability $(x-\mu_0)/(x-y)$, which increases with $y$. Because the doctor's payoff is $\bar p$ at $y$ and $\underline p(x)$ at $x$, and $\bar p>\underline p(x)$, she prefers a larger $y$. At $y=\mu_\mathcal N$, concavity gives
\[
V_0(\mu_0)
\geq
\frac{x-\mu_0}{x-\mu_\mathcal N}
\bigl[\alpha+(1-\alpha)\phi(\bar p-c)\bigr]
+
\frac{\mu_0-\mu_\mathcal N}{x-\mu_\mathcal N}V_0(x).
\]
Thus, at $y=\mu_\mathcal N$, the ex ante participation constraint is either binding or violated. For a fixed $x$, the doctor therefore chooses the largest feasible $y$, at which the participation constraint binds. This gives
\begin{equation}
\frac{x-\mu_0}{x-y}
=
\frac{V_0(x)-V_0(\mu_0)}
{V_0(x)-\alpha-(1-\alpha)\phi(\bar p-c)},
\qquad
\frac{\mu_0-y}{x-y}
=
\frac{V_0(\mu_0)-\alpha-(1-\alpha)\phi(\bar p-c)}
{V_0(x)-\alpha-(1-\alpha)\phi(\bar p-c)}.
\label{eq:pc_binding_weight}
\end{equation}
Bayes plausibility then determines $y$ as
\begin{equation}
y
=
\frac{
\mu_0-
\dfrac{V_0(\mu_0)-\alpha-(1-\alpha)\phi(\bar p-c)}
{V_0(x)-\alpha-(1-\alpha)\phi(\bar p-c)}
x
}{
\dfrac{V_0(x)-V_0(\mu_0)}
{V_0(x)-\alpha-(1-\alpha)\phi(\bar p-c)}
}.
\label{eq:pc_lower_from_upper}
\end{equation}
To express $y\geq0$ as a restriction on $x$, define, for $z>0$, the participation slope
\[
s_0(z)
\equiv
\frac{V_0(z)-\alpha-(1-\alpha)\phi(\bar p-c)}{z}.
\]
The condition $y\geq0$ is equivalent to $s_0(x)\geq s_0(\mu_0)$. Whenever this inequality holds, the binding participation constraint and the concavity of $V_0$ imply $\alpha+(1-\alpha)\phi(\bar p-c)\geq V_0(y)$, and hence $y\leq\mu_\mathcal N$ automatically.

Substituting \eqref{eq:pc_binding_weight} into the doctor's expected payoff conditional on $\theta=0$ gives
\[
\bar p
+
\frac{V_0(\mu_0)-\alpha-(1-\alpha)\phi(\bar p-c)}
{V_0(x)-\alpha-(1-\alpha)\phi(\bar p-c)}
\bigl[\underline p(x)-\bar p\bigr].
\]
Differentiating this payoff with respect to $x$ gives
\[
\frac{V_0(\mu_0)-\alpha-(1-\alpha)\phi(\bar p-c)}
{\bigl[V_0(x)-\alpha-(1-\alpha)\phi(\bar p-c)\bigr]^2}
\bigg(
\underline p'(x)
\bigl[V_0(x)-\alpha-(1-\alpha)\phi(\bar p-c)\bigr]
+
[\bar p-\underline p(x)]V_0'(x)
\bigg)
>0.
\]
The doctor therefore selects the largest $x$ satisfying $s_0(x)\geq s_0(\mu_0)$.

It remains to describe this largest feasible posterior. For $x>0$,
\[
s_0'(x)
=
\frac{xV_0'(x)-V_0(x)+\alpha+(1-\alpha)\phi(\bar p-c)}{x^2},
\]
and the derivative of its numerator is $xV_0''(x)\leq0$. When $\alpha+(1-\alpha)\phi(\bar p-c)>V_0(0)$, strict concavity of $V_0$ implies that $s_0$ is strictly increasing up to its unique maximizer and strictly decreasing thereafter; the maximizer may be $1$. This maximizer is $\mu_\mathcal V$ as defined in the text.

If full disclosure satisfies the ex ante participation constraint, equivalently $s_0(\mu_0)\leq s_0(1)$, the largest feasible posterior is $x=1$. Equations \eqref{eq:pc_binding_weight} and \eqref{eq:pc_lower_from_upper} give the lower posterior in \eqref{eq.thelowerbelief}. If $s_0(\mu_0)=s_0(1)$, the lower posterior is $0$, and thus the optimal signal is full disclosure.

If $\mu_0\in(\mu_\mathcal N,\mu_\mathcal V)$ and full disclosure violates the ex ante participation constraint, equivalently $s_0(\mu_0)>s_0(1)$, the largest feasible posterior is the largest solution $\bar h(\mu_0)>\mu_0$ to $s_0\bigl(\bar h(\mu_0)\bigr)=s_0(\mu_0)$. At $x=\bar h(\mu_0)$, $s_0(x)=s_0(\mu_0)$ and \eqref{eq:pc_lower_from_upper} gives $y=0$. Rearranging $s_0\bigl(\bar h(\mu_0)\bigr)=s_0(\mu_0)$ gives
\[
\left(1-\frac{\mu_0}{\bar h(\mu_0)}\right)\mathcal V(0)
+
\frac{\mu_0}{\bar h(\mu_0)}\mathcal V(\bar h(\mu_0))
=
\mathcal V(\mu_0),
\]
and hence the optimal signal is the precautionary-comfort signal $\{0,\bar h(\mu_0)\}$.

If $\mu_0\geq\mu_\mathcal V$, no $x>\mu_0$ satisfies $s_0(x)\geq s_0(\mu_0)$. Thus no feasible dynamic policy induces treatment with positive probability.

It remains to compare the policy constructed above with policies that never induce treatment. For every $\mu_0\in(\mu_\mathcal N,\mu_\mathcal V)$, the constructed policy strictly improves the doctor's payoff relative to no treatment because
\[
\frac{x-\mu_0}{x-y}\bar p
+
\frac{\mu_0-y}{x-y}\underline p(x)
-\underline p(\mu_0)
=
\frac{x-\mu_0}{x-y}
\bigl[\bar p-\underline p(y)\bigr]
>0.
\]

If $\alpha+(1-\alpha)\phi(\bar p-c)=V_0(0)$, strict concavity of $V_0$ implies that $s_0$ is strictly decreasing, and therefore $\mu_\mathcal V=0$ and no positive prior permits treatment with positive probability.

Finally, suppose \eqref{eq:fear_reactive} fails, and hence $\alpha+(1-\alpha)\phi(\bar p-c)<V_0(0)$. We show that no feasible policy can induce treatment with positive probability. Consider first a policy under which both treatment and no treatment occur with positive probability. Let $y$ and $x$ denote the posterior means among realizations with and without treatment, respectively. The argument in the proof of Lemma \ref{lemma.SRP} bounds the patient's expected payoff by
\begin{equation}
\frac{x-\mu_0}{x-y}
\bigl[\alpha+(1-\alpha)\phi(\bar p-c)\bigr]
+
\frac{\mu_0-y}{x-y}V_0(x).
\label{eq:pc_fear_avoidant_bound}
\end{equation}
Fear avoidance makes the first term in \eqref{eq:pc_fear_avoidant_bound} strictly smaller than $\frac{x-\mu_0}{x-y}V_0(y)$, and concavity implies that the resulting weighted average of $V_0(y)$ and $V_0(x)$ is at most $V_0(\mu_0)$. Hence, the bound in \eqref{eq:pc_fear_avoidant_bound} is strictly below $V_0(\mu_0)$, contradicting ex ante participation. If instead treatment occurs with probability one, the patient obtains $\alpha+(1-\alpha)\phi(\bar p-c)<V_0(0)\leq V_0(\mu_0)$, and thus participation again fails. Hence no feasible policy induces treatment with positive probability.
\subsection{Proof of Proposition \ref{prop:value_dynamic}}
\label{proof.value_dynamic}

Under the static restriction, a message before the test induces some posterior $\mu$, after which no further information is disclosed. If $\mu\leq\mu_e$, the patient accepts treatment after an unfavorable diagnosis, and his value from testing is $V(\mu)=V(0)$. His value from refusing is $V_0(\mu)$.

Suppose first that the patient is fear-reactive. Define
\[
\mu_S\equiv\sup\{\mu\in[0,\mu_e]:V(0)\geq V_0(\mu)\}.
\]
Fear reactivity implies $V(0)\geq V_0(0)$, and hence $\mu_S$ is well-defined. Since $V_0$ is increasing, a static posterior $\mu$ induces both testing and treatment if and only if $\mu\leq\mu_S$. If $\mu_0\leq\mu_S$, non-disclosure induces both testing and treatment, and it is therefore optimal because it attains the doctor's upper bound $\bar p$. If $\mu_0>\mu_S$, the optimal static signal is $\{1,\mu_S\}$. To see this, note that the testing-and-treatment region $[0,\mu_S]$ and the refusal region $(\mu_S,1]$ are convex, and therefore aggregating all messages in the former into their conditional mean $\mu'\in[0,\mu_S]$ and all messages in the latter into their conditional mean $\mu''\in(\mu_S,1]$ preserves both the induced actions and Bayes plausibility. Affinity of $\underline p$ also preserves the doctor's payoff under this aggregation. For fixed $\mu'$, raising $\mu''$ to $1$ raises the probability of treatment. Given $\mu''=1$, the doctor's payoff is increasing in $\mu'$ because $\bar p>\underline p_H$. Hence $\mu'=\mu_S$ is optimal. The resulting policy is the static counterpart of warning-in-advance. It uses the same warning-oriented message structure but cannot condition subsequent disclosure on the testing decision.

Suppose next that the patient is fear-avoidant. Then
\[
V(0)<\bar V_0(0)=V_0(0).
\]
For every posterior $\mu\leq\mu_e$ that would induce treatment, $V(\mu)=V(0)<V_0(0)\leq V_0(\mu)$, and hence the patient refuses testing. For every posterior $\mu>\mu_e$, testing does not affect the treatment decision and only resolves diagnostic uncertainty. Strict concavity of $\phi$ implies
\[
U_1(\alpha,\mu)
=
\alpha+(1-\alpha)\phi\bigl(\underline p(\mu)\bigr)
<
\phi\bigl(\alpha+(1-\alpha)\underline p(\mu)\bigr)
=
U_0(\alpha,\mu),
\]
and thus the patient also refuses testing. Thus, for any ex ante disclosure $\tau_0\in\mathcal R(\mu_0)$, $U_1(\alpha,\mu)<U_0(\alpha,\mu)$ at every $\mu\in\operatorname{Supp}(\tau_0)$, and therefore
\[
\int_0^1 U_1(\alpha,\mu)\,\tau_0(d\mu)
<
\int_0^1 U_0(\alpha,\mu)\,\tau_0(d\mu)
\leq
U_0(\alpha,\mu_0).
\]
The patient refuses testing after every message, and thus no static policy induces testing with positive probability.

\subsection{Proof of Proposition \ref{prop:dynamic_commitment}}
\label{proof.dynamic_commitment}

The full-commitment value is attainable without dynamic commitment if and only if there exists a full-commitment-optimal policy under which every prescribed on-path disclosure remains optimal after preceding patient actions are sunk, with a continuation maximizer specified after every off-path history.

Under voluntary testing, non-disclosure and warning-in-advance satisfy this criterion. Whenever testing occurs, the posterior is no greater than $\mu^*=\min\{\mu_e,\mu_{\mathcal F}\}$, and hence no further disclosure induces treatment and is doctor-optimal; after a refusal, Bayes plausibility and the affinity of $\underline p$ make the doctor indifferent across continuation disclosures. By contrast, under fear avoidance, $\mathcal F=\varnothing$, and thus a doctor-optimal post-test continuation violates the testing constraint after every ex ante message. Hence, without dynamic commitment, no equilibrium induces testing, whereas every committed-comfort optimum induces treatment with positive probability and therefore requires dynamic commitment.

Finally, consider voluntary consultation. Non-disclosure remains optimal after entry when $\mu_0\leq\mu_{\mathcal N}$. For $\mu_0\in(\mu_{\mathcal N},\mu_{\mathcal V})$, strict concavity gives $\mu_{\mathcal N}<\mu^*=\min\{\mu_e,\mu_{\mathcal F}\}$. By Lemma~\ref{lemma.SRP}, the full-commitment value is attained by a binary signal $\{x,y\}$, where $y\leq\mu_{\mathcal N}<\mu_0<x\leq1$, with conditional-sick payoff
\begin{equation}
\label{eq:consultation_binary_payoff}
\frac{x-\mu_0}{x-y}P(y)
+
\frac{\mu_0-y}{x-y}P(x)
=
\underline p(\mu_0)
+
\frac{x-\mu_0}{x-y}
\bigl[\bar p-\underline p(y)\bigr].
\end{equation}
If $\mu_0\leq\mu^*$, switching after entry to non-disclosure yields $\bar p$, which is strictly greater than the payoff in \eqref{eq:consultation_binary_payoff}. If $\mu_0>\mu^*$, the payoff in \eqref{eq:consultation_binary_payoff} is strictly increasing in both $x$ and $y$, and is therefore strictly below its value at $\{1,\mu^*\}$, the sequentially implementable warning-in-advance policy. Thus, after entry, the doctor has a strict profitable deviation from the full-commitment consultation value, regardless of its information-tree representation. Every informative consultation optimum therefore requires dynamic commitment.

\subsection{Proof of Corollary \ref{cor:no_commitment}}
\label{proof.no_commitment}

Without dynamic commitment, after an unfavorable diagnosis the doctor selects a doctor-optimal interim signal. 
By the definition of $\mathcal F$, this continuation satisfies the testing constraint at interim belief $\mu_1$ if and only if $\mu_1\in\mathcal F$. 
But for a fear-avoidant patient, Lemma~\ref{lemma:three_sets} gives $\mathcal F=\emptyset$. 
The patient therefore anticipates that the doctor's continuation choice will violate the testing constraint after every ex ante message and refuses testing.

The set of outcomes available without dynamic commitment is contained in the full-commitment feasible set, and therefore full commitment weakly improves the doctor's health payoff. 
Without commitment, no testing occurs. If the full-commitment optimum induces treatment with positive probability, its health payoff is strictly higher because $\bar p>\underline p_H$.

\subsection{Proof of Proposition~\ref{proposition.optimal_diagnostic_disclosure}}
\label{proof.optimal_diagnostic_disclosure}

\noindent\textbf{Step 1:} Establish the concavification geometry and reduce the doctor's problem to binary posterior distributions.

\medskip

On both the treatment region and the no-treatment region, $V$ is strictly concave because $\phi$ is strictly concave and its argument is a nonconstant affine function of $\alpha$. 
At the treatment threshold, $V$ is continuous and its slope jumps upward, since
\begin{align*}
&(1-\underline p)\phi'\bigl(\alpha_e+(1-\alpha_e)\underline p\bigr)
-(1-\bar p)\phi'\bigl(\alpha_e+(1-\alpha_e)\bar p-c\bigr)\\
&\qquad
=(\bar p-\underline p)\phi'\bigl(\alpha_e+(1-\alpha_e)\underline p\bigr)>0.
\end{align*}
Thus the only failure of global concavity is at $\alpha_e$.

For every $x>\alpha_e$, define
\begin{equation}
l(x)
=
\arg\min_{0\leq z\leq\alpha_e}
\frac{V(x)-V(z)}{x-z}.
\label{eq:patient_locus}
\end{equation}
Now we show that $l(x)$ is well defined since the minimizer is unique. 
The objective in equation~\eqref{eq:patient_locus} is the slope of the chord connecting $(z,V(z))$ and $(x,V(x))$. 
Its derivative with respect to $z$ has numerator $V(x)-V(z)-(x-z)V'(z)$, where $V'(\alpha_e)$ is taken along the treatment region.
Strict concavity on the treatment region makes this numerator strictly increasing in $z$, and thus the chord slope first falls and then rises, with either part possibly absent at a boundary.
If $l(x)$ is interior, it satisfies $V'(l(x))=[V(x)-V(l(x))]/[x-l(x)]$.

Standard concavification now gives one affine segment crossing $\alpha_e$. 
Strict concavity on the two regions makes both contact points unique. 
Denote them by $l(\alpha_v)$ and $\alpha_v$, where $0\leq l(\alpha_v)<\alpha_e<\alpha_v\leq1$.
When both contact points are interior, the affine segment of the concave envelope that connects them is tangent to $V$ at both endpoints. 
At a boundary contact, the corresponding tangency equality is replaced by the one-sided supporting inequality. 
This establishes the patient-value benchmark used in the text without an additional regularity condition.

Both regions $[0,\alpha_e]$ and $(\alpha_e,1]$ are convex, and on each region both $P$ and $V$ are weakly concave. 
Proposition~3.2 of \citet{doval2024constrained} therefore allows the doctor to restrict attention to signals with at most two posterior beliefs. 
Any informative signal that induces treatment with positive probability can be written as $\{x,y\}$, where $0\leq y\leq\alpha_e<\alpha_0<x\leq1$.
Its expected health payoff is
\begin{equation*}
\alpha_0+(1-\alpha_0)\underline p
+(\bar p-\underline p)\frac{x-\alpha_0}{x-y}(1-y).
\end{equation*}
The lower-posterior message occurs with probability $(x-\alpha_0)/(x-y)$ and induces treatment. Since the posterior probability of sickness after this message is $1-y$, the factor $(x-\alpha_0)(1-y)/(x-y)$ is the ex ante probability that a sick patient receives a treatment-inducing message. Maximizing expected health payoff is therefore equivalent to maximizing this probability, which gives problem~\eqref{eq:coverage}.

\bigskip
\noindent\textbf{Step 2:} The doctor-optimal and no-treatment regions.

\medskip

If $\alpha_0\leq\alpha_e$, non-disclosure induces treatment and attains the highest feasible health payoff. 
Now suppose $\alpha_0>\alpha_e$. Bayes plausibility implies
\begin{align*}
1-\alpha_0
=
\frac{x-\alpha_0}{x-y}(1-y)
+\frac{\alpha_0-y}{x-y}(1-x)\geq
\frac{x-\alpha_0}{x-y}(1-y),
\end{align*}
with equality only when $x=1$. Every doctor-optimal signal therefore takes the form $\{1,y\}$ with $y\leq\alpha_e$. 
Among them, $\{1,l(1)\}$ maximizes the patient's value. 
Define
\begin{equation}
\alpha_g
=
\sup\left(
\{\alpha_e\}
\cup
\left\{
\alpha_0\in(\alpha_e,1):
\frac{1-\alpha_0}{1-l(1)}V(l(1))
+
\frac{\alpha_0-l(1)}{1-l(1)}V(1)
\geq V(\alpha_0)
\right\}
\right).
\label{eq:alpha_g}
\end{equation}
The minimizing property of $l(1)$ gives
\begin{equation*}
\frac{1-\alpha_e}{1-l(1)}V(l(1))
+
\frac{\alpha_e-l(1)}{1-l(1)}V(1)
\geq V(\alpha_e).
\end{equation*}
It can be verified that the function that maps $\alpha_0$ to
\begin{align*}
\frac{1-\alpha_0}{1-l(1)}V(l(1))
+
\frac{\alpha_0-l(1)}{1-l(1)}V(1)
-V(\alpha_0)
\end{align*}
is strictly convex on $(\alpha_e,1)$ and vanishes at $1$. Hence the doctor-optimal signal is feasible precisely for $\alpha_0\in(\alpha_e,\alpha_g]$. 
The inequality is strict below $\alpha_g$, and equality holds at $\alpha_g$ whenever $\alpha_g<1$. 
Moreover,
\begin{equation*}
\frac{1-\alpha_0}{1-l(1)}V(l(1))
+
\frac{\alpha_0-l(1)}{1-l(1)}V(1)
\leq\operatorname{cav}\circ V(\alpha_0),
\end{equation*}
and $\operatorname{cav}\circ V=V$ on $[\alpha_v,1]$, and therefore $\alpha_e\leq\alpha_g\leq\alpha_v$.

If $\alpha_0\geq\alpha_v$, Jensen's inequality applied to $\operatorname{cav}\circ V$ shows that no signal can give the patient more than $V(\alpha_0)$. 
Feasibility therefore requires the signal to give the patient exactly $V(\alpha_0)$, and thus Jensen's inequality must hold with equality. Any signal with mean $\alpha_0\geq\alpha_v$ that assigns positive probability below $\alpha_e$ must also assign positive probability above $\alpha_v$, and its support therefore extends beyond the unique affine segment of $\operatorname{cav}\circ V$ that crosses $\alpha_e$. Strict concavity outside this segment makes Jensen's inequality strict, contradicting feasibility. Hence no feasible signal induces treatment with positive probability. 
If $\alpha_v=1$, this affine segment connects $(l(1),V(l(1)))$ and $(1,V(1))$. The signal $\{1,l(1)\}$ therefore attains the patient-value benchmark, and hence $\alpha_g=1$ and the binding region $(\alpha_g,\alpha_v)$ is empty.

\bigskip

\noindent\textbf{Step 3:} Characterize the unique optimal posterior distribution when the testing constraint binds.

\medskip

Fix $\alpha_0\in(\alpha_g,\alpha_v)$, and then an optimal policy with a positive probability of treatment cannot have a slack testing constraint.
Otherwise, holding $y$ fixed and raising $x$ would increase the probability that a sick patient receives a treatment-inducing message until the constraint binds or $x=1$; the latter would contradict the infeasibility of the doctor-optimal benchmark. 
For a fixed $x<1$, this probability rises strictly with $y$, and thus the doctor selects the largest feasible lower posterior.

The testing constraint binds when
\begin{equation*}
\frac{V(\alpha_0)-V(y)}{\alpha_0-y}
=
\frac{V(x)-V(\alpha_0)}{x-\alpha_0}.
\end{equation*}
The right-hand side is strictly decreasing in $x$, and the numerator determining the derivative of the left-hand side with respect to $y$ is $V(\alpha_0)-V(y)-(\alpha_0-y)V'(y)$, which is strictly increasing in $y$. 
Hence, the left-hand side first decreases and then increases, and the largest solution of the displayed equality is therefore unique.
Whenever it is interior, it lies on the increasing part of the left-hand chord slope and falls continuously as $x$ rises. 
At a binding solution, the secant connecting $(y,V(y))$ and $(\alpha_0,V(\alpha_0))$ has the same slope as the secant connecting $(\alpha_0,V(\alpha_0))$ and $(x,V(x))$. 
The largest solution therefore lies on the increasing side of the chord-slope function in~\eqref{eq:patient_locus} and satisfies $l(x)\leq y<\alpha_e$, where equality can occur only at an endpoint of the feasible set of binding solutions. 
In particular, a lower-bound solution has $y=l(x)=0$.

For an interior largest solution,
\begin{equation*}
\frac{V(x)-V(y)}{x-y}>V'(y)
\qquad\text{and}\qquad
\frac{V(x)-V(y)}{x-y}>V'(x).
\end{equation*}
On each interval where $V$ is differentiable, let $\mathcal I(z)\equiv V(z)+(1-z)V'(z)$ denote the intercept at $\alpha=1$ of the tangent to $V$ at $z$. The function $\mathcal I$ is strictly decreasing on each such interval. Implicit differentiation of the binding constraint gives
\begin{equation}
\frac{d}{dx}
\left[
\frac{x-\alpha_0}{x-y}(1-y)
\right]
=
\frac{(\alpha_0-y)\bigl[\mathcal I(x)-\mathcal I(y)\bigr]}
{(x-y)^2\left(\dfrac{V(x)-V(y)}{x-y}-V'(y)\right)}.
\label{eq:coverage_derivative}
\end{equation}
As $x$ rises while the testing constraint remains binding, $y$ falls. Hence $\mathcal I(x)-\mathcal I(y)$ falls strictly and can cross zero at most once. The probability that a sick patient receives a treatment-inducing message is therefore single-peaked over the feasible set of binding solutions. This probability is zero at the degenerate lower endpoint, whereas $\alpha_0<\alpha_v$ guarantees a feasible signal for which it is positive. If the upper endpoint satisfies $l(x)>0$, tangency gives $\mathcal I(x)-\mathcal I(l(x))=(1-x)[V'(x)-V'(l(x))]<0$.
The maximizer must then occur before that endpoint. If instead the feasible set of binding solutions ends at a point with $l(x)=0$, that endpoint may maximize this probability. Together with the binding constraint displayed above, an interior maximizer satisfies $\mathcal I(x)=\mathcal I(y)$.
If the maximizer lies at the lower boundary $y=0$, then $\mathcal I(x)\geq\mathcal I(0)$.
This proves uniqueness among binary posterior distributions.

The preceding argument establishes uniqueness only among binary posterior distributions. To rule out an optimal distribution with richer support, pool all treatment posteriors and all no-treatment posteriors separately. Because $P$ is affine and $V$ is strictly concave on each region, this pooling preserves the doctor's payoff and strictly raises the patient's value unless each region already contains a single posterior. A strict increase would make the testing constraint slack. If the pooled upper posterior were below $1$, a small increase in that posterior, holding the lower posterior fixed, would preserve feasibility and increase the probability that a sick patient receives a treatment-inducing message. If the pooled upper posterior were $1$, the pooled signal would be a feasible doctor-optimal signal, contradicting $\alpha_0>\alpha_g$. Hence every optimal posterior distribution must have one posterior in each region and must coincide with the unique binary optimum characterized above.

\bigskip

\noindent\textbf{Step 4:} Show how the optimal posteriors vary with the prior and determine when the optimal signal is warning-tilted or comfort-tilted.

\medskip

For each $\alpha_0\in(\alpha_g,\alpha_v)$, let $\{x(\alpha_0),y(\alpha_0)\}$ denote the unique optimal signal characterized in Step~3. To simplify the derivative calculation below, write $x$ for $x(\alpha_0)$ and $y$ for $y(\alpha_0)$. First consider priors for which $y(\alpha_0)>0$. Since $\mathcal I(x)=\mathcal I(y)$ and $\mathcal I$ is strictly decreasing on both $[0,\alpha_e]$ and $(\alpha_e,1]$, $y$ rises strictly with $x$. After clearing denominators in the binding constraint, view
\begin{equation*}
(x-\alpha_0)V(y)
+(\alpha_0-y)V(x)
-(x-y)V(\alpha_0)
\end{equation*}
as a function of $x$, $y$, and $\alpha_0$. At a binding optimum, its three partial derivatives are
\begin{align*}
(\alpha_0-y)
\left[
V'(x)-\frac{V(x)-V(y)}{x-y}
\right]
&<0,\\
(x-\alpha_0)
\left[
V'(y)-\frac{V(x)-V(y)}{x-y}
\right]
&<0,\\
(x-y)
\left[
\frac{V(x)-V(y)}{x-y}-V'(\alpha_0)
\right]
&<0.
\end{align*}
The last inequality follows from strict concavity of $V$ on $(\alpha_e,1]$ because the binding constraint makes the slope of the secant connecting $(y,V(y))$ and $(x,V(x))$ equal to $[V(x)-V(\alpha_0)]/(x-\alpha_0)<V'(\alpha_0)$. These signs and the increase of $y$ with $x$ imply that $\alpha_0$ falls strictly as $x$ rises. When the optimal lower posterior equals zero, the same conclusion follows from the first and third displayed derivatives. Consequently, as $\alpha_0$ rises, $x(\alpha_0)$ falls strictly and $y(\alpha_0)$ falls weakly; $y(\alpha_0)$ can remain constant only after it reaches zero.

Bayes plausibility gives the upper-posterior probability $[\alpha_0-y(\alpha_0)]/[x(\alpha_0)-y(\alpha_0)]$,
and subtracting the full-disclosure probability $\alpha_0$ from the upper-posterior probability yields
\begin{equation*}
\frac{\alpha_0-y(\alpha_0)}{x(\alpha_0)-y(\alpha_0)}
-\alpha_0
=
\frac{
\alpha_0[1-x(\alpha_0)]
-(1-\alpha_0)y(\alpha_0)
}{
x(\alpha_0)-y(\alpha_0)
}.
\end{equation*}
Because $x(\alpha_0)$ is strictly decreasing and $y(\alpha_0)$ is weakly decreasing in $\alpha_0$, the numerator $\alpha_0[1-x(\alpha_0)]-(1-\alpha_0)y(\alpha_0)$ is continuous and strictly increasing in $\alpha_0$.

To determine the endpoint behavior, consider the numerator at the two endpoints. At the lower endpoint, its limit is $-(1-\alpha_g)l(1)\leq0$ when $\alpha_g>\alpha_e$ and is strictly negative when $\alpha_g=\alpha_e$ because $y\to\alpha_e<x$. Thus, the lower-end limit can equal zero only when $l(1)=0$. At the upper endpoint of a nonempty binding region,
\begin{equation*}
\lim_{\alpha_0\uparrow\alpha_v}
\bigg(
\alpha_0[1-x(\alpha_0)]
-(1-\alpha_0)y(\alpha_0)
\bigg)
=(1-\alpha_v)[\alpha_v-l(\alpha_v)]>0.
\end{equation*}
Hence there is a unique $\alpha_w\in[\alpha_g,\alpha_v]$ satisfying
\begin{equation*}
\alpha_w[1-x(\alpha_w)]
-(1-\alpha_w)y(\alpha_w)=0,
\end{equation*}
at which the diagnostic tilt changes sign. The equality $\alpha_w=\alpha_g$ occurs when the initial signal is full disclosure. 

Finally, under $\{1,l(1)\}$, the posterior $1$ occurs with probability $[\alpha_0-l(1)]/[1-l(1)]$, whereas under full disclosure it occurs with probability $\alpha_0$. Subtracting the full-disclosure probability from the probability under $\{1,l(1)\}$ gives
\begin{equation*}
\frac{\alpha_0-l(1)}{1-l(1)}-\alpha_0
=
-\frac{(1-\alpha_0)l(1)}{1-l(1)}\leq0.
\end{equation*}
The doctor-optimal signal is warning-tilted when $l(1)>0$ and is full disclosure when $l(1)=0$. This completes the proof.

\section{Physical Testing-Cost Benchmark}
\label{appendix.physical_cost_benchmark}

This appendix isolates the role of psychological information costs by replacing the anticipatory channel with an ordinary physical cost of testing.
Suppose $\phi(v)=v$ (no information avoidance) but taking the test entails a physical cost $\psi>0$.
Write $\chi\equiv\psi/(1-\alpha)$ for the testing cost in the same conditional health-probability units as $\underline{p}$ and $\bar{p}-c$.
The patient's payoff functions when $a=1$ and $a=0$ become
\begin{align*}
     \tilde{V}(\mu)=\left\{
     \begin{array}{cc}
         \alpha+(1-\alpha)(\bar{p}-c) & \mu\leq\mu_e \\
         \alpha+(1-\alpha)\underline{p}(\mu) & \mu>\mu_e
     \end{array}
     \right.
     \quad\text{and}\quad
     \tilde{V}_0(\mu)=\alpha+(1-\alpha)\underline{p}(\mu),
\end{align*}
respectively.
The interim problem becomes
\begin{align*}
     \max_{\tau_1^1,\tau_1^0\in\mathcal{R}(\mu_1)}:\,&\int_0^1P(\mu)\tau_1^1(d\mu)\\
    &\text{subject to: }\int_0^{1}\tilde{V}(\mu)\tau_1^1(d\mu)-\psi\geq\int_0^1\tilde{V}_0(\mu)\tau_1^0(d\mu).
\end{align*}

\begin{prop}[Physical testing-cost benchmark]
\label{prop:physical_cost_benchmark}
In the physical cost setting, identify
\begin{align*}
         \tilde{\mu}_{\mathcal{M}}=\frac{\bar{p}-c-\chi-\underline{p}_L}{\bar{p}-c-\underline{p}_L}
         \quad\text{and}\quad
         \tilde{\mu}_{\mathcal{F}}=\frac{\bar{p}-c-\chi-\underline{p}_L}{\underline{p}_H-\underline{p}_L}.
\end{align*}
The following statements hold.
\begin{itemize}
         \item Testing can be induced if and only if $\mu_1\leq\tilde{\mu}_{\mathcal{M}}$.
         \item The optimal interim disclosure is $\{\mu_1\}$ if $\mu_1\leq\tilde{\mu}_{\mathcal{F}}$, and is $\{1,\tilde{l}(\mu_1)\}$ if $\mu_1\in(\tilde{\mu}_{\mathcal{F}},\tilde{\mu}_{\mathcal{M}}]$, where
         \[
             \tilde{l}(\mu_1)=
             \frac{(1-\mu_1)(\bar{p}-c-\underline{p}_L)-\chi}
             {(1-\mu_1)(\underline{p}_H-\underline{p}_L)-\chi}.
         \]
         \item If $\bar{p}-c\geq \underline{p}_L+\chi$, the optimal ex ante disclosure is $\{\mu_0\}$ if $\mu_0\leq\tilde{\mu}_{\mathcal{F}}$, and is $\{1,\tilde{\mu}_{\mathcal{F}}\}$ otherwise. If $\bar{p}-c<\underline{p}_L+\chi$, the patient cannot be induced to take the test.
\end{itemize}
\end{prop}
\begin{proof}
    See Appendix \ref{proof.benchmark}.
\end{proof}

 Comparing Proposition~\ref{prop:physical_cost_benchmark} with the voluntary-testing result in Proposition~\ref{prop:voluntary_testing}, two contrasts emerge. First, the warning-in-advance structure survives: when testing can be induced, the doctor places the patient at the threshold $\tilde{\mu}_{\mathcal{F}}$ (the analogue of $\min\{\mu_e,\mu_\mathcal{F}\}$) and uses a warning-oriented interim signal whenever the testing constraint binds. 
 Second, committed comfort \emph{disappears}, and the optimal interim signal is never comfort-oriented. 
 The reason is that with $\phi$ linear, $V_0$ is itself linear and the testing constraint can always be relaxed by reducing the lower posterior under warning-oriented disclosure; comfort has value only when $\phi$ is concave so that bad news is psychologically costly. 
 Hence no information greenhouse arises in the physical-cost benchmark.

\subsection*{Proof of Proposition \ref{prop:physical_cost_benchmark}}
\phantomsection
\label{proof.benchmark}
 Let $\chi\equiv\psi/(1-\alpha)$.
 Since $\tilde{V}$ and $\tilde{V}_0$ are affine transformations of conditional health probabilities, the testing constraint can be written in conditional units by subtracting the common healthy-state term and dividing by $1-\alpha$.
 Because $\tilde V_0$ is affine, every refusal signal $\tau_1^0\in\mathcal R(\mu_1)$ satisfies
 \[
 \int_0^1\tilde V_0(\mu)\tau_1^0(d\mu)=\tilde V_0(\mu_1).
 \]
 Refusal disclosure $\tau_1^0$ is therefore payoff-irrelevant.
 Moreover, both $P$ and $\tilde V$ are affine on $[0,\mu_e]$ and $(\mu_e,1]$.
 The averaging argument in Lemma~\ref{lemma:binary_signals} then reduces any feasible post-test signal to at most one posterior in each interval without changing either payoff.
 It is therefore enough to consider binary interim signals.

 If the doctor uses no post-test information at belief $\mu_1$, treatment is accepted and testing is induced if and only if
 \[
      \bar{p}-c-\chi\geq \underline{p}(\mu_1),
 \]
 which is equivalent to $\mu_1\leq\tilde{\mu}_{\mathcal{F}}$.
 Hence no disclosure is optimal for $\mu_1\leq\tilde{\mu}_{\mathcal{F}}$.

 For a binary signal $\{x,y\}$ with $y\leq\mu_e<x$, the probability of the treatment-inducing lower posterior is $(x-\mu_1)/(x-y)$. 
 Since $\underline p(\mu)$ is affine in $\mu$ and the posterior mean is $\mu_1$, the testing constraint is equivalent to
 \[
 \frac{x-\mu_1}{x-y}
 \bigl[\bar p-c-\underline p(y)\bigr]\geq\chi,
 \]
 and the doctor's expected payoff is
 \[
 \alpha+(1-\alpha)\left[
 \underline p(\mu_1)
 +\frac{x-\mu_1}{x-y}
 \bigl[\bar p-\underline p(y)\bigr]
 \right].
 \]
 Holding $y$ fixed, $(x-\mu_1)/(x-y)$ strictly increases with $x$, and thus raising $x$ preserves feasibility and strictly raises the doctor's payoff. Hence the optimal upper posterior is $x=1$. With $x=1$, the left-hand side of the testing constraint strictly decreases with $y$, whereas the doctor's payoff strictly increases with $y$. The doctor therefore chooses the largest feasible lower posterior, and the testing constraint binds.

 For higher beliefs, the optimal binary signal keeps the favorable posterior at $1$ and lowers the other posterior until the testing constraint binds.
 Let $y$ denote this lower posterior.
 The binding testing constraint, written in conditional units, is
 \[
      \frac{1-\mu_1}{1-y}(\bar{p}-c)
      +\frac{\mu_1-y}{1-y}\underline{p}_H
      -\chi
      =
      \underline{p}(\mu_1),
 \]
 and solving for $y$ gives
 \[
      y=
      \frac{(1-\mu_1)(\bar{p}-c-\underline{p}_L)-\chi}
      {(1-\mu_1)(\underline{p}_H-\underline{p}_L)-\chi}
      \equiv \tilde{l}(\mu_1).
 \]
 The function $\tilde V$ is convex on $[0,1]$, and thus its concave envelope is the chord joining its endpoint values. Full disclosure therefore maximizes the patient's post-test payoff among all Bayes-plausible signals.
 Full disclosure can induce testing if and only if
 \[
      (1-\mu_1)(\bar{p}-c)+\mu_1\underline{p}_H-\chi
      \geq
      \underline{p}(\mu_1),
 \]
 which is equivalent to $\mu_1\leq\tilde{\mu}_{\mathcal{M}}$.
 Therefore testing can be induced if and only if $\mu_1\leq\tilde{\mu}_{\mathcal{M}}$.

 The induced continuation value is $\bar{p}$ for $\mu_1\leq\tilde{\mu}_{\mathcal{F}}$ and $\underline{p}(\mu_1)$ for $\mu_1>\tilde{\mu}_{\mathcal{M}}$.
 For $\mu_1\in[\tilde{\mu}_{\mathcal{F}},\tilde{\mu}_{\mathcal{M}}]$,
 \begin{align*}
      \bar{P}(\mu_1)
      &=\frac{1-\mu_1}{1-y}\bar{p}
      +\frac{\mu_1-y}{1-y}\underline{p}_H \\
      &=\frac{(\bar{p}-\underline{p}_H)(\underline{p}_H\mu_1+\underline{p}_L(1-\mu_1)+\chi)-\underline{p}_Hc}{\bar{p}-c-\underline{p}_H}.
 \end{align*}
 This expression is decreasing and linear in $\mu_1$, with $\bar{P}(\tilde{\mu}_{\mathcal{F}})=\bar{p}$.
 If $\bar{p}-c<\underline{p}_L+\chi$, then $\tilde{\mu}_{\mathcal{M}}<0$, and hence testing cannot be induced.
 If $\bar{p}-c\geq\underline{p}_L+\chi$, the concavification of $\bar{P}$ yields the ex ante policy stated in the proposition.
\section{General Information Preferences}
\label{section.general}

 The analysis in Sections~\ref{section.voluntary_testing} and~\ref{section.exante} assumes that $\phi$ is globally concave, which captures a global preference for late resolution of uncertainty.
 However, empirical evidence has demonstrated that the information attitudes in the real world are more complicated than that.
 For example, people who avoid information often also exhibit a preference for positively skewed signals and prior dependence (see \cite{ganguly2016fantasy}, \cite{nielsen2020preferences}, \cite{gul2020thrill} and \cite{masatlioglu2023intrinsic} for recent evidence),\footnote{A positively (negatively) skewed signal is a signal that has a potential to confirm the good (bad) state but is not diagnostic about the bad (good) state. 
 Meanwhile, prior dependence refers to the information attitude that people prefer early (late) resolution of uncertainty if and only if the prior belief is optimistic (pessimistic).} and this cannot be captured by function $\phi$ when it is globally concave.
 
 This appendix allows $\phi$ to take any continuous form and shows that the warning-versus-comfort logic survives under more general information attitudes.
 Especially, I emphasize the case that $\phi$ is twice differentiable and inversely S-shaped; that is, $\phi$ satisfies: (i) $\phi'(\underline{p}_L)>1$ and $\phi'(\underline{p}_H)>1$, and (ii) there exists an inflection point $\tilde{v}\in(\underline{p}_L,\underline{p}_H)$ such that $\phi$ is concave when $v<\tilde{v}$ and convex when $v>\tilde{v}$ (See Figure \ref{fig-sshaped}).
 This functional form deserves a peculiar highlight because it can predict the preference for positive skewness and prior-dependence.
 Figure \ref{fig-sshaped} provides two numerical examples when there are only two possible levels of date-2 payoff, $v_H$ and $v_L$. 
 
 \begin{itemize}
     \item Suppose the prior belief of $v_H$ is $\frac{1}{2}$, and compare signals $\{1,\frac{1}{3}\}$ (positively skewed) and $\{\frac{2}{3},0\}$ (negatively skewed), which are mutually symmetric and have an identical variance. As is illustrated in panel (a) of Figure \ref{fig-sshaped}, the payoff of $\{1,\frac{1}{3}\}$ ($\phi_1$) is obviously higher than the payoff of $\{\frac{2}{3},0\}$ ($\phi_2$).
     \item Compare full disclosure and non-disclosure when the prior belief is $\frac{1}{4}$ and $\frac{3}{4}$, respectively. As is illustrated in panel (b) of Figure \ref{fig-sshaped}, non-disclosure dominates full disclosure when the prior is $\frac{1}{4}$ ($\phi_{12}>\phi_{11}$), but full disclosure dominates non-disclosure when the prior is $\frac{3}{4}$ ($\phi_{22}<\phi_{21}$).\footnote{\cite{eliaz2006can} show that the linear model by \cite{caplin2001psychological} cannot predict prior dependence in a general setting. However, it is possible with such an inverse S-shaped function when the prize space contains only two outcomes.
    Indeed, to show the impossibility result in \cite{eliaz2006can}, we need the prize space to have at least three elements.}
 \end{itemize}

  \begin{figure}[ht!]
   \centering
    \begin{subfigure}[t]{0.47\textwidth}
    \begin{tikzpicture} 

    \draw[help lines,->] (0,0)--(5.5,0);
    \draw[help lines,->] (0,0)--(0,5.5);

    \draw (5.5,0) node [rectangle,below]{$v$};
    \draw (0,5.5) node [rectangle,left]{$\phi(v)$};
    \draw (2.5,0) node [rectangle,below]{$\frac{1}{2}$};
    \draw (0,0) node [rectangle,below] {$v_L$};
    \draw (5,0) node [rectangle,below] {$v_H$};
    \draw (0,5) node [rectangle,left] {$\phi(v_H)$};

    \draw[domain=0:5] plot(\x,{(5*(0.2*\x)^0.5)/((0.2*\x)^0.5+(1-0.2*\x)^0.5)});
    \draw[dashed] (5,0)--(5,5);
    \draw[dashed] (2.5,5)--(2.5,0);
    
    \draw[dashed] (2.5,2.8033)--(0,2.8033) node [left] {$\phi_1$};
    \draw[dashed] (2.5,2.1967)--(0,2.1967) node [left] {$\phi_2$};
    \draw[dashed] (5/3,5)--(5/3,0) node [below] {$\frac{1}{3}$};
    \draw[dashed] (0,5)--(5,5);
    \draw[dashed] (5/3,2.07107)--(5,5);
    \draw[dashed] (0,0)--(10/3,2.92893);
    \draw[dashed] (10/3,5)--(10/3,0) node [below] {$\frac{2}{3}$};
    \fill (0,0) circle (0.5mm);
    \fill (5/3,2.07107) circle (0.5mm);
    \fill (10/3,2.92893) circle (0.5mm);
    \fill (2.5,2.8033) circle (0.5mm);
    \fill (2.5,2.1967) circle (0.5mm);
    \end{tikzpicture} 
  \caption{Preference for skewness}
  \end{subfigure}
  \hfill
    \begin{subfigure}[t]{0.47\textwidth}   
    \begin{tikzpicture} 

    \draw[help lines,->] (0,0)--(5.5,0);
    \draw[help lines,->] (0,0)--(0,5.5);

    \draw (5.5,0) node [rectangle,below]{$v$};
    \draw (0,5.5) node [rectangle,left]{$\phi(v)$};
    \draw (2.5,0) node [rectangle,below]{$\frac{1}{2}$};
    \draw (0,0) node [rectangle,below] {$v_L$};
    \draw (5,0) node [rectangle,below] {$v_H$};
    \draw (0,5) node [rectangle,left] {$\phi(v_H)$};

    \draw[domain=0:5] plot(\x,{(5*(0.2*\x)^0.5)/((0.2*\x)^0.5+(1-0.2*\x)^0.5)});
    \draw[dashed] (5,0)--(5,5);
    \draw[dashed] (2.5,5)--(2.5,0);
    
    \draw[dashed] (1.25,1.25)--(0,1.25) node [left] {$\phi_{11}$};
    \draw[dashed] (1.25,1.83013)--(0,1.83013) node [left] {$\phi_{12}$};
    \draw[dashed] (3.75,3.75)--(0,3.75) node [left] {$\phi_{21}$};
    \draw[dashed] (3.75,3.16987)--(0,3.16987) node [left] {$\phi_{22}$};
    \draw[dashed] (1.25,5)--(1.25,0) node [below] {$\frac{1}{4}$};
    \draw[dashed] (0,5)--(5,5);
    \draw[dashed] (0,0)--(5,5);
    \draw[dashed] (3.75,5)--(3.75,0) node [below] {$\frac{3}{4}$};
    \fill (0,0) circle (0.5mm);
    \fill (1.25,1.83013) circle (0.5mm);
    \fill (1.25,1.25) circle (0.5mm);
    \fill (3.75,3.16987) circle (0.5mm);
    \fill (3.75,3.75) circle (0.5mm);
    \end{tikzpicture} 
  \caption{Prior-dependence}
  \end{subfigure}
  \caption{Inverse S-shaped utility function and anomalies.}
  \label{fig-sshaped}
 \end{figure}

 \subsection{Interim Disclosure}
 For an interim belief $\mu_1$, refusing the test gives the patient payoff $V_0(\mu)=\phi\bigl(\alpha+(1-\alpha)\underline p(\mu)\bigr)$ at posterior $\mu$. Because the doctor minimizes this refusal payoff, the relevant outside option is
 \[
 \bar V_0(\mu_1)
 \equiv \min_{\tau_1^0\in\mathcal R(\mu_1)}\int_0^1V_0(\mu)\tau_1^0(d\mu)
 =-\operatorname{cav}\circ(-V_0)(\mu_1).
 \]
 Continuity of $\phi$ guarantees that the minimum is attained, and a minimizing signal can be chosen with at most two posteriors. When $\phi$ is globally concave, $V_0$ is concave and its convex envelope is the chord joining its endpoints, and therefore full disclosure attains $\bar V_0$. Under a general $\phi$, the minimizing signal may instead have interior support. 
 Notably, since information may be both instrumentally valuable and emotionally comfortable, in the interval where non-disclosure is optimal, concealing information not only shuts down the instrumental value, but also imposes emotional cost on the patient through anxiety or by failing to satisfy curiosity.

 With only a mild modification, the method to solve the optimal interim disclosure in Proposition \ref{prop:optimal_interim_disclosure} can also be applied for any general payoff function $\phi$ without the assumption of global concavity.
 For any interval $I\subset\mathbb{R}$, define the \textit{local convex hull} of the graph of $V$ on $I$ as the set
 \begin{align*}
    co_I(V)=\bigg\{(\mu,v)\bigg|\,\mu_1,\mu_2\in I,\text{ and }\pi\in[0,1],\text{ such that }&\pi\mu_1+(1-\pi)\mu_2=\mu\\&\text{ and }\pi V(\mu_1)+(1-\pi)V(\mu_2)=v\bigg\}.
 \end{align*}
 Then the \textit{local concavification} of $V$ on $I$ is given by
 \begin{align*}
    \operatorname{cav}_I\circ V(\mu)\equiv\sup\{v\big|\,\mu\in I,\,(\mu,v)\in co_I(V)\}.
 \end{align*}
 Define
 \begin{align*}
    \hat{V}(\mu)=\left\{
     \begin{array}{cc}
        \operatorname{cav}_{[0,\mu_e]}\circ V(\mu)  & \mu\leq\mu_e \\
        \operatorname{cav}_{[\mu_e,1]}\circ V(\mu)  & \mu>\mu_e \\
     \end{array}
    \right..
 \end{align*}
 Then the following proposition provides the method to solve the optimal interim disclosure when $V$ is not concave on some subsets of $[0,1]$.

 \begin{prop}
 \label{s-shaped}
    When the doctor is persuading a patient with payoff function $V$, the optimal health probability is the same as when the payoff function is $\hat{V}$.
    Moreover, the optimal signal has a trinary support if and only if (i) the optimal interim disclosure is $\{0,\mu\}$ ($\mu>\mu_e$) when the payoff function is $\hat{V}$ and (ii) $\hat{V}(\mu)>V(\mu)$.
 \end{prop}
 The proposition states that if $V$ is not globally concave, the doctor can solve the persuasion problem \textit{as if} the patient's payoff function is replaced by the ``local concavification'' $\hat{V}$, which is piecewise concave, and thus Lemma \ref{lemma:binary_signals} holds.
  Then Proposition \ref{prop:optimal_interim_disclosure} can be applied to identify an optimal interim signal $\tau^*$ with at most two posterior beliefs.
 The new pattern emerges when there exists $\mu\in\text{Supp}(\tau^*)$ such that $\hat{V}(\mu)>V(\mu)$. 
 In this case, a signal inducing posterior belief $\mu$ with a positive probability is not directly implementable, and the optimal signal contains three posterior beliefs. 
 Figure \ref{fig-algorithm-general} provides such an example, where the shaded area is the local convex hull of $V$ when $\mu\geq\mu_e$ and the local concavification can be achieved simply by connecting points $A$ and $B$.
 Then when the interim belief is $\mu_1$, by Proposition \ref{prop:optimal_interim_disclosure}, the optimal signal is given by $\{0,h(\mu_1)\}$.
 However, at posterior belief $h(\mu_1)$ we have $\hat{V}(h(\mu_1))>V(h(\mu_1))$, and thus to achieve the payoff of the hypothetical persuasion problem, the doctor needs to further split $h(\mu_1)$ into $\mu_2$ and $1$, which is guaranteed by the Carath\'eodory theorem. 

    \begin{figure}[ht!]
     \centering
     \begin{tikzpicture}
     \draw[thick,->](0,.5)--(9,.5) node [below] {$\mu$};
     \draw[thick,->](0,.5)--(0,6.5);
     \draw[dashed,thick] (8,6)--(0,6) node [left] {$V(1)$};
     \draw[dashed,thick] (8,6)--(8,.5) node [below] {$1$};
     \draw[thick] (1,1)--(0,1) node [left] {$V(0)$};
     \node[circle,fill=black,inner sep=0pt,minimum size=3.2pt] (a) at (5.5,4.8) {};
     \node[circle,fill=black,inner sep=0pt,minimum size=3.2pt] (a) at (3,3.6) {};
     \draw[thick, dashed] (3,3.6)--(8,6);
     \draw[thick,domain=1:3.5] plot(\x,{-(41/100)*\x^2+(147/50)*\x-153/100});
     \draw[thick,domain=3.5:8] plot(\x,{(779/8100)*\x^2-(2443/4050)*\x+9498/2025});
     \draw[thick, densely dotted] (0,1)--(5.5,4.8);
     \node[circle,fill=black,inner sep=0pt,minimum size=3.2pt] (a) at (4.5,226/55) {};
     \draw[thick, dashed] (8,226/55)--(0,226/55) node [left] {$\bar{V}_0(\mu_1)$};
     \draw[thick, dashed] (4.5,6)--(4.5,0.5);
     \draw[thick, dashed] (5.5,6)--(5.5,0.5);
     \draw[thick, dashed] (3,6)--(3,0.5);
     \draw[thick, dashed] (1,6)--(1,0.5) 
        node [above] at (1,0) {$\mu_e$}
        node [above] at (3,0) {$\mu_2$}
        node [above] at (2.8,3.6) {$A$}
        node [right] at (8,6) {$B$}
        node [above] at (5.5,-0.05) {$h(\mu_1)$}
        node [above] at (4.5,0) {$\mu_1$};
     \fill[gray, opacity=.5, variable=\x]  plot[domain=1:3] (\x,{-(41/100)*\x^2+(147/50)*\x-153/100}) -- (8,6) -- plot[domain=8:5.79228](\x,{(779/8100)*\x^2-(2443/4050)*\x+9498/2025})
     -- cycle;
    \end{tikzpicture}  
    \caption{\textup{Constructing the optimal signal when $V$ is not globally concave.}}
    \label{fig-algorithm-general}
 \end{figure}
 
 The existence of the third posterior belief echoes the earlier finding in constrained Bayesian persuasion \citep{le2019persuasion,boleslavsky2020bayesian,doval2024constrained}. 
 Here, the additional piece of information has a straightforward interpretation when the intrinsic preference for information is present. 
 When $\phi$ is not globally concave, there exist some interim beliefs at which the patient is \textit{intrinsic information loving}, and belief $\mu$ must be such a belief if $V(\mu)<\hat{V}(\mu)$.
 Since it is costless for the doctor to disclose information when the posteriors are restricted to be larger than $\mu_e$, she can use this additional information as a ``psychological prize'', which can reduce anxiety or fulfill curiosity even though it is instrumentally useless, to reward the decision of accepting the test.

 On a more general level, \cite{doval2024constrained} identifies the conditions under which the number of posteriors in the optimal signal is equal to the number of states, including (i) the set of beliefs for each action to be optimal is connected, and (ii) the persuader's value function is piecewise concave.
 Since local concavification guarantees the second requirement, if the first requirement is also satisfied (as is in the current setting), the local concavification-based approach reduces the constrained Bayesian persuasion problem to an algebra problem with a fixed dimensionality. 

 \subsection{Ex Ante Disclosure}
 Recursing back to the ex ante disclosure, we can observe that although the optimal ex ante signal can be simply derived by backward induction, the exact solution still varies with the specific functional form of $\phi$.
 The following proposition provides the solution when $\phi$ is inverse S-shaped.

 \begin{prop}
 \label{proposition-general}
  Suppose $\phi$ is inverse S-shaped.
  \begin{itemize}
   \item If the patient is fear-reactive, $\mathcal{M}=\mathcal{D}=[0,1]$, and therefore the optimal interim disclosure is either non-disclosure or warning-oriented for all $\mu_1\in[0,1]$. Also, the optimal ex ante disclosure is non-disclosure for all $\mu_0\in[0,1]$.
   \item If not, there exists $\mu'_\mathcal{M}\leq\mu'_\mathcal{D}\leq\mu'_\mathcal{F}$, such that $\mathcal{M}=[\mu'_\mathcal{M},1]$, $\mathcal{D}=[\mu'_\mathcal{D},1]$ and $\mathcal{F}=[\mu'_\mathcal{F},1]$. Also, there exists $\mu'_v\in[\mu'_\mathcal{M},\mu'_\mathcal{D}]$, such that the optimal ex ante signal is $\{0,\mu'_v\}$ if $\mu_0<\mu'_v$, and is non-disclosure otherwise.
  \end{itemize}
 \end{prop}
 
 The concavification results are given in Figure \ref{fig-exante-general}.
 Compared to Proposition \ref{prop:voluntary_testing} that the patient is globally information-avoidant, the most essential difference here is that the motivation availability problem is solved when $\mu$ is large. 
 That is, when the patient is optimistic about the untreated prospect $\underline{p}$, he will nonetheless be willing to be tested because of the intrinsic preference for more information.
 In particular, the inverse S-shaped utility function makes warning-oriented signals become psychologically attractive to the patient, which are unattractive to the patient when $\phi$ is concave.
 
 \begin{figure}[ht!]
   \centering
    \begin{subfigure}[t]{0.47\textwidth}
   \begin{tikzpicture}
   \draw[thick,->](0,0)--(4.5,0) node [right] {$\mu$};
   \draw[thick,->](0,0)--(0,4.5) node [left] {$\bar{P}(\mu)$};

   \draw[thick,-] (0,4) -- (0.5,4);
   \draw[domain=0.5:3,thick] plot(\x,{-(4/25)*\x^2-1/25*\x+203/50});
   \draw[thick,-] (3,2.5) -- (4,1);
   
   \fill[gray, opacity=.5, domain=0.5:3, variable=\x]
    (0,4) -- (1,4) -- plot(\x,{-(4/25)*\x^2-1/25*\x+203/50}) -- (4,1) -- cycle;
   
   \draw[thick, -, dashed](0,4)--(4,4)
    node [left] at (0,4) {$\bar{p}$}
    node [left] at (0,0) {$0$}
    node [left] at (0,1) {$\underline{p}_H$}
    node [below] at (4,0) {$1$};
   
   \draw[thick, dashed, -](0.5,4)--(0.5,0) node [below] {$\mu_\mathcal{F}^L$};
   \draw[thick, dashed, -](3,4)--(3,0) node [below] {$\mu_\mathcal{F}^H$};

   \draw[thick, dashed, -](0,1)--(4,1);
   \draw[thick, dashed, -](4,0)--(4,4);
  \end{tikzpicture}
  \caption{(\ref{eq:fear_reactive}) holds}
  \end{subfigure}
  \hfill
    \begin{subfigure}[t]{0.47\textwidth}
   \begin{tikzpicture}
   \draw[thick,->](0,0)--(4.5,0) node [right] {$\mu$};
   \draw[thick,->](0,0)--(0,4.5) node [left] {$\bar{P}(\mu)$};
   \draw[thick, -] (0,0.5) -- (1,0.625);
   \draw[thick, -] (3,2.5) -- (4,1);
   \draw[thick, -, dashed] (0,1) -- (4,1);
   \draw[thick, -, dashed] (0,0.5) --(1,3.5);
   \draw[domain=1:2,thick] plot(\x,{-.4*(\x-1.5)^2+3.6});
   \draw[domain=2:3,thick] plot(\x,{-11/20*\x^2+7*\x/4+2.2});
   \fill[gray, opacity=.5, variable=\x]
    (0,0.5) -- (1,3.5) -- plot[domain=1:2](\x,{-.4*(\x-1.5)^2+3.6}) -- plot[domain=2:3](\x,{-11/20*\x^2+7*\x/4+2.2}) -- (4,1) -- cycle;
   
   \draw[thick, -, dashed](0,4)--(4,4)
    node [left] at (0,4) {$\bar{p}$}
    node [left] at (0,0) {$0$}
    node [left] at (0,.5) {$\underline{p}_L$}
    node [left] at (0,1) {$\underline{p}_H$}
    node [below] at (4,0) {$1$};
   
   \draw[thick, dashed, -](1,4)--(1,0) node [below] {$\mu'_\mathcal{M}$};
   \draw[thick, dashed, -](2,4)--(2,0) node [below] {$\mu'_\mathcal{D}$};
   \draw[thick, dashed, -](3,4)--(3,0) node [below] {$\mu'_\mathcal{F}$};
   \draw[thick, dashed, -](1,0.625)--(4,1);
   \draw[thick, dashed, -](0,3)--(4,3);
   \draw[thick, dashed, -](4,0)--(4,4);
  \end{tikzpicture}
  \caption{(\ref{eq:fear_reactive}) does not hold}
  \end{subfigure}
  \caption{Optimal ex ante disclosure when $\phi$ is inverse S-shaped.}
  \label{fig-exante-general}
 \end{figure}

 This difference has several implications.
 When the patient is fear-reactive, the signal that maximizes the willingness to be treated already supplies enough anticipatory utility to make the patient willing to be tested.
 Consequently, $\mathcal{F}$ may not be a connected interval any more.
 Also, when the patient is fear-reactive, $\mathcal{D}=[0,1]$, which implies that full disclosure is always sufficient to guarantee voluntary participation.
 Therefore, even when $\mu_1\notin\mathcal{F}$, the doctor can satisfy the testing constraint with a warning-oriented interim signal, and a comfort-oriented signal is never optimal.
 Thus, the optimal policy uses no ex ante disclosure, and any informative disclosure occurs only at the interim stage.

 When the patient is fear-avoidant,  $[\mu_\mathcal{M}^H,1]$ in Proposition \ref{prop:voluntary_testing}, an interval containing large beliefs that are not motivation available, disappears due to the convexity of $\phi$.
 Thus, the doctor benefits from the ex ante disclosure only if the patient is pessimistic about $\underline{p}$, a circumstance in which the patient feels fear but does not want to do anything against it.\footnote{Also, both comfort-oriented disclosure and information used as an instrumental carrot --- i.e., an interim signal whose role is solely to raise the patient's expected payoff from being tested rather than to discipline treatment --- can emerge in the interim disclosure.}
\subsubsection*{Proof of Proposition \ref{s-shaped}}

 Consider the following optimization problem:
 \begin{equation}
 \label{optimization2}
 \begin{aligned}
  \max_{\tau\in\mathcal{R}(\mu_1)}:\,&\int_0^1P(\mu)\tau(d\mu)\\
    &\text{subject to: }\int_0^{1}\hat{V}(\mu)\tau(d\mu)\geq\bar{V}_0(\mu_1).
 \end{aligned}
 \end{equation}
 Here, since $\hat{V}(\mu)\geq V(\mu)$ for all $\mu\in[0,1]$, the optimal value of (\ref{optimization2}) weakly exceeds the optimal value of (\ref{eq:interim_problem}).
 
 Then it is sufficient to show that the solution of (\ref{optimization2}) can be achieved by a signal that is feasible in persuading the patient with payoff function $V$.
 Let $\mu_H$ and $\mu_L$ be the upper and lower beliefs in the signal\footnote{Throughout this proposition, binary and trinary refer to the minimum support cardinality among all optimal signals, equal to two and three, respectively. Hence redundant posterior refinements do not change the classification.} of problem (\ref{optimization2}) ($\mu_H>\mu_e\geq\mu_L$), and $\pi=(\mu_H-\mu_1)/(\mu_H-\mu_L)$ is the probability of posterior $\mu_L$.
 Note that we always have $V(\mu_L)=\hat{V}(\mu_L)$.
 If $V(\mu_H)=\hat{V}(\mu_H)$, $\{\mu_H,\mu_L\}$ is also an optimum of program (\ref{eq:interim_problem}). 
 If $V(\mu_H)<\hat{V}(\mu_H)$, by the definition of local concavification and the Carath\'eodory theorem, there exists two posteriors $\mu_H^1$ and $\mu_H^2$, both larger than $\mu_e$, and a probability $\rho$, such that $\rho\mu_H^2+(1-\rho)\mu_H^1=\mu_H$ and $\rho V(\mu_H^2)+(1-\rho)V(\mu_H^1)=\hat{V}(\mu_H)$. 
 That is, lottery of posterior belief $(\mu_L,\pi;\mu_H^2,(1-\pi)\rho;\mu_H^1,(1-\pi)(1-\rho))$ is feasible for program (\ref{eq:interim_problem}). 
 This refinement preserves the overall posterior mean and the patient's expected payoff, and therefore preserves both Bayesian plausibility and the testing constraint.
 By the piecewise linearity of $P$, it generates the same health probability as in (\ref{optimization2}). Hence the two programs have the same optimal value.

 Having established that the two problems have the same optimal value, it remains to characterize when the original problem requires three posteriors rather than at most two. Suppose the proxy optimum is the nondegenerate comfort-oriented signal $\{0,\mu\}$ and $\hat V(\mu)>V(\mu)$. Because its testing constraint binds and $\hat V(0)=V(0)$, this signal is infeasible under $V$. Any optimal signal under $V$ with at most two posteriors would also be proxy-optimal and, by the strict best-good-news argument in Steps 3--4 of the proof of Proposition~\ref{prop:voluntary_testing}, must have support $\{0,\mu\}$, a contradiction. Hence at least three posteriors are required, while Carath\'eodory's theorem shows that splitting $\mu$ into two posteriors makes three sufficient. Conversely, if $\hat V(\mu)=V(\mu)$ or the proxy optimum is not comfort-oriented, its support lies at contact points of $V$ and $\hat V$, and hence a proxy optimum with at most two posteriors is directly feasible under $V$. This proves the stated equivalence.
 
\subsubsection*{Proof of Proposition \ref{proposition-general}}

 \noindent\textit{Part 1: The patient is fear-reactive.}

 In this case, we have $\hat{V}(0)=V(0)\geq\bar{V}_0(0)$. 
 Also, since $\phi$ is inverse S-shaped, it is convex at $\underline{p}_H$; that is,
 \begin{align*}
  \hat{V}(1)=V(1)=\alpha+(1-\alpha)\phi(\underline{p}_H)\geq\phi(\alpha+(1-\alpha)\underline{p}_H)=\bar{V}_0(1).
 \end{align*}
 By the convexity of $\bar{V}_0(\cdot)$, $\mathcal{D}=[0,1]$.
 Moreover, $\mathcal{F}$ is non-empty at the two extremes ($\mu_1=0$ and $\mu_1=1$).
 Panel (a) of Figure \ref{fig-general-proof} provides a graphic illustration.
 Consequently, by Propositions \ref{s-shaped} and \ref{prop:optimal_interim_disclosure}, the interim signal must be either non-disclosure or warning-oriented.
 
 If we also have $\mathcal{F}=[0,1]$, then the interim disclosure must be doctor-optimal.
 The doctor's payoff is
 \[
 \bar P(\mu_1)=
 \begin{cases}
 \bar p, & \mu_1\leq\mu_e,\\[4pt]
 \dfrac{1-\mu_1}{1-\mu_e}\bar p
 +\dfrac{\mu_1-\mu_e}{1-\mu_e}\underline p_H,
 & \mu_1>\mu_e.
 \end{cases}
 \]
 This function is piecewise linear and concave, and therefore the optimal ex ante disclosure is non-disclosure.

 If $\mathcal{D}\setminus\mathcal{F}$ is non-empty, since the patient's payoff under the doctor-optimal signal is piecewise linear and $\bar{V}_0(\cdot)$ is convex, there exists $\mu_\mathcal{F}^L\in[0,\mu_e]$ and $\mu_\mathcal{F}^H\in[\mu_e,1]$, such that 
 $\mathcal{D}\setminus\mathcal{F}=[\mu_\mathcal{F}^L,\mu_\mathcal{F}^H]$.
 In this interval, the optimal interim disclosure is warning-oriented, and the lower belief is given by (\ref{eq:l_mu}).
 Thus, 
 \begin{align*}
     \bar{P}(\mu_1)=\frac{1-\mu_1}{1-l(\mu_1)}\cdot\bar{p}+\frac{\mu_1-l(\mu_1)}{1-l(\mu_1)}\cdot\underline{p}_H=\frac{\hat{V}(1)-\bar{V}_0(\mu_1)}{\hat{V}(1)-\hat{V}(0)}\cdot\bar{p}+\frac{\bar{V}_0(\mu_1)-\hat{V}(0)}{\hat{V}(1)-\hat{V}(0)}\cdot\underline{p}_H.
 \end{align*}
 The first and second derivatives of $\bar{P}$ with respect to $\mu_1$ are given respectively by
 \begin{align*}
  \bar{P}'(\mu_1)=-\frac{\bar{p}-\underline{p}_H}{\hat{V}(1)-\hat{V}(0)}\bar{V}_0'(\mu_1)<0\text{ and }\bar{P}''(\mu_1)=-\frac{\bar{p}-\underline{p}_H}{\hat{V}(1)-\hat{V}(0)}\bar{V}_0''(\mu_1)\leq 0.
 \end{align*}

 It can be straightforwardly verified that $\bar{P}$ is continuous.
 Also, $\bar{P}$ is constant in $[0,\mu_\mathcal{F}^L]$ and is decreasing and linear in $[\mu_\mathcal{F}^H,1]$.
 So it suffices to show that 
 $$\lim_{\mu\to(\mu_\mathcal{F}^H)^-}\bar{P}'(\mu)\geq \lim_{\mu\to(\mu_\mathcal{F}^H)^+}\bar{P}'(\mu).$$
 Indeed, at point $\mu_\mathcal{F}^H$, we have $\bar{V}_0'(\mu_\mathcal{F}^H)\leq(\hat{V}(1)-\hat{V}(0))/(1-\mu_e)$.
 Thus, the one-sided derivatives of $\bar{P}$ at $\mu_\mathcal{F}^H$ satisfy
 \begin{align*}
    \lim_{\mu\to(\mu_\mathcal{F}^H)^-}\bar{P}'(\mu)=-\frac{\bar{p}-\underline{p}_H}{\hat{V}(1)-\hat{V}(0)}\bar{V}_0'(\mu_\mathcal{F}^H)\geq-\frac{\bar{p}-\underline{p}_H}{1-\mu_e}=\lim_{\mu\to(\mu_\mathcal{F}^H)^+}\bar{P}'(\mu),
 \end{align*}
 which implies that $\bar{P}$ is globally concave on $[0,1]$.

   \begin{figure}[ht!]
   \centering
    \begin{subfigure}[t]{0.47\textwidth}
   \begin{tikzpicture}
   \draw[thick,->](0,0)--(5.5,0) node [below] {$\mu$};
   \draw[thick,->](0,0)--(0,4.5) node [left] {$V(\mu)$};
   
   \draw[domain=2:5,thick] plot(\x,{-(1/12)*\x^2+13/12*\x+2/3});
   \draw[thick] (0,1) -- (2,2.5);
   \draw[thick,dashed] (0,1) -- (2,1);
   \draw[thick, dashed] (2,1) -- (5,4);
   \draw[thick, dashed] (0,1) -- (5,4);

   \draw[domain=4:5,thick] plot(\x,{(1/2)*\x^2-7/2*\x+17/2});
   \draw[thick] (0,0.5) -- (4,2.5);
   
   \draw[thick, -, dashed](0,4)--(5,4)
    node [left] at (0,4) {$\bar{p}$}
    node [left] at (0,0) {$0$}
    node at (4,2) {$\bar{V}_0(\mu)$}
    node at (2,3) {$\operatorname{cav}\circ\hat{V}(\mu)$}
    node[above] at (0.6,1) {$\hat{V}(\mu)$}
    node [left] at (0,1.5) {$\underline{p}_H$}
    node [below] at (5,0) {$1$};
   
   \draw[thick, dashed, -](1,4)--(1,0) node [below] {$\mu^L_\mathcal{F}$};
   \draw[thick, dashed, -](3,4)--(3,0) node [below] {$\mu^H_\mathcal{F}$};

   \draw[thick, dashed, -](5,0)--(5,4);
  \end{tikzpicture}
  \caption{(\ref{eq:fear_reactive}) holds}
  \end{subfigure}
 \hfill
    \begin{subfigure}[t]{0.47\textwidth}
   \begin{tikzpicture}
   \draw[thick,->](0,0)--(5.5,0) node [below] {$\mu$};
   \draw[thick,->](0,0)--(0,4.5) node [left] {$V(\mu)$};
   
   \draw[domain=2:5,thick] plot(\x,{-(1/12)*\x^2+13/12*\x+2/3});
   \draw[thick] (0,1) -- (2,2.5);
   \draw[thick,dashed] (0,1) -- (2,1);
   \draw[thick, dashed] (2,1) -- (5,4);
   \draw[thick, dashed] (0,1) -- (5,4);

   \draw[domain=4:5,thick] plot(\x,{(1/3)*\x^2-7/3*\x+41/6});
   \draw[thick] (0,1.5) -- (4,17/6);
   
   \draw[thick, -, dashed](0,4)--(5,4)
    node [left] at (0,4) {$\bar{p}$}
    node [left] at (0,0) {$0$}
    node at (4.5,2.5) {$\bar{V}_0(\mu)$}
    node at (2.5,3.3) {$\operatorname{cav}\circ\hat{V}(\mu)$}
    node[above] at (0.6,1) {$\hat{V}(\mu)$}
    node [left] at (0,1.5) {$\underline{p}_H$}
    node [below] at (5,0) {$1$};
   
   \draw[thick, dashed, -](1.2,4)--(1.2,0) node [below] {$\mu'_\mathcal{M}$};
   \draw[thick, dashed, -](1.875,4)--(1.875,0) node [below] {$\mu'_\mathcal{D}$};
   \draw[thick, dashed, -](3.75,4)--(3.75,0) node [below] {$\mu'_\mathcal{F}$};

   \draw[thick, dashed, -](5,0)--(5,4);
  \end{tikzpicture}
  \caption{(\ref{eq:fear_reactive}) does not hold}
  \end{subfigure}
  \caption{Structure of $\mathcal{M}$, $\mathcal{D}$ and $\mathcal{F}$. The upper solid curve is $\operatorname{cav}\circ\hat V$. The dashed piecewise curve that is flat on the left is $\hat V$.}
  \label{fig-general-proof}
 \end{figure}

 \bigskip
 \noindent\textit{Part 2: The patient is fear-avoidant.}

 In this case, we have $V(1)=\hat{V}(1)\geq\bar{V}_0(1)$ but $V(0)=\hat{V}(0)<\bar{V}_0(0)$.
 Then the structure of $\mathcal{M}$, $\mathcal{D}$ and $\mathcal{F}$ can be derived by the same reasoning in Lemma \ref{lemma:three_sets}. 
 See panel (b) of Figure \ref{fig-general-proof} for a graphic illustration.
 Here, $\mu'_\mathcal{D}=\mu'_\mathcal{M}$ means that full disclosure is patient-optimal, which happens if and only if $\hat{V}(1)-\hat{V}(0)\leq\hat{V}'(1)$.

 Part 1 already establishes that $\bar P$ is concave on $\mathcal D$. On $\mathcal M\setminus\mathcal D$, apply (\ref{eq.pbn''}) with $V$ replaced by $\hat V$. Since $\hat V''\leq0$ and $\bar V_0''\geq0$, $\bar P$ is concave on every smooth part of this interval. At any interior kink, equation (\ref{eq.h'}) continues to hold with one-sided derivatives. Concavity of $\hat V$ and convexity of $\bar V_0$ imply that $h'_+(\mu)\leq h'_-(\mu)$ at every interior kink of $h$. Since
 \[
 \bar P(\mu)=\bar p+(\underline p_H-\underline p_L)\mu-(\bar p-\underline p_L)\frac{\mu}{h(\mu)},
 \]
 it follows that
 \[
 \bar P'_+(\mu)-\bar P'_-(\mu)
 =(\bar p-\underline p_L)\frac{\mu}{h(\mu)^2}
 \bigl[h'_+(\mu)-h'_-(\mu)\bigr]\leq0.
 \]
 Thus, the slope of $\bar P$ cannot jump upward at an interior kink. It remains to compare the one-sided derivatives at $\mu'_{\mathcal D}$.

 On the right, the warning-oriented policy gives
 \begin{align*}
  \lim_{\mu\rightarrow^+\mu_{\mathcal{D}}'}\bar{P}'(\mu)=-\frac{\bar{p}-\underline{p}_H}{\hat{V}(1)-\hat{V}(0)}\cdot\bar{V}_0'(\mu_{\mathcal{D}}').
 \end{align*}
 On the left, $h(\mu_{\mathcal{D}}')=1$. Applying (\ref{eq.h'}) with $V$ replaced by $\hat V$ gives
 \begin{align*}
  \lim_{\mu\rightarrow^-\mu_{\mathcal{D}}'}\bar{P}'(\mu)=\frac{\underline{p}'(1) \left(\bar{V}_0'(\mu_{\mathcal{D}}')+\hat{V}(1)-\hat{V}(0)\right)+(\bar{p}-\underline{p}_H) \left(\bar{V}_0'(\mu_{\mathcal{D}}')-\hat V'(1)\right)}{\hat{V}'(1)-\hat{V}(1)+\hat{V}(0)}.
 \end{align*}
 The difference is
 \begin{align*}
    \lim_{\mu\rightarrow^-\mu_{\mathcal{D}}'}\bar{P}'(\mu)-\lim_{\mu\rightarrow^+\mu_{\mathcal{D}}'}\bar{P}'(\mu)=&\frac{(\bar{p}-\underline{p}_H)\hat{V}'(1)+\underline{p}'(1)(\hat{V}(1)-\hat{V}(0))}{\hat{V}(1)-\hat{V}(0)}\cdot\frac{\hat{V}(1)-\hat{V}(0)-\bar{V}_0'(\mu_{\mathcal{D}}')}{\hat{V}(1)-\hat{V}(0)-\hat{V}'(1)}\\
    \propto&\frac{\hat{V}(1)-\hat{V}(0)-\bar{V}_0'(\mu_{\mathcal{D}}')}{\hat{V}(1)-\hat{V}(0)-\hat{V}'(1)}.
 \end{align*}
 Since $\mathcal{M}\setminus\mathcal{D}$ is non-empty, $\hat{V}(1)-\hat{V}(0)-\hat{V}'(1)>0$. At $\mu_{\mathcal{D}}'$, $\bar{V}_0'(\mu_{\mathcal{D}}')<\hat{V}(1)-\hat{V}(0)$. Consequently,
 \[
 \lim_{\mu\rightarrow^-\mu_{\mathcal{D}}'}\bar{P}'(\mu)
 >
 \lim_{\mu\rightarrow^+\mu_{\mathcal{D}}'}\bar{P}'(\mu),
 \]
 which implies that $\bar P$ is concave on $\mathcal M$.
 The remaining part follows directly by concavification.


 \section{Optimal Disclosure of \texorpdfstring{$\bar{p}$}{bar p}}
 \label{appendix.barp_disclosure}
 This appendix studies disclosure about the true value of $\bar p$. 
 The variable $\bar p$ takes values in $\{q_H,q_L\}$, where $q_H>q_L>\underline p$, and is independent of $\theta$. 
 The belief that $\bar p=q_H$ is denoted by $\nu$, with $\nu_0$ and $\nu_1$ denoting the prior and interim beliefs, respectively. The posterior mean of $\bar p$ is $q(\nu)\equiv\nu q_H+(1-\nu)q_L$.
 When the ex post belief is $\nu$, the patient accepts the treatment if and only if
 \begin{align*}
     \nu\geq\nu_e\equiv q^{-1}(\underline{p}+c).
 \end{align*}
 The analysis focuses on the nontrivial case $\nu_e\in(0,1)$.
 Thus, the payoff functions of the doctor and the patient are given respectively by
 \begin{align*}
     P(\nu)=\left\{
        \begin{array}{cc}
         q(\nu) & \nu\geq\nu_e,\\
        \underline{p} & \nu<\nu_e,
        \end{array}
    \right.\text{ and}\quad
    V(\nu)=\left\{
        \begin{array}{cc}
         \alpha+(1-\alpha)\phi(q(\nu)-c) & \nu\geq\nu_e,\\
         \alpha+(1-\alpha)\phi(\underline{p}) & \nu<\nu_e.
        \end{array}
        \right.
 \end{align*}
 The payoff from refusing the test is the constant $\phi(\alpha+(1-\alpha)\underline p)$.
 
 \subsection{Interim Disclosure}

 In a similar vein, the optimal signals for guiding the future and rewarding the past are derived by concavifying $P$ and $V$, respectively, and are illustrated in Figure \ref{fig-barp}.
 Unlike the case of $\underline{p}$, here both signals are comfort-oriented. 
 To guide the future, the optimal signal is $\{0,\nu_e\}$ if $\nu_1<\nu_e$ and is non-disclosure otherwise. 
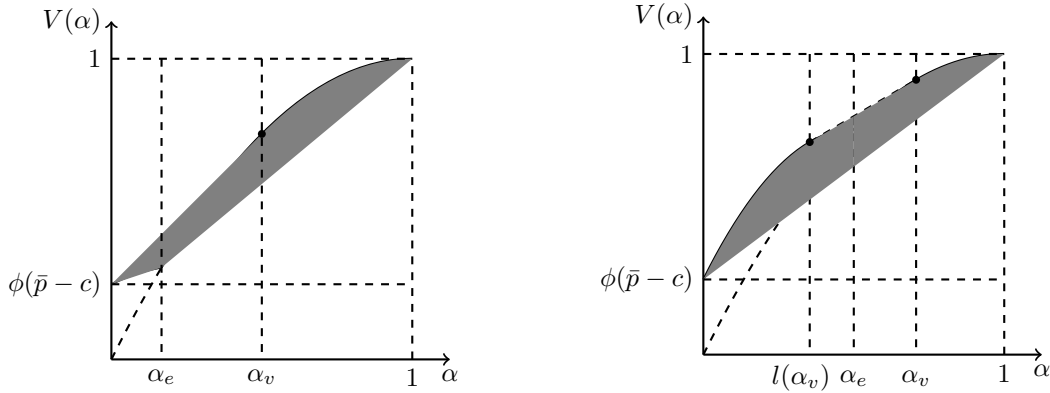
 The smallest maximizer of $(V(\nu)-V(0))/\nu$ over $[\nu_e,1]$ is denoted by $\nu_v$. The optimal signal for rewarding the past is $\{0,\nu_v\}$ when $\nu_1<\nu_v$, and is non-disclosure otherwise.

 \begin{figure}[ht!]
   \centering
   \hspace{-2cm}
   \begin{subfigure}[ht]{.4\textwidth}
   \begin{tikzpicture}
   \draw[thick,->](0,0)--(4.5,0) node [below] {$\nu$};
   \draw[thick,->](0,0)--(0,4.5) node [left] {$P(\nu)$};
   \draw[thick, dashed, -](3.0189705450977025,4)--(3.0189705450977025,0) node [below] {$\nu_v$};
   \draw[thick, -, dashed](0,4)--(4,4);
   \draw[thick, -, dashed](0,2)--(4,2);
   \draw[thick, -](1,1)--(0,1)
    node [left] at (0,4) {$q_H$}
    node [left] at (0,0) {$0$}
    node [left] at (0,1) {$\underline{p}$}
    node [left] at (0,2) {$q_L$}
    node [below] at (1,0) {$\nu_e$}
    node [below] at (4,0) {$1$};
   \draw[thick, -](1,2.5) -- (4,4);
   \draw[thick, -, dashed] (1,2.5) -- (0,2);
   \draw[thick, dashed, -](0,1)--(4,1);
   \draw[thick, dashed, -](1,0)--(1,4);
   \draw[thick, dashed, -](4,0)--(4,4);
   \begin{scope}
      \clip(0,0)rectangle(5,5);
      \fill[gray, opacity=.5] (0,1) -- (1,2.5) -- (4,4) -- (1,1) -- cycle;
   \end{scope}
  \end{tikzpicture}
  \caption{Guiding the future.}
  \end{subfigure}
  \begin{subfigure}[ht]{0.4\textwidth}
   \begin{tikzpicture}
   \draw[thick,->](0,0)--(4.5,0) node [below] {$\nu$};
   \draw[thick,->](0,0)--(0,4.5) node [left] {$V(\nu)$};
   \draw[domain=1:4,thick] plot(\x,{(1 - exp(- 1.5 * (\x / 4)))/(1 - exp(-1.5))*4});
   \draw[domain=0:1,dashed] plot(\x,{(1 - exp(- 1.5 * (\x / 4)))/(1 - exp(-1.5))*4});
   \fill[gray, opacity=.5, domain=3.0189705450977025:4, variable=\x]
    (0,1.61011) -- (3.0189705450977025,3.48912) -- plot(\x,{(1 - exp(- 1.5 * (\x / 4)))/(1 - exp(-1.5))*4}) -- (1,1.61011) -- cycle;
   \draw[thick, -] (0,1.61011)--(1,1.61011);
   \draw[thick, dashed, -](0,4)--(4,4)
    node [left] at (0,4) {$V(1)$}
    node [left] at (0,0) {$0$}
    node [left] at (0,1.61011) {$V(0)$}
    node [below] at (3.0189705450977025,0) {$\nu_v$}
    node [above] at (1,4) {$\nu_e$}
    node [below] at (4,0) {$1$};
   \draw[thick, dashed, -](3.0189705450977025,0)--(3.0189705450977025,4);
   \draw[thick, dashed, -](0,1.61011)--(3.0189705450977025,3.48912);
   \draw[thick, dashed, -](0.626429,4)--(0.626429,0) node [below] {$\nu_\mathcal{M}$};
   \draw[thick, dashed, -](1.31129,4)--(1.31129,0) node [below] {$\nu_\mathcal{F}$}; 
   \draw[thick, dashed, -](0,4)--(4,4);
   \draw[thick, dashed, -](4,2)--(0,2);
   \node[above] at (2,2) {\scriptsize$\phi(\alpha+(1-\alpha)\underline p)$};
   \draw[thick, dashed, -](1,0)--(1,4);
   \draw[thick, dashed, -](4,0)--(4,4);
   \node[circle,fill=black,inner sep=0pt,minimum size=3pt] (a) at (3.0189705450977025,3.48912) {};
  \end{tikzpicture}
  \caption{Rewarding the past.}
  \end{subfigure}
  \caption{Optimal interim disclosure when $\nu_\mathcal{M}<\nu_\mathcal{F}$.}
  \label{fig-barp}
 \end{figure}

 If $a=0$, the patient's payoff is $\phi(\alpha+(1-\alpha)\underline p)$, which is a constant with respect to $\nu$, and therefore disclosure after refusal is payoff-irrelevant.
 Thus, the optimal interim disclosure is given by the following optimization problem:
 \begin{align*}
     \max_{\tau\in\mathcal{R}(\nu_1)}:&\int_0^1P(z)\tau(dz)\\
     &\text{subject to: }\int_0^1V(z)\tau(dz)\geq \phi(\alpha+(1-\alpha)\underline p).
 \end{align*}
 As illustrated in panel (b) of Figure~\ref{fig-barp}, there is a cutoff $\nu_\mathcal{F}>\nu_e$ such that non-disclosure induces testing if and only if $\nu_1\geq\nu_\mathcal{F}$. There is also a cutoff $\nu_\mathcal{M}\leq\nu_\mathcal{F}$ such that no signal can induce testing when $\nu_1<\nu_\mathcal{M}$. Hence disclosure affects whether testing occurs only when $\nu_1\in(\nu_\mathcal{M},\nu_\mathcal{F})$. The remaining analysis considers the nonempty case $\nu_\mathcal{M}<\nu_\mathcal{F}$.
 
 \begin{prop}
     \label{prop:barp_interim}
     If $\nu_1\in(\nu_\mathcal{M},\nu_\mathcal{F})$, the optimal interim disclosure is comfort-oriented; that is, there exists $h(\nu_1)\in[\nu_\mathcal{F},\nu_v]$ such that the optimal interim disclosure is $\{0,h(\nu_1)\}$.
     Also, $h(\nu_1)$ is given implicitly by equation
     \begin{equation}
         \label{eq.barp}
        \phi(\alpha+(1-\alpha)\underline p)=\left(1-\frac{\nu_1}{h(\nu_1)}\right)\cdot V(0)+\frac{\nu_1}{h(\nu_1)}\cdot V(h(\nu_1)).
     \end{equation}
 \end{prop}

 \subsection{Ex Ante Disclosure}

Since backward induction is still valid, it remains to specify the optimal ex ante disclosure by concavifying function $\bar{P}(\nu_1)$, the health probability under the optimal interim disclosure at belief $\nu_1$.
 The following proposition characterizes \textit{comfort-in-advance}, a comfort-oriented signal before the medical test followed by no further disclosure.

 \begin{prop}
 \label{proposition.transparency}
     The optimal ex ante disclosure is $\{0,\nu_\mathcal{F}\}$ if $\nu_0<\nu_\mathcal{F}$, and is non-disclosure otherwise. 
     Therefore, there is no information disclosure after the medical test.
 \end{prop}

 The ex ante signal $\{0,\nu_\mathcal{F}\}$ is comfort-oriented. The lower message fully reveals that $\bar p=q_L$, whereas the upper message pools $q_L$ with $q_H$ and induces posterior $\nu_\mathcal{F}$. No interim disclosure is needed. After the lower message, posterior zero leaves no information to disclose. After the upper message, $\nu_\mathcal{F}$ is the lowest belief at which non-disclosure induces testing and, because $\nu_\mathcal{F}>\nu_e$, also induces treatment after a sick result. When $\nu_0\geq\nu_\mathcal{F}$, the prior itself supports both decisions.

 \subsubsection*{Proof of Proposition \ref{prop:barp_interim}}

 By \citet[Theorem 3.2 and 3.3]{doval2024constrained}, there exists a Lagrangian multiplier $\lambda\geq 0$, such that the optimal signal $\tau\in\mathcal{R}(\nu_1)$ maximizes the Lagrangian $P(z)+\lambda V(z)$. 
 For any $\lambda\geq 0$, the function $P(z)+\lambda V(z)$ is constant on $[0,\nu_e)$, jumps upward at $\nu_e$, and is concave on $[\nu_e,1]$.
 Thus, by concavification, regardless of the exact value of $\lambda$, the optimal signal is comfort-oriented.
 When $\nu_1\in(\nu_\mathcal{M},\nu_\mathcal{F})$, the participation constraint must be binding, and therefore equation (\ref{eq.barp}) must hold.
 
 Because $\nu_\mathcal{M}<\nu_\mathcal{F}$, we have $\nu_\mathcal{F}<\nu_v$. To see this, suppose instead that $\nu_\mathcal{F}\geq\nu_v$. For every $\nu<\nu_v$, the patient-optimal signal characterized above gives the payoff
 \[
 \left(1-\frac{\nu}{\nu_v}\right)V(0)
 +\frac{\nu}{\nu_v}V(\nu_v)
 <
 V(\nu_v)
 \leq
 V(\nu_\mathcal{F})
 =
 \phi(\alpha+(1-\alpha)\underline p).
 \]
 For every $\nu\in[\nu_v,\nu_\mathcal{F})$, the patient-optimal policy is non-disclosure and gives
 \[
 V(\nu)<V(\nu_\mathcal{F})=\phi(\alpha+(1-\alpha)\underline p).
 \]
 Thus no belief below $\nu_\mathcal{F}$ can induce testing, which would imply $\nu_\mathcal{M}=\nu_\mathcal{F}$, a contradiction.

 For $\nu_1\in(\nu_\mathcal{M},\nu_\mathcal{F})$, equation~(\ref{eq.barp}) is equivalent to
 \[
 \frac{V(h(\nu_1))-V(0)}{h(\nu_1)}
 =
 \frac{\phi(\alpha+(1-\alpha)\underline p)-V(0)}{\nu_1}.
 \]
 At $\nu_\mathcal{M}$, the patient-optimal signal $\{0,\nu_v\}$ makes the participation constraint bind, while $V(\nu_\mathcal{F})=\phi(\alpha+(1-\alpha)\underline p)$. Hence the right-hand side lies strictly between
 \[
 \frac{V(\nu_\mathcal{F})-V(0)}{\nu_\mathcal{F}}
 \quad\text{and}\quad
 \frac{V(\nu_v)-V(0)}{\nu_v}.
 \]
 Strict concavity of $V$ makes the expression on the left strictly increasing on $[\nu_e,\nu_v]$. Hence the binding equation has a unique solution $h(\nu_1)\in(\nu_\mathcal{F},\nu_v)$ on this interval. If it has another solution above $\nu_v$, the doctor's payoff
 \[
 \left(1-\frac{\nu_1}{h(\nu_1)}\right)\underline p
 +\frac{\nu_1}{h(\nu_1)}q(h(\nu_1))
 =
 \underline p+\nu_1(q_H-q_L)
 +\frac{\nu_1(q_L-\underline p)}{h(\nu_1)}
 \]
 is strictly decreasing in $h(\nu_1)$ because $q_L>\underline p$. The smaller solution is therefore optimal. The same characterization gives $h(\nu_\mathcal{F})=\nu_\mathcal{F}$.

 \subsubsection*{Proof of Proposition \ref{proposition.transparency}}

 For $\nu_1<\nu_\mathcal{F}$, it suffices to show that $\bar{P}(\nu_1)$, the health probability under the optimal interim disclosure when there is no ex ante disclosure, is no larger than the health probability under ex ante signal $\{0,\nu_\mathcal{F}\}$, which is given by
     \begin{align*}
         \left(1-\frac{\nu_1}{\nu_\mathcal{F}}\right)\cdot\underline{p}+\frac{\nu_1}{\nu_\mathcal{F}}\cdot q(\nu_\mathcal{F}).
     \end{align*}
     If $\nu_1<\nu_\mathcal{M}$, testing cannot be induced, and thus $\bar P(\nu_1)=\underline p$ and the inequality holds. It remains to consider $\nu_1\in(\nu_\mathcal{M},\nu_\mathcal{F})$.
     Note that function $(1-\nu_1/x)\cdot\underline{p}+(\nu_1/x)\cdot q(x)$ is decreasing with $x$.
     Thus, since
     \begin{align*}
         \bar{P}(\nu_1)=\left(1-\frac{\nu_1}{h(\nu_1)}\right)\cdot\underline{p}+\frac{\nu_1}{h(\nu_1)}\cdot q(h(\nu_1)),
     \end{align*}
     it is sufficient to show that $h(\nu_1)\geq\nu_\mathcal{F}$. The first derivative for both sides of equation (\ref{eq.barp}) yields
     \begin{align*}
         h'(\nu_1)=\frac{h(\nu_1)}{\nu_1}\cdot\frac{V(h(\nu_1))-V(0)}{V(h(\nu_1))-V(0)-h(\nu_1)V'(h(\nu_1))}<0.
     \end{align*}
     The inequality follows because $h(\nu_1)<\nu_v$, and hence $[V(h)-V(0)]/h$ is increasing at $h=h(\nu_1)$ and therefore $V(h(\nu_1))-V(0)-h(\nu_1)V'(h(\nu_1))<0$.
     Thus, $\nu_\mathcal{F}=h(\nu_\mathcal{F})<h(\nu_1)$.
	     For $\nu_1\geq\nu_\mathcal{F}$, non-disclosure is feasible and attains $\bar P(\nu_1)=q(\nu_1)$. Combining the three regions, define
	     \begin{align*}
	         Q(\nu)\equiv
	         \begin{cases}
	         \displaystyle \underline p+\frac{\nu}{\nu_\mathcal{F}}\bigl[q(\nu_\mathcal{F})-\underline p\bigr],&\nu<\nu_\mathcal{F},\\[0.8em]
	         q(\nu),&\nu\geq\nu_\mathcal{F}.
	         \end{cases}
	     \end{align*}
	     The function $Q$ is concave because its left slope at $\nu_\mathcal{F}$ is $(q_H-q_L)+(q_L-\underline p)/\nu_\mathcal{F}$, which exceeds its right slope $q_H-q_L$. It dominates $\bar P$ everywhere, with equality at $0$, $\nu_\mathcal{F}$, and throughout $[\nu_\mathcal{F},1]$. Hence $Q=\operatorname{cav}\circ\bar P$. The signal $\{0,\nu_\mathcal{F}\}$ attains this envelope value when $\nu_0<\nu_\mathcal{F}$, and non-disclosure attains it otherwise, completing the proof.
\clearpage
\bibliographystyle{ecta}
\bibliography{reference}

\end{document}